\documentclass[acmsmall,screen]{acmart} 
\pdfoutput = 1 
\usepackage{etoolbox}
\newtoggle{arxiv}
\settoggle{arxiv}{true}

\AtBeginDocument{%
  }

\iftoggle{arxiv}{}{
\setcopyright{acmlicensed}
\acmJournal{PACMPL}
\acmYear{2025} \acmVolume{9} \acmNumber{OOPSLA1} \acmArticle{95} \acmMonth{4}\acmDOI{10.1145/3720429}
}

\iftoggle{arxiv}{}{
\acmConference[Conference acronym 'XX]{Make sure to enter the correct
  conference title from your rights confirmation emai}{June 03--05,
  2018}{Woodstock, NY}
\acmISBN{978-1-4503-XXXX-X/18/06}
}

\usepackage{xspace} 

\usepackage[disable]{todonotes}
\usepackage{tikz}
\usetikzlibrary{automata,positioning}
\usepackage{pgfplots}
\pgfmathsetseed{1234}

\definecolor{ourblue}{RGB}{0,84,159}
\definecolor{ourdarkblue}{RGB}{0,84,159}
\definecolor{ourlightblue}{RGB}{142,186,229}
\definecolor{ourlightestblue}{RGB}{232,241,250}
\definecolor{ouryellow}{RGB}{246,168,0}
\definecolor{ourred}{RGB}{161,16,53} 
\definecolor{ourgrey}{RGB}{207,209,210}

\usepackage{csquotes} 

\usepackage{enumerate}

\usepackage{mathtools} 

\usepackage{cleveref}
\Crefname{rem}{Remark}{Remarks}
\crefname{rem}{remark}{remarks}

\usepackage{adjustbox}

\usepackage{bm} 

\usepackage{thm-restate}

\usepackage{mdframed}

\usepackage{tensor}

\usepackage{subcaption}

\usepackage{nicefrac}

\usepackage{soul}

\usepackage{listings}
\usepackage{multirow}

\usepackage{siunitx} 

\newenvironment{myhighlight}{%
    \begin{center}%
        \begin{minipage}{0.85\textwidth}%
        }{%
        \end{minipage}%
    \end{center}%
}

\newcommand{\inductionCase}[1]{\subsubsection*{The case #1}} 

\newcommand{\gray}[1]{\textcolor{black!60}{#1}}

\newcommand{\toolfont}[1]{\textnormal{\textsc{#1}}}
\newcommand{\toolcaesar}{\toolfont{Caesar}\xspace}
\newcommand{\toolboogie}{\toolfont{Boogie}\xspace}
\newcommand{\toolzt}{\toolfont{Z3}\xspace}

\newcommand{\xqed}[1]{%
    \leavevmode\unskip\penalty9999 \hbox{}\nobreak\hfill
    \quad\hbox{#1}}
\newcommand{\qedDef}{\xqed{$\triangle$}}

\newcommand{\reals}{\mathbb{R}}
\newcommand{\extendedReals}{{\overline{\reals}}}
\newcommand{\exReals}{\extendedReals}
\newcommand{\realConst}{c}

\newcommand{\nonNegReals}{\reals_{\geq 0}}
\newcommand{\nnReals}{\nonNegReals}
\newcommand{\exNonNegReals}{\exReals_{\geq 0}}
\newcommand{\ennReals}{\exNonNegReals}
\newcommand{\monus}{\mathbin{\dot-}}

\newcommand{\nats}{\mathbb{N}}

\newcommand{\rats}{\mathbb{Q}}
\newcommand{\nonNegRats}{{\rats_{\geq 0}}}
\newcommand{\nnRats}{{\nonNegRats}}
\newcommand{\ratConst}{q}

\newcommand{\bools}{\mathbb{B}}
\newcommand{\boolConstTrue}{\mathsf{true}}
\newcommand{\boolConstFalse}{\mathsf{false}}

\newcommand{\PSPACE}{\mathsf{PSPACE}}
\newcommand{\QFNRA}{\mathsf{QF\_NRA}}
\newcommand{\coRE}{\mathsf{coRE}}
\newcommand{\RE}{\mathsf{RE}}

\newcommand{\intervalLeft}{a} 
\newcommand{\intervalLeftb}{c}
\newcommand{\ivalL}{\intervalLeft}
\newcommand{\ivalLb}{\intervalLeftb}
\newcommand{\intervalRight}{b} 
\newcommand{\intervalRightb}{d}
\newcommand{\ivalR}{\intervalRight}
\newcommand{\ivalRb}{\intervalRightb}
\newcommand{\closedInterval}[2]{{[#1,#2]}} 
\newcommand{\clIval}{\closedInterval}
\newcommand{\clIvalGen}{\closedInterval{\intervalLeft}{\intervalRight}}
\newcommand{\unitInterval}{\closedInterval{0}{1}}
\newcommand{\uIval}{{\unitInterval}}

\newcommand{\iv}[1]{{\left[#1\right]}} 

\newcommand{\lam}[2]{\lambda #1 . #2} 

\newcommand{\poleq}{\preceq}
\newcommand{\poleqb}{\preceqq}
\newcommand{\pogeq}{\succeq}

\newcommand{\partialOrderDomain}{L}
\newcommand{\partialOrderDomainb}{M}
\newcommand{\poDom}{\partialOrderDomain}
\newcommand{\poDomb}{\partialOrderDomainb}
\newcommand{\partialOrder}[2]{(#1,\,#2)}
\newcommand{\po}[2]{\partialOrder{#1}{#2}}
\newcommand{\partialOrderGen}{(\partialOrderDomain,\,\poleq)}
\newcommand{\partialOrderGenb}{(\partialOrderDomainb,\,\poleqb)}
\newcommand{\poGen}{\partialOrderGen}
\newcommand{\poGenb}{\partialOrderGenb}
\newcommand{\partialOrderElement}{a}
\newcommand{\partialOrderElementb}{b}
\newcommand{\poElem}{\partialOrderElement}
\newcommand{\poElemb}{\partialOrderElementb}
\newcommand{\partialOrderSubset}{A}
\newcommand{\poSubset}{\partialOrderSubset}
\newcommand{\poBot}{\bot}
\newcommand{\poTop}{\top}

\newcommand{\sigmaAlgebra}{\mathcal{F}}
\newcommand{\sAlg}{\sigmaAlgebra}
\newcommand{\proseSigmaAlgebra}{$\sigma$-algebra\xspace} 
\newcommand{\proseSigmaAlgebras}{$\sigma$-algebras\xspace} 
\newcommand{\measureSpaceUniverse}{\Omega}
\newcommand{\mUniv}{\measureSpaceUniverse}
\newcommand{\measurableSet}{A}
\newcommand{\measurableSetb}{B}
\newcommand{\countableCollectionOfMeasurableSets}{\mathcal{A}}
\newcommand{\sigmaAlgebraGenerators}{\mathcal{C}}
\newcommand{\sigmaAlgebraGeneratedBy}[1]{\sigma(#1)}

\newcommand{\borelSets}[1]{{\mathcal{B}(#1)}}
\newcommand{\lebesgueSets}[1]{{\overline{\mathcal{B}}(#1)}}
\newcommand{\simpleFun}{s}
\newcommand{\measure}{\boldsymbol\mu}
\newcommand{\lebesgueMeasure}{\boldsymbol\lambda}
\newcommand{\lebmes}{\lebesgueMeasure}

\newcommand{\slfp}[1]{\mathsf{lfp}\,{#1}}

\newcommand{\exlfp}[3]{\iv{\neg #3} \cdot #2 + \iv{#3} \cdot \wpTrans{#1}} 
\newcommand{\lfp}{\operatorname{lfp}}
\newcommand{\gfp}{\operatorname{gfp}}
\newcommand{\lfpIn}[2]{\lfp#1.\,#2}
\newcommand{\gfpIn}[2]{\gfp#1.\,#2}
\newcommand{\fpVar}{Y}

\newcommand{\realDomain}{D}
\newcommand{\fun}{f}
\newcommand{\funb}{g}

\newcommand{\boundedFunBound}{B}

\newcommand{\partition}{P}
\newcommand{\partitionb}{Q}

\newcommand{\partitions}[1]{\mathfrak{P}#1} 
\newcommand{\partitionSize}{N}
\newcommand{\partitionSizeb}{M}
\newcommand{\partitionNorm}[1]{\|#1\|} 
\newcommand{\partitionRefines}{\preceq}

\newcommand{\upperSum}[2]{{U_{#1,#2}}} 
\newcommand{\lowerSum}[2]{{L_{#1,#2}}} 
\newcommand{\upperInt}[2]{\overline{\int_{#1}^{#2}}} 
\newcommand{\lowerInt}[2]{\underline{\int_{#1}^{#2}}} 
\newcommand{\upperIntGen}{\upperInt{\ivalL}{\ivalR}}
\newcommand{\lowerIntGen}{\lowerInt{\ivalL}{\ivalR}}

\newcommand{\prog}{C} 
\newcommand{\progBody}{B} 
\newcommand{\pWhile}{\mathsf{pWhile}} 
\newcommand{\heyVL}{\mathsf{HeyVL}} 
\newcommand{\pWhileWith}[2]{\pWhile(#1,#2)} 
\newcommand{\pVars}{V} 
\newcommand{\pVar}{\mathtt{x}} 
\newcommand{\pVarb}{\mathtt{y}}
\newcommand{\pVarc}{\mathtt{z}}
\newcommand{\pVarm}{\mathtt{M}}
\newcommand{\pVari}{\mathtt{i}}
\newcommand{\pVarj}{\mathtt{j}}
\newcommand{\pVaraa}{\mathtt{a}}
\newcommand{\pVarbb}{\mathtt{b}}
\newcommand{\pVarcount}{\mathtt{count}}
\newcommand{\pVart}{\mathtt{t}}
\newcommand{\pVarh}{\mathtt{h}}
\newcommand{\pVarTmp}{\mathtt{tmp}} 
\newcommand{\pState}{\sigma} 
\newcommand{\pStates}{\nnReals^\pVars} 
\newcommand{\pSt}{\pState}
\renewcommand{\st}{\pState} 
\newcommand{\states}{\pStates}
\newcommand{\pStUpdate}[3]{#1[#2 \mapsto #3]}
\newcommand{\arithmeticExpression}{E}
\newcommand{\aExp}{\arithmeticExpression}
\newcommand{\aExps}{\mathcal{E}} 
\newcommand{\guard}{\varphi} 
\newcommand{\guards}{\mathcal{G}} 
\newcommand{\prob}{p} 
\newcommand{\unfold}[2]{#1^{#2}} 
\newcommand{\depth}{\ell} 

\newcommand{\SKIP}{\mathtt{skip}}
\newcommand{\DIVERGE}{\mathtt{diverge}}
\newcommand{\INVARIANT}{\mathtt{@invariant}}
\newcommand{\ASSIGN}[2]{#1\coloneq#2}
\newcommand{\NONDETASSIGN}[3]{#1\coloneq [ #2, #3 ]}
\newcommand{\UNIF}{\mathtt{unif}} 
\newcommand{\DISCRETEUNIF}{\mathtt{discrete\_unif}} 
\newcommand{\UNIFOVER}[2]{\UNIF\closedInterval{#1}{#2}} 
\newcommand{\RNDASSIGN}[2]{#1:\approx#2} 
\newcommand{\UNIFASSIGN}[1]{\RNDASSIGN{#1}{\UNIF_{\uIval}}} 
\newcommand{\DISCRETEUNIFASSIGN}[2]{\RNDASSIGN{#1}{\DISCRETEUNIF \left( #2 \right) }} 

\newcommand{\KWOBSERVE}{\mathtt{observe}} 
\newcommand{\OBSERVE}[1]{\KWOBSERVE(#1)}
\newcommand{\SEMICOLON}{;}
\newcommand{\SEQ}[2]{#1\,\SEMICOLON\,#2}
\newcommand{\KWIF}{\mathtt{if}} 
\newcommand{\KWELSE}{\mathtt{else}}
\newcommand{\ITE}[3]{\KWIF(#1)~\{#2\}~\KWELSE~\{#3\}}
\newcommand{\PCHOICE}[3]{\{#1\}~[#2]~\{#3\}}
\newcommand{\KWWHILE}{\mathtt{while}}
\newcommand{\WHILE}[2]{\KWWHILE~(#1)~\{#2\}}
\newcommand{\WHILENOBODY}[1]{\KWWHILE~(#1)~\{}
\newcommand{\INVARIANTANNOTATE}[1]{\INVARIANT \left( #1 \right)}

\newcommand{\pNEQ}{\neq}
\newcommand{\pAND}{\land}

\newcommand{\transSymb}{\mathcal{T}}
\newcommand{\wpSymb}{\mathsf{wp}}
\newcommand{\vcSymb}{\mathsf{vc}}
\newcommand{\wlpSymb}{\mathsf{wlp}}
\newcommand{\cwpSymb}{\mathsf{cwp}}
\newcommand{\lwpSymb}[1]{\underline{\mathsf{w}}\mathsf{p}^{#1}} 
\newcommand{\uwpSymb}[1]{\overline{\mathsf{wp}}^{#1}} 
\newcommand{\lwlpSymb}[1]{\underline{\mathsf{wl}}\mathsf{p}^{#1}} 
\newcommand{\uwlpSymb}[1]{\overline{\mathsf{w}}\mathsf{l}\overline{\mathsf{p}}^{#1}} 
\newcommand{\lcwpSymb}[1]{\underline{\mathsf{cw}}\mathsf{p}^N} 
\newcommand{\ucwpSymb}[1]{\overline{\mathsf{cwp}}^N} 
\newcommand{\wpwlpSymb}{\mathsf{w}(\mathsf{l})\mathsf{p}}

\newcommand{\someTrans}[1]{\transSymb\llbracket#1\rrbracket}
\newcommand{\wpTrans}[1]{\wpSymb\llbracket#1\rrbracket} 
\newcommand{\wpwlpTrans}[1]{\wpwlpSymb\llbracket#1\rrbracket}
\newcommand{\vcTrans}[1]{\vcSymb\llbracket#1\rrbracket} 
\newcommand{\wlpTrans}[1]{\wlpSymb\llbracket#1\rrbracket}
\newcommand{\cwpTrans}[1]{\cwpSymb\llbracket#1\rrbracket}
\newcommand{\lwpTrans}[2]{\lwpSymb{#1}\llbracket#2\rrbracket}
\newcommand{\uwpTrans}[2]{\uwpSymb{#1}\llbracket#2\rrbracket}
\newcommand{\lwlpTrans}[2]{\lwlpSymb{#1}\llbracket#2\rrbracket}
\newcommand{\uwlpTrans}[2]{\uwlpSymb{#1}\llbracket#2\rrbracket}

\newcommand{\somewp}[2]{\someTrans{#1}\left(#2\right)} 
\renewcommand{\wp}[2]{\wpTrans{#1}\left(#2\right)} 
\newcommand{\vc}[2]{\vcTrans{#1}\left(#2\right)} 
\newcommand{\wlp}[2]{\wlpTrans{#1}\left(#2\right)}
\newcommand{\cwp}[2]{\cwpTrans{#1}\left(#2\right)}
\newcommand{\lwp}[3]{\lwpTrans{#1}{#2}\left(#3\right)} 
\newcommand{\uwp}[3]{\uwpTrans{#1}{#2}\left(#3\right)}
\newcommand{\lwlp}[3]{\lwlpTrans{#1}{#2}\left(#3\right)} 
\newcommand{\uwlp}[3]{\uwlpTrans{#1}{#2}\left(#3\right)}

\newcommand{\charfun}[3]{\tensor*[^{#1}_{#2}]{\Phi}{_{{#3}}}}

\newcommand{\charfuntrans}[2]{\charfun{\transSymb}{#1}{#2}}
\newcommand{\charfunwp}[2]{\charfun{\wpSymb}{#1}{#2}}
\newcommand{\charfunwlp}[2]{\charfun{\wlpSymb}{#1}{#2}}
\newcommand{\charfunuwp}[3]{\charfun{\uwpSymb{#1}}{#2}{#3}}
\newcommand{\charfunlwp}[3]{\charfun{\lwpSymb{#1}}{#2}{#3}}
\newcommand{\charfunlwlp}[3]{\charfun{\lwlpSymb{#1}}{#2}{#3}}

\newcommand{\N}{N} 
\newcommand{\Nb}{M} 

\newcommand{\ex}{f} 
\newcommand{\exb}{g}

\newcommand{\exI}{I}
\newcommand{\exJ}{J}
\newcommand{\eleq}{\sqsubseteq} 
\newcommand{\egeq}{\sqsupseteq} 
\newcommand{\expectations}{\mathbb{E}}
\newcommand{\exps}{\expectations}
\newcommand{\boundedExpectations}{{\expectations^{\leq 1}}}
\newcommand{\bexps}{\boundedExpectations}
\newcommand{\measurableExpectations}{{\expectations_{meas}}}
\newcommand{\expsmeas}{\measurableExpectations}
\newcommand{\boundedMeasurableExpectations}{{\expectations^{\leq 1}_{meas}}}
\newcommand{\bexpsmeas}{\boundedMeasurableExpectations}
\newcommand{\measurableLocallyBoundedExpectations}{{\expectations_{meas}^{\lesssim}}}
\newcommand{\expsmeaslb}{\measurableLocallyBoundedExpectations}
\newcommand{\compactSet}{K}
\newcommand{\expsClass}{\mathcal{F}} 
\newcommand{\exSubs}[3]{#1[#2 / #3]} 
\newcommand{\exSubsGen}{\exSubs{\ex}{\pVar}{\aExp}}

\newcommand{\synEx}{\mathsf{f}}
\newcommand{\synExb}{\mathsf{g}}
\newcommand{\synExc}{\mathsf{h}}
\newcommand{\synExI}{\mathsf{I}}
\newcommand{\synExJ}{\mathsf{J}}
\newcommand{\synExps}{\mathsf{Expr}}
\newcommand{\synSubs}[3]{#1[#2/#3]}
\newcommand{\supQuantifierSymbol}{\reflectbox{\textnormal{\textsf{\fontfamily{phv}\selectfont S}}}\hspace{.2ex}}
\newcommand{\infQuantifierSymbol}{\raisebox{.6\depth}{\rotatebox{-30}{\textnormal{\textsf{\fontfamily{phv}\selectfont \reflectbox{J}}}}\hspace{-.1ex}}}
\newcommand{\quantitativeQuantifierSymbol}{\reflectbox{\textnormal{\textsf{\fontfamily{phv}\selectfont Q}}}\hspace{.2ex}}
\newcommand{\synSupBd}[3]{{\supQuantifierSymbol\!_\clIval{#2}{#3}\,#1\colon}}
\newcommand{\synInfBd}[3]{{\infQuantifierSymbol\!_\clIval{#2}{#3} \,#1\colon}}
\newcommand{\synVarQuantBd}[4]{{#4\!_\clIval{#2}{#3} \,#1\colon}}

\newcommand{\limplies}{{}\rightarrow{}}
\newcommand{\liff}{{}\leftrightarrow{}}
\newcommand{\sem}[1]{\llbracket#1\rrbracket}
\newcommand{\free}[1]{{\mathsf{free}(#1)}}
\newcommand{\varsIn}[1]{{\mathsf{free}(#1)}}

\newcommand{\logicVariables}{\mathit{LV}}
\newcommand{\lvars}{\logicVariables}
\newcommand{\logicVariablesSubset}{{\logicVariables'}}
\newcommand{\lvarsSubset}{\logicVariablesSubset}
\newcommand{\logicVariable}{x}
\newcommand{\lvar}{\logicVariable}
\newcommand{\logicVariableb}{y}
\newcommand{\lvarb}{\logicVariableb}
\newcommand{\logicVariablec}{z}
\newcommand{\lvarc}{\logicVariablec}
\newcommand{\logicVariabled}{w}
\newcommand{\lvard}{\logicVariabled}

\newcommand{\syntacticTerm}{\mathsf{t}}
\newcommand{\syntacticTermb}{\mathsf{u}}
\newcommand{\syntacticTermc}{\mathsf{v}}
\newcommand{\synTerm}{\syntacticTerm}
\newcommand{\synTermb}{\syntacticTermb}
\newcommand{\synTermc}{\syntacticTermc}
\newcommand{\syntacticTerms}{\mathsf{Terms}}
\newcommand{\synTerms}{\syntacticTerms}
\newcommand{\synMonus}{\mathbin{\dot{-}}} 

\newcommand{\syntacticGuard}{\upvarphi}
\newcommand{\synGuard}{\syntacticGuard}
\newcommand{\syntacticGuards}{\mathsf{Guards}}
\newcommand{\synGuards}{\syntacticGuards}

\newcommand{\firstOrderFormula}{\uppsi}
\newcommand{\foForm}{\firstOrderFormula}
\newcommand{\firstOrderFormulab}{\uptheta}
\newcommand{\foFormb}{\firstOrderFormulab}
\newcommand{\firstOrderFormulas}{\mathsf{FO}}
\newcommand{\foForms}{\firstOrderFormulas}
\newcommand{\firstOrderExistsQuantifier}{\exists}
\newcommand{\foExists}{\firstOrderExistsQuantifier}
\newcommand{\firstOrderForallQuantifier}{\forall}
\newcommand{\foForall}{\firstOrderForallQuantifier}
\newcommand{\foQuantifierSymbol}{\reflectbox{$\mathsf{Q}$}\hspace{.2ex}}

\newcommand{\resVal}{r}

\newcommand{\semiAlgebraicSet}{A}
\newcommand{\saSet}{\semiAlgebraicSet}

\newcommand{\expsClassExp}{\exps_{\synExps}}

\newcommand{\representableExpectations}{\exps_{\synExps}}
\newcommand{\repExps}{\representableExpectations}
\newcommand{\boundedRepresentableExpectations}{\exps^{\leq 1}_{\synExps}}
\newcommand{\brepExps}{\boundedRepresentableExpectations}

\newcommand{\syntacticExpectations}{\synExps}
\newcommand{\bsyntacticExpectations}{\syntacticExpectations^{\leq 1}}

\newcommand{\eps}{\varepsilon}

\newcommand{\zerofun}{0}
\newcommand{\onefun}{1}
\newcommand{\grammarSymb}{\Coloneqq}

\newcommand{\eeq}{~{}={}~}

\newcommand{\lleq}{~{}\leq{}~}
\newcommand{\ggeq}{~{}\geq{}~}
\newcommand{\llt}{~{}<{}~}

\newcommand{\eeleq}{~{}\eleq{}~}

\newcommand{\mminus}{~{}-{}~}
\newcommand{\pplus}{~{}+{}~}
\newcommand{\mmid}{~{}\mid{}~}

\newcommand{\qand}{\quad\text{and}\quad}
\newcommand{\qqand}{\qquad\text{and}\qquad}

\newcommand{\colprog}{\textbf{Program}}
\newcommand{\coln}{$\mathbf{\N}$}
\newcommand{\colt}{\textbf{Time}}
\newcommand{\colformsize}{\textbf{|AST|}}
\newcommand{\colpost}{\textbf{Post}}
\newcommand{\colbound}{\textbf{Bound}}

\newcommand{\progdiverging}{\texttt{diverging}}
\newcommand{\progirwin}{\texttt{IrwinHall}}
\newcommand{\progmontecarlo}{\texttt{MonteCarlo}}
\newcommand{\progtortoise}{\texttt{TortoiseHare}}

\iftoggle{arxiv}{
    \setcopyright{none}
    \settopmatter{printacmref=false} 
    \renewcommand\footnotetextcopyrightpermission[1]{} 
    \pagestyle{plain}
}{}

\theoremstyle{remark}
\newtheorem{rem}{Remark}

\iftoggle{arxiv}{}{
\setcopyright{acmlicensed}
\acmDOI{10.1145/3720429}
\acmYear{2025}
\acmJournal{PACMPL}
\acmVolume{9}
\acmNumber{OOPSLA1}
\acmArticle{95}
\acmMonth{4}
\received{2024-10-16}
\received[accepted]{2025-02-18}
}

\begin{document}

\title{Foundations for Deductive Verification of Continuous Probabilistic Programs}
\subtitle{From Lebesgue to Riemann and Back}

\author{Kevin Batz}
\orcid{0000-0001-8705-2564}
\affiliation{%
    \institution{RWTH Aachen University}
    \city{Aachen}
    \country{Germany}
}
\affiliation{%
    \institution{University College London}
    \city{London}
    \country{United Kingdom}
}
\email{kevin.batz@cs.rwth-aachen.de}

\author{Joost-Pieter Katoen}
\orcid{0000-0002-6143-1926}
\affiliation{%
    \institution{RWTH Aachen University}
    \city{Aachen}
    \country{Germany}
}
\email{katoen@cs.rwth-aachen.de}

\author{Francesca Randone}
\orcid{0009-0002-3489-9600}
\affiliation{%
    \institution{University of Trieste}
    \city{Trieste}
    \country{Italy}
}
\email{francesca.randone@units.it}

\author{Tobias Winkler}
\orcid{0000-0003-1084-6408}
\affiliation{%
    \institution{RWTH Aachen University}
    \city{Aachen}
    \country{Germany}
}
\email{tobias.winkler@cs.rwth-aachen.de}

\iftoggle{arxiv}{\titlenote{This document is the full version of a paper with the same title published at OOPSLA 2025.}}{}


\begin{abstract}
    We lay out novel foundations for the computer-aided verification of guaranteed bounds on expected outcomes of imperative probabilistic programs featuring (i) general \emph{loops}, (ii) \emph{continuous} distributions, and (iii) \emph{conditioning}.
To handle loops we rely on user-provided quantitative \emph{invariants}, as is well established.
However, in the realm of continuous distributions, \emph{invariant verification} becomes extremely challenging due to the presence of \emph{integrals} in expectation-based program semantics.
Our key idea is to soundly \emph{under-} or \emph{over-approximate} these integrals via \emph{Riemann sums}.
We show that this approach enables the SMT-based invariant verification for programs with a fairly general control flow structure.
On the theoretical side, we prove \emph{convergence} of our Riemann approximations, and establish $\coRE$-completeness of the central verification problems.
On the practical side, we show that our approach enables to use existing automated verifiers targeting \emph{discrete} probabilistic programs for the verification of programs involving \emph{continuous sampling}. Towards this end,
we implement our approach in the recent quantitative verification infrastructure \toolcaesar by encoding Riemann sums in its intermediate verification language.
We present several promising case studies.

\end{abstract}



\maketitle


\section{Introduction}
\label{sec:intro}

\newcommand*{\diffeo}{%
    \mathrel{\vcenter{\offinterlineskip
            \hbox{$\sim$}\vskip-.35ex\hbox{$\sim$}\vskip-.35ex\hbox{$\sim$}}}}
        
\newcommand{\progPiInner}{\textcolor{ourdarkblue}{\prog_{inner}}}
\newcommand{\progPi}{\prog}

Probabilistic programs (PP) are usual programs with the additional ability to
\begin{enumerate}
    \item\label{ability:branch} \emph{branch} their flow of control based on the outcomes of coin flips,
    \item\label{ability:sample} \emph{sample} numbers from predefined continuous and/or discrete probability distributions, and
    \item\label{ability:condition} \emph{condition} their current state on observations using dedicated $\KWOBSERVE$-instructions.
\end{enumerate}
Programs with one or more of these (partly overlapping) abilities have been extensively studied in different contexts and throughout various communities:
\emph{Randomized algorithms} typically rely on ability \eqref{ability:branch} to speed up computations on average~\cite{DBLP:journals/dam/Karp91} or for breaking symmetries in distributed computing~\cite{DBLP:journals/iandc/ItaiR90}.
\emph{Sampling} and \emph{conditioning} (abilities \eqref{ability:sample} and \eqref{ability:condition}) are extensively employed in connection with Bayes' rule in machine learning, AI, and cognitive science to statistically infer information from observed data~\cite{dippl,DBLP:conf/icse/GordonHNR14,DBLP:journals/corr/abs-1809-10756}.
\emph{Continuous distributions} are pivotal for these applications.

Many applications require \emph{expressive} PP with the usual constructs from traditional programs --- including unbounded $\KWWHILE$-loops.
The latter pose a major challenge:
For instance, reasoning about probabilistic termination is provably harder than reasoning about classical termination~\cite{DBLP:conf/lics/KaminskiK17}.

\paragraph{Guaranteed Bounds on Expected Outcomes via Weakest Pre-Expectations}

In recent years, there has been considerable interest in analyzing expressive PP in a mathematically rigorous manner --- with \emph{hard} rather than statistical guarantees --- see e.g., \cite{DBLP:conf/cav/ChakarovS13,DBLP:conf/cav/GehrMV16,DBLP:journals/pacmpl/AgrawalC018,DBLP:journals/pacmpl/MoosbruggerSBK22,DBLP:conf/pldi/GehrSV20} as representative examples.
Such stringent analyses are motivated by the fact that more conventional statistical techniques can be far off even for simple programs~\cite{DBLP:conf/pldi/BeutnerOZ22}, and may be inappropriate for safety-critical applications.

The \emph{weakest pre-expectation ($\wpSymb$) calculus} is a prominent means to specify and establish a broad spectrum of PP properties in a rigorous manner~\cite{DBLP:conf/stoc/Kozen83,DBLP:series/mcs/McIverM05,DBLP:phd/dnb/Kaminski19}.
Generalizing Dijkstra's classical weakest pre-\emph{conditions}, the central objects are so-called \emph{expectations} $\ex$ --- random variables over a program's state space mapping program states to numbers --- which take over the role of predicates from classical program verification.
More precisely, given a probabilistic program $\prog$, an expectation $\ex$, and an initial program state $\pState$, we have
\[
	\wp{\prog}{\ex}(\pState)
	\quad=\quad 
	\substack{\text{\normalsize\emph{expected value} of $\ex$ w.r.t.\ the distribution of \emph{final} states} \\ \text{\normalsize reached after executing $\prog$ on \emph{initial} state $\pState$} \\
	\text{\normalsize and aborting all executions violating an observation}~.}
\]
Hence, $\wp{\prog}{\ex}$ is a map from \emph{initial} program states to \emph{expected final} values of $\ex$.
The weakest \emph{liberal} pre-expectation $\wlp{\prog}{\ex}$ adds to the above quantity the probability of $\prog$'s divergence\footnote{Provided $\ex$ is upper-bounded by $1$.
See \Cref{sec:programs:wp} for details.}.
Both $\wpSymb$ and $\wlpSymb$ can be defined by induction on the structure of $\prog$.
Notice that neither of these calculi yield \emph{conditional} expected values as one might expect in the presence of conditioning.
One can, however, define \emph{conditional weakest pre-expectations}~\cite{DBLP:journals/toplas/OlmedoGJKKM18,DBLP:conf/aaai/NoriHRS14} as\footnote{Notice that this quantity is undefined if $\wlp{\prog}{1}(\pState)=0$ as is natural when conditioning on a probability-$0$-event.} 
\[
	\cwp{\prog}{\ex}(\pState)
	\eeq
	\frac{\wp{\prog}{\ex}(\pState)}
	{\wlp{\prog}{1}(\pState)}
	\eeq
	\frac{\text{\normalsize{expected value} of $\ex$ \ldots}}
	{\text{probability of not violating an observation}}~,
\]
which indeed yields the sought-after conditional expected values.
We refer to quantities such as $\wp{\prog}{\ex}(\pState)$, $\wlp{\prog}{\ex}(\pState)$, and $\cwp{\prog}{\ex}(\pState)$ as \emph{(conditional) expected outcomes} of probabilistic programs.
Reasoning about expected outcomes of loops is typically tackled by means of \emph{quantitative loop invariants}, which are, naturally, often hard to find.
However, in the presence of continuous sampling, even \emph{verifying} a \emph{given} candidate loop invariant poses severe challenges as reasoning about the required expected values involves possibly complex \emph{integrals}.
The aim of this paper is to lay the foundations for automated techniques tackling these challenges.

\paragraph{Problem Statement}

We study \emph{imperative} PP with all three abilities \eqref{ability:branch}, \eqref{ability:sample}, and \eqref{ability:condition}.
Our goal is to
\begin{myhighlight}
    \emph{semi-automatically verify bounds on (conditional) expected outcomes of programs featuring \underline{continuous uni}f\underline{orm sam}p\underline{lin}g, unbounded \underline{$\KWWHILE$-loo}p\underline{s} with user-provided q\underline{uantitative invariants}, and \underline{conditionin}g.}    
\end{myhighlight}
\enquote{Semi-automatic} means that we focus on verification of \emph{user-provided} quantitative loop invariants in the sense of~\cite{DBLP:series/mcs/McIverM05}.
Unlike fully automatic approaches we do not \emph{synthesize} such invariants automatically, which is a promising direction for future work but outside the scope of this paper.

\paragraph{Approach}

Our key idea is simple, yet powerful:
\begin{myhighlight}
    \emph{We replace the \underline{inte}g\underline{rals} occurring in the definition of the programs' exact weakest pre-expectations by simpler \underline{sound a}pp\underline{roximations}, namely lower and upper \underline{Riemann sums}}.
\end{myhighlight}
This results in a family of approximate \emph{lower} and \emph{upper Riemann $\wpSymb$ transformers}, denoted by $\lwpSymb{\N}$ and $\uwpSymb{\N}$, where $\N$ is the number of sub-intervals used to discretize the domain of integration. 
Crucially, these Riemann expectation transformers ultimately give rise to automated SMT-based techniques for verifying quantitative loop invariants.
Let us now consider an introductory example.

\paragraph{Illustrative Example: A Monte Carlo Approximator.}

The program shown in \Cref{fig:intro} draws $\pVarm$ $(x,y)$-samples uniformly at random from the unit square.
Each time a sample lands in the quarter unit circle, the variable $\pVarcount$ is incremented.
The program hence approximates the transcendental number $\tfrac \pi 4 \approx 0.785$ (the area of the quarter unit circle) in a Monte Carlo manner.

We will first focus on the \textcolor{ourdarkblue}{blue} fragment $\progPiInner$ of the program from \Cref{fig:intro}.
$\progPiInner$ contains two uniform continuous sampling instructions.
Suppose we aim to compute the expected value of $\pVarcount$ after executing $\progPiInner$.
Applying the rules for determining $\wpSymb$'s~\cite{DBLP:conf/setss/SzymczakK19}, we obtain the following double Lebesgue integral, where $\iv{\ldots}$ denotes the indicator function of the enclosed predicate:
\begin{align*}
    \wp{\progPiInner}{\pVarcount}
    \eeq
    \underbrace{\pVarcount + \int_0^1 \int_0^1 \iv{\pVar^2 + \pVarb^2 \leq 1} \,d\lebmes(x), d\lebmes(y)
    \eeq
    \pVarcount + \frac{\pi}{4}}_{\text{expected final value of $\pVarcount$ in terms of its initial value}}
    ~.
    \tag{$\dagger$}\label{eq:doubleIntegral}
\end{align*}
This reflects the fact that, in \emph{expectation}, $\progPiInner$ increments $\pVarcount$ by $\tfrac \pi 4$.

Symbolic integration-based tools such as PSI~\cite{DBLP:conf/cav/GehrMV16} can solve the above integral directly.
However, we observe two significant issues:
(i) there are  integrals that cannot be evaluated symbolically in closed form,
and (ii) even if all guards and assignments in the program are polynomial and distributions are restricted to uniform distributions in $[0,1]$, the resulting integrals may evaluate to transcendental numbers, making automation notoriously difficult.

As outlined, instead of solving integrals exactly, our strategy is to soundly \emph{under-} or \emph{over-approximate} them by lower or upper Riemann sums.
In our example, the resulting upper sum
\begin{align*} 
    \frac{1}{\N^2} \sum_{i=0}^{N-1} \sum_{j=0}^{N-1} ~\sup_{\xi \in [\frac{i}{\N}, \frac{i+1}{\N}]} ~ \sup_{\zeta \in [\frac{j}{\N}, \frac{j+1}{\N}]} ~ \iv{\xi^2 + \zeta^2 \le 1} ~,
\end{align*}
where $\N \geq 1$, is a guaranteed upper bound on the double integral in \eqref{eq:doubleIntegral}.
For instance, with $\N = 16$, we can verify that the upper sum is at most $0.85$, and hence that $\wp{\progPiInner}{\pVarcount}(\pState) \leq \pState(\pVarcount) + 0.85$ for all initial states $\pState$.
Crucially, to prove upper bounds on upper Riemann sums, we do \emph{not} have to evaluate any suprema explicitly because for all $A \subseteq \reals$ and $b \in \reals$, 
\[
    \sup A \leq b \qquad\text{iff}\qquad \text{$\forall a \in A$}\colon  a \leq b~.
\]
In practice, we can thus drop the $\sup$'s in upper Riemann sums (and, dually, the $\inf$'s in lower sums) by introducing $\foForall$-quantifiers, which is highly beneficial in the context of SMT-based automation.

We stress that even though we discretize the integrals' domains uniformly, our approach is \emph{not} the same as replacing the continuous $\UNIF_{\uIval}$ distributions by discrete uniform distributions of $N$ point-masses.
Indeed, the latter would neither yield under- nor over-approximations in general, which is, however, essential for invariant verification. 

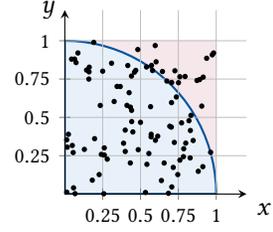
\begin{figure}[t]
    \begin{minipage}{0.5\textwidth}
        \centering
        \begin{align*}
            & \ASSIGN{\pVari}{1} \,\SEMICOLON\, \ASSIGN{\pVarcount}{0} \,\SEMICOLON\, \\
            &\WHILENOBODY{\pVari \leq \pVarm} \\
            &\qquad \textcolor{ourdarkblue}{\SEQ{\UNIFASSIGN{\pVar}}{\UNIFASSIGN{\pVarb}}\,\SEMICOLON} \\
            &\qquad \textcolor{ourdarkblue}{\ITE{\pVar^2 + \pVarb^2 \leq 1}{\ASSIGN{\pVarcount}{\pVarcount + 1}}{\SKIP}\,\SEMICOLON} \\
            &\qquad \ASSIGN{\pVari}{\pVari + 1} \quad \}
        \end{align*}
    \end{minipage}
    \hspace{5mm}
    \begin{minipage}{0.3\textwidth}
        \centering
        \begin{tikzpicture}[scale=1.0]
    \begin{axis}[
        axis lines=middle,
        xtick={0,.25,.50,.75,1},
        ytick={0,.25,.50,.75,1},
        tick label style={font=\footnotesize},
        xlabel={$x$},
        ylabel={$y$},
        xlabel style={at={(axis description cs:1.1,0)}, anchor=north},
        ylabel style={at={(axis description cs:0,1.0)}, anchor=east},
        xmin=0, xmax=1.2,
        ymin=0, ymax=1.2,
        width=4cm, 
        height=4cm, 
        axis equal,
        domain=0:1,
        scale=1.0,
        grid=both, 
        major grid style={line width=.2pt,draw=gray!50},
        minor grid style={line width=.1pt,draw=gray!20},
        axis on top
        ]
        \addplot[draw=ourred!10,fill=ourred!10] coordinates {(0,0) (0,1) (1,1) (1,0) (0,0)};
        \addplot[thick, draw=ourblue, fill=ourlightblue!20, domain=0:90] ({cos(x)},{sin(x)}) -- (axis cs:0,0) -- cycle;
    \end{axis}
    \foreach \i in {1,...,100} {
        \pgfmathsetmacro{\x}{(rand+1)}
        \pgfmathsetmacro{\y}{(rand+1)}
        \fill (\x, \y) circle (1pt);
    }
\end{tikzpicture}
    \end{minipage}    
    \caption{
        Left: Monte-Carlo approximator $\progPi$ for $\pi$ using ability \eqref{ability:sample}.
        Right: 100 random samples.
    }
    \label{fig:intro}
\end{figure}

\paragraph{Loops and Invariants}

Now, consider the entire program $\progPi$ in \Cref{fig:intro}.
We have $\wp{\progPi}{\pVarcount} = \tfrac \pi 4 \cdot \pVarm$ because, intuitively, the percentage of samples landing in the quarter circle equals its area.
Our techniques enable to soundly bound $\wp{\progPi}{\pVarcount}$ in a semi-automated manner \emph{for arbitrary initial values of~$\pVarm$}.
This is done, as is standard in deductive verification, using quantitative loop invariants, which we detail in \Cref{sec:invariants_and_unrolling}.
Suffices to say here that, given a suitable invariant $\exI$, our techniques enable to verify its validity automatically.
This in turn allows us to conclude that:
\[
	\forall\,\text{initial states $\pState$}\colon \quad 
	\wp{\progPi}{\pVarcount}(\pState) \lleq 0.85 \cdot \pState(\pVarm)~.
\]
See \Cref{ex:invariant_monte_carlo} (\cpageref{ex:invariant_monte_carlo}) and \Cref{ex:caesar_loop} (\cpageref{ex:caesar_loop}) for details.

\paragraph{Contributions}

In summary, this paper introduces the novel \emph{lower} and \emph{upper Riemann weakest pre-expectation transformers} and demonstrates their applicability to the semi-automated verification of PP with \emph{continuous uniform sampling} and general \emph{$\KWWHILE$-loops}.
Put more concretely:
\begin{description}
    \item[Verification of Invariants] 
    We show that externally provided loop invariant candidates can be verified using the Riemann $\wpSymb$ transformers in a way that is suitable for SMT-based automation.
    Such invariants, once verified, imply (one-sided) bounds on $\wpSymb$, $\wlpSymb$, and $\cwpSymb$.
    \item[Convergence and Complexity]
    On the more theoretical side, our approximate Riemann $\wpSymb$'s are shown to \emph{converge} to the exact $\wpSymb$'s under mild assumptions as the discretization becomes finer.
    As a consequence, we obtain $\coRE$-completeness\footnote{This means that the refutation of such bounds is semi-decidable.} results for the problems of verifying upper bounds on $\wpSymb$ and lower bounds on $\wlpSymb$ for fairly general loopy programs.
    \item[Implementation and Case Studies]
    We incorporate our method in the PP verification infrastructure \toolcaesar~\cite{DBLP:journals/pacmpl/SchroerBKKM23} by encoding Riemann sums in its \emph{intermediate verification language}.
    This enables transferring principles from classical SMT-based program verification \cite{DBLP:conf/vmcai/0001SS16}, such as custom first-order theories for reasoning about exponentials, to the verification of PPs involving continuous sampling. 
    We provide various case studies and an empirical evaluation.
\end{description}

\paragraph{Paper Structure}

\Cref{sec:birdseye} provides a bird's eye view on the distinctive features and technicalities of our method.
\Cref{sec:prelims} introduces basic fixed point theory and \Cref{sec:programs} summarizes the existing definitions of $\wpSymb$, $\wlpSymb$, and $\cwpSymb$.
In \Cref{sec:approxwp}, we introduce our Riemann $\wpSymb$ transformers and prove them sound.
\Cref{sec:invariants_and_unrolling} is devoted to loop invariants.
The convergence results are presented in \Cref{sec:convergence}.
Towards automation, \Cref{sec:syntax} introduces a concrete, effective syntax for expectations.
\Cref{sec:case_studies} details the implementation in \toolcaesar and presents case studies.
Additional related work is surveyed in \Cref{sec:relwork}.
We conclude in \Cref{sec:conclusion}.
\iftoggle{arxiv}{Omitted proofs can be found in \Cref{app:proofs}.}{An extended version of this article containing the ommitted proofs and additional material is available~\cite{arxiv}.}

\section{A Bird's Eye View}
\label{sec:birdseye}

In this section, we explain what makes our method unique by providing a compressed overview of its characteristic features.
We also discuss its limitations and point out a subtle technical challenge.

\subsection{Highlights and Comparison to Other Approaches}

Since we are not the first to address the problem of proving bounds on quantities expressible as $\wpwlpSymb$'s, we now highlight five distinctive features of our approach --- their combination is unique in the literature.
A more in-depth review of related work is deferred to \Cref{sec:relwork}.

Of the following five items, the first two are specific to the way we handle integrals via Riemann sums, whereas the latter three tend to apply to $\wpSymb$-based approaches in general.

\paragraph{An Alternative to Moment-Based Analyses}

A major thread of existing work, e.g.,~\cite{DBLP:conf/sas/ChakarovS14,DBLP:journals/pacmpl/AgrawalC018,DBLP:journals/pacmpl/MoosbruggerSBK22}, circumvents explicit integration altogether.
They achieve this by focusing on classes of programs and properties where, simply speaking, the analysis can soundly replace the distributions in the program by their \emph{mean}~\cite[Appendix A, p.~31]{DBLP:journals/corr/abs-1709-04037} or some higher moment~\cite{DBLP:journals/pacmpl/MoosbruggerSBK22}.
For example, $\wp{\UNIFASSIGN{\pVar}}{\pVar \pVarb}$ is equal to\footnote{In general, replacing $\UNIFASSIGN{\pVar}$ by $\ASSIGN{\pVar}{\tfrac 1 2}$ is sound whenever the post-expectation is linear in $\pVar$. While our approach is compatible with such optimizations, we shall not discuss them any further.} $\wp{\ASSIGN{\pVar}{\tfrac 1 2}}{\pVar \pVarb} = \tfrac 1 2 y$, where $\tfrac 1 2$ is the mean of $\UNIFOVER{0}{1}$.
However, such mean-based analyses do not always suffice.
Indeed, reconsidering the example from \Cref{sec:intro}, the inequality $\wp{\progPiInner}{\pVarcount} \eleq \pVarcount + 0.85$, proved correct with our approach, becomes false if we substitute $\UNIFASSIGN{\pVar}$ and $\UNIFASSIGN{\pVarb}$ in $\progPiInner$ by $\ASSIGN{\pVar}{\tfrac 1 2}$ and $\ASSIGN{\pVarb}{\tfrac 1 2}$, respectively.
In fact, the so-obtained program $\progPiInner'$ satisfies $\wp{\progPiInner'}{\pVarcount} = \pVarcount + 1$ as the point $(\tfrac 1 2, \tfrac 1 2)$ \emph{is certainly inside} the quarter unit circle.

\paragraph{General Conditional Branching}

We make relatively mild assumptions regarding the shape of $\KWIF$- and $\KWWHILE$-guards in the programs --- in practice, general Boolean combinations of the polynomial (in)equalities over all program variables are supported, as in \Cref{fig:intro}.
This is  different from, e.g.,~\cite{DBLP:journals/pacmpl/MoosbruggerSBK22} which restricts to finitely-valued variables in guards or~\cite{DBLP:conf/sas/ChakarovS14,DBLP:conf/popl/ChatterjeeFNH16} which restrict to linear guards.
We achieve this high level of generality thanks to the simplicity of the Riemann approximation which treats integration over the resulting indicator functions in a uniform manner.

\paragraph{Moving Backwards: Verification and Parametric Models}

We conduct a \emph{backward analysis}: 
Given a random variable $\ex$ over final program states (called \emph{post-expectation} in this paper), we approximate the (conditional) weakest pre-expectation $\wp{\prog}{\ex}$, which is \emph{a function of the initial state}.
This is dual to works like~\cite{DBLP:conf/cav/GehrMV16,DBLP:conf/pldi/BeutnerOZ22,wang2024static} that conduct a \emph{forward analysis} to solve a classic Bayesian inference problem, namely to compute a program's posterior density --- \emph{a function of the final state} --- for a given prior distribution.
Let us contrast these two paradigms in greater detail.

Backward analysis, as done in this paper, allows verifying program properties that hold for \emph{all initial states}.
For instance, our framework supports questions such as, \enquote{\emph{Is the posterior mean of $x$ at most twice the initial value of $y$?}}, expressed as $\cwpSymb[\prog](x) \le 2y$.
The backward approach is thus well-suited for tasks such as \emph{verification} --- ensuring a program behaves as intended \emph{for all inputs} --- and \emph{parametric analysis} --- examining how expected behavior changes when program parameters are altered (parameters can be modeled as uninitialized program variables).

Forward analyses, on the other hand, are typically motivated by probabilistic programming languages (PPLs) explicitly designed for automating Bayesian inference, see e.g., \cite{DBLP:journals/corr/abs-1809-10756,dippl,DBLP:journals/jmlr/BinghamCJOPKSSH19,carpenter2017stan}. 
We remark that there are some fundamental discrepancies between these PPLs and the one studied in this paper.
For example, the former usually offer extensive support for general continuous distributions, but do not always support general loops.
Moreover, Bayesian inference problems often require \emph{soft conditioning}, which our approach only encodes as syntactic sugar (see Remark 2).
Finally, some applications of Bayesian inference involve loading and processing datasets with thousands of observations.
In principle, such data could be hardcoded into our programs, but the result would likely be too large for current SMT-based analysis techniques (as indicated by our experiments in \Cref{sec:empirical_eval}).

In summary, our backward approach is particularly suited for proving properties of complex stochastic processes described as probabilistic programs, especially those involving unbounded loops
 However, it is less effective for data-intensive statistical analysis.

\paragraph{A Zoo of Quantities}

The majority of existing methods for (semi-)automated exact analysis of PP with loops and continuous distributions focus on one of these tasks:
(i) bound the posterior distribution~\cite{DBLP:conf/cav/GehrMV16,DBLP:conf/pldi/BeutnerOZ22,wang2024static},
(ii) prove various flavors of (non-)termination, e.g.,~\cite{DBLP:conf/cav/ChakarovS13,DBLP:conf/cav/ChatterjeeFG16,DBLP:journals/pacmpl/AgrawalC018}, and
(iii) bound assertion violation probabilities~\cite{DBLP:conf/cav/ChakarovS13,DBLP:conf/pldi/WangS0CG21}.
An exception is~\cite{DBLP:conf/pldi/Wang0GCQS19} which tackles \emph{cost analysis}.
In contrast, in the $\wpSymb$-based framework, we can express and reason about a remarkably large variety of quantities, including the following (assume that $\prog$ does not contain conditioning for simplicity):
\begin{itemize}
    \item Termination probabilities: $\wp{\prog}{1}$
    \item The probability of terminating in a predicate\footnote{Recall that $\iv{\guard}$ is the indicator function of the predicate $\guard$.} $\guard$: $\wp{\prog}{\iv{\guard}}$
    \item The expected value of variable $\pVar$ on termination: $\wp{\prog}{\pVar}$
    \item Higher moments of $\pVar$ after termination: $\wp{\prog}{\pVar^k}$ for $k \geq 1$
    \item Expected distance between variables $\pVar$ and $\pVarb$ after termination: $\wp{\prog}{|\pVar-\pVarb|}$
\end{itemize}
\iftoggle{arxiv}{%
Note that this includes reasoning about \emph{unbounded} random variables.
For example, a program variable $\pVar$ may become arbitrarily large during program execution.
}{}

\begin{example}
    Let us illustrate the above quantities.
    Consider the following program $\prog$:
    \begin{align*}
        \ITE{\pVar = 1}{ \PCHOICE{\ASSIGN{\pVarb}{0}}{\nicefrac{1}{2}}{\ASSIGN{\pVarb}{2}}}{\PCHOICE{\ASSIGN{\pVarb}{0}}{\nicefrac{4}{5}}{\ASSIGN{\pVarb}{3}}}
    \end{align*}
    \begin{itemize}
    	\item $\wp{\prog}{1}=1$, i.e., $\prog$ terminates with probability $1$ for each initial state.
    	\item $\wp{\prog}{\iv{\pVarb = 0}} \eeq \nicefrac{1}{2} \cdot \iv{\pVar = 1} + \nicefrac{4}{5} \cdot \iv{\pVar \neq 1}$, i.e., if initially $\pVar=1$, then the probability to terminate in $\pVarb = 0$ is $\nicefrac{1}{2}$. For all other initial states, this probability is $\nicefrac{4}{5}$.
    	\item  $\wp{C}{\pVarb} = 1\cdot \iv{\pVar = 1} + \nicefrac{3}{5} \cdot \iv{\pVar \neq 1}$, i.e., the expected final value of $\pVarb$ is either $1$ or $\nicefrac{3}{5}$, depending on whether $\pVar = 1$ holds initially or not.
    	\item $ \wp{C}{\pVarb^2}  = 2 \cdot \iv{\pVar = 1} + \nicefrac{9}{5} \cdot \iv{\pVar \neq 1}$, i.e., the second moment of $\pVarb$ on termination is either $2$ or $\nicefrac{9}{5}$, depending on whether $\pVar = 1$ holds initially or not.
    	\item $\wp{\prog}{| \pVar - \pVarb |} = \iv{\pVar = 1} \cdot (\nicefrac{1}{2} \cdot |\pVar| + \nicefrac{1}{2} \cdot | \pVar - 2| ) + \iv{\pVar \neq 1} \cdot (\nicefrac{4}{5} \cdot |\pVar| + \nicefrac{1}{5} \cdot | \pVar - 3|)$, which is the expected difference between the final values of $\pVar$ and $\pVarb$ in terms of their initial values. 
    \end{itemize}
    The above weakest pre-expectations can be determined by applying the rules in \Cref{tab:original} (\cpageref{tab:original}).
\end{example}

\paragraph{Conditioning, Everywhere}

Following~\cite{DBLP:journals/toplas/OlmedoGJKKM18}, we allow conditioning in the form of (hard) $\KWOBSERVE$-statements at arbitrary places in the program --- even inside loops.
All of the above quantities can be generalized sensibly to the conditional setting.
We can go further:
For instance, the probability to never violate an $\KWOBSERVE$, even if the program does not terminate, is given by $\wlp{\prog}{1}$.

\subsection{Limitations and Assumptions}

As mentioned, to analyze loops, we assume \emph{user-provided invariants}.
Moreover, we restrict to native support for continuous \emph{uniform} distributions and discrete coin flips.
Further, while we support two-sided bounds for loop-free programs, the kind of invariants we study in this paper can only prove \emph{upper} bounds on $\wpSymb$ and $\cwpSymb$, and \emph{lower} bounds on $\wlpSymb$.
We conjecture that our approach extends to invariant-based methods for the other directions~\cite{DBLP:journals/pacmpl/HarkKGK20} as well, but leave the details for future work.
To keep the presentation simple, we only discuss invariant verification for \emph{non-nested loops}.
The case of general loop structures, including sequential and/or nested, can be dealt with as in~\cite{DBLP:journals/pacmpl/BatzBKW24}.
We assume non-negative real-valued program variables and do not support data structures.
Soft-conditioning is only supported indirectly.
See \Cref{rem:generalDist,rem:nonNeg,rem:softConditioning} for more details.

\subsection{From Lebesgue Integrals to Riemann Sums and Back: A Technical Challenge}
\label{sec:challenges}

\newcommand{\progdirichlet}{D}

The Lebesgue integral is \emph{the} standard integral in probability theory.
There is good reason for this:
Lebesgue integrals can be defined for a broader class of functions than, say, Riemann integrals, and have favorable mathematical properties like Monotone Convergence Theorems.

Therefore, to define the semantics of PPs, most works (e.g., \cite{DBLP:conf/setss/SzymczakK19,dahlqvist2020semantics}) have resorted to Lebesgue integrals, and we do not question that this is the way to go.
To illustrate the usefulness of the Lebesgue integral as opposed to the Riemann integral\footnote{The Riemann integral is the limit of both the lower and upper Riemann sums as $\N \to \infty$, provided the two limits coincide.}, consider the following program $\progdirichlet$: 
\begin{align*}
    & \SEQ{\ASSIGN{\pVar}{0}}{\ASSIGN{\pVarb}{1}} \, \SEMICOLON \\
	& \WHILE{\pVarb\cdot\pVarc \pNEQ \pVar} 
	{\ASSIGN{\pVarb}{\pVarb+1} \,\SEMICOLON\, \ASSIGN{\pVar}{0} \,\SEMICOLON \,
	 \WHILE{\pVarb \cdot\pVarc \pNEQ \pVar \pAND \pVar < \pVarb}{\ASSIGN{\pVar}{\pVar+1}} \, \} }
\end{align*}
This program \enquote{searches} for non-negative integers $\pVar$ and $\pVarb$ such that $\pVarc = \tfrac \pVar \pVarb$ and stops once it finds such $\pVar$ and $\pVarb$.
It thus terminates if and only if $\pVarc$ is a \emph{rational number} in $\uIval$ initially (recall that we allow actual \emph{real}-valued variables).
In symbols, we have $\wp{\progdirichlet}{1} = \iv{\pVarc \in \rats \land 0 \leq \pVarc \leq 1}$ --- this is the well-known \emph{Dirichlet function} restricted to the unit interval.
Now, consider the program $\SEQ{\UNIFASSIGN{\pVarc}}{\progdirichlet}$.
Its termination probability is zero, the Lebesgue integral of the Dirichlet function over the interval $\uIval$.
Intuitively, this is because almost all real numbers are irrational, i.e., the rational numbers have Lebesgue measure zero.
However, the Dirichlet function is \emph{not Riemann-integrable} --- it is \enquote{too discontinuous}.
As a consequence, one cannot easily define $\wpSymb$ --- neither in general nor in this specific example --- relying on Riemann integrals.

In light of the above example, the following question arises naturally:
\begin{myhighlight}
    \emph{How sensible is an analysis of loops based on Riemann sums given the fact that $\wpSymb$'s of loops are \underline{not Riemann-inte}g\underline{rable} in general, not even for programs using only polynomial arithmetic and constant post-expectations?}
\end{myhighlight}
This question has at least two dimensions.

(1) Regarding \emph{soundness}, we prove that our Riemann $\wpSymb$'s are \emph{always} sound under- or over-approximations, i.e., the following inequalities hold for all $\N \geq 1$ in a very general setting:
\[  
\forall~\text{initial states $\st$}\colon\quad
    \lwp{\N}{\prog}{\ex}(\st)
    \lleq
    \wp{\prog}{\ex}(\st)
    \lleq
    \uwp{\N}{\prog}{\ex}(\st)        
\]
and similarly for $\wlpSymb$.
This works because the \emph{lower Riemann integral} --- the limit of the lower Riemann sums as the discretization becomes fines --- is a lower bound on the Lebesgue integral, and similarly for the upper Riemann integral.
This is always true, even if the lower and upper integral are different, i.e., if the function at hand is not Riemann-integrable. 

(2) Will the Riemann approximation always converge (in some sense) to the exact $\wpSymb$ as the discretization becomes finer?
We prove that, under mild assumptions such as polynomial arithmetic in the program, the answer is \emph{yes} for loop-free programs.
As one of our theoretical main results for loopy programs, we show that (\Cref{thm:pointwiseConv})
\[
    \sup_{n \geq 1} \lwp{n}{\unfold{\prog}{n}}{\ex}
    \eeq
    \wp{\prog}{\ex}
    \qquad
    \text{but}
    \qquad
    \inf_{n \geq 1} \uwp{n}{\unfold{\prog}{n}}{\ex}
    ~\overset{\text{in general}}{\neq}~
    \wp{\prog}{\ex}
    ~,
    \tag{$\ddagger$}\label{eq:wow}
\]
where $\unfold{\prog}{n}$ arises from $\prog$ by unfolding all loops up to depth $n$ by taking a simultaneous limit of the unfolding depth and the fineness of the lower Riemann sum approximation (the $n$ in $\lwpSymb{n}$).
Via equation \eqref{eq:wow} we \emph{recover} the exact $\wpSymb$ --- defined in terms of \emph{Lebesgue} integrals --- as a limit of our approximate $\wpSymb$'s based on lower Riemann sums for a general class of probabilistic loops.

\section{Preliminaries}
\label{sec:prelims}

In this section, we set up our notations and treat the fixed point-theoretic foundations we rely on.

\subsection{General Notation}

$\nats$ is the set of non-negative integers.
The set of non-negative \emph{extended reals} is $\ennReals = \nnReals \cup \{\infty\}$.
We adopt the following standard conventions:
For all $x \in \ennReals$, we let $x \leq \infty$ and $x + \infty = \infty + x = \infty$.
Moreover, we define $\infty \cdot 0 = 0 \cdot \infty = 0$ and $x \cdot \infty = \infty \cdot x = \infty$ for all $x > 0$.
Given $\ivalL, \ivalR \in \reals$, we denote by $\clIvalGen$ the real closed interval with endpoints $\ivalL$ and $\ivalR$.
The set of truth values is $\bools = \{\boolConstFalse, \boolConstTrue\}$.
Given a predicate $\guard \colon A \to \bools$ over a set $A$, we define the \emph{Iverson bracket}
\begin{align*}
    \iv{\guard} \colon A \to \{0,1\},~ a \mapsto
    \begin{cases}
        0 & \text{ if } \guard(a) = \boolConstFalse~, \\
        1 & \text{ if } \guard(a) = \boolConstTrue~. 
    \end{cases}
\end{align*}
For example, assuming that it is understood from context that $x$ is a real number, then $\iv{x \in \rats}$ denotes the \emph{Dirichlet function} that sends every $x \in \reals$ to $1$ if $x$ is rational, and to $0$ otherwise.

Lambda notation $\lam{x}{\aExp}$, where $\aExp$ is some mathematical expression with free variable $x$, is used to introduce unnamed functions whose domain and codomain will be clear from the context.

\subsection{Fixed Point Theory}
\label{sec:prelims:fixpoint}

We introduce the foundations from fixed point theory \cite[Section 5.4]{DBLP:books/daglib/0070910} required to sensibly reason about the semantics of loops and to obtain suitable notions of quantitative loop invariants.

Let $\poGen$ and $\poGenb$ be partial orders.
A function $\fun \colon \poDom \to \poDomb$ is called \emph{monotonic}, if for all $\poElem \poleq \poElemb$, we have $\fun(\poElem) \poleqb \fun(\poElemb)$.
An \emph{$\omega$-chain} in $\poGen$ is a monotonic function $\poElem \colon \nats \to \poDom$ (where $\nats$ is ordered by the usual $\leq$ relation), i.e., an $\omega$-chain is a non-decreasing sequence of elements $\poElem(0) \poleq \poElem(1) \poleq \ldots$ from $\poDom$.
The partial order $\poGen$ is an \emph{$\omega$-complete partial order} ($\omega$-cpo, for short), if all $\omega$-chains $\poElem$ in $\poDom$ have a \emph{supremum} (least upper bound) $\sup_{i \in \nats} \poElem(i)$ in $\poGen$. 
An $\omega$-cpo $\poGen$ with an element $\poBot \in \poDom$ satisfying $\poBot \poleq \poElem$ for all $\poElem$ is called \emph{$\omega$-cpo with bottom}.
A function $\fun \colon \poDom \to \poDomb$ between $\omega$-cpos $\poGen$ and $\poGenb$ is called \emph{$\omega$-continuous}, if for all $\omega$-chains $\poElem$ in $\poDom$ it holds that $\sup_{i \in \nats} \fun(\poElem(i)) =  \fun(\sup_{i \in \nats} \poElem(i))$.
Let $\fun \colon \poDom \to \poDom$ be a function where $\poDom$ is an arbitrary set.
An element $\poElem \in \poDom$ satisfying $\fun(\poElem) = \poElem$ is called a \emph{fixed point} of $\fun$.
The following is often attributed to Kleene:
\begin{theorem}[Kleene's Fixed Point Theorem~\textnormal{\cite[Theorem~5.11]{DBLP:books/daglib/0070910}}]
    \label{thm:kleene}
    Let $\poGen$ be an $\omega$-cpo with bottom and let $\fun \colon \poDom \to \poDom$ be $\omega$-continuous.
    Then $\fun$ has a least fixed point $\lfp\fun \in \poDom$ and $\lfp\fun = \sup_{i \in \nats} f^i(\poBot)$ where $\fun^i(\poBot)$ denotes the $i$-fold application of $\fun$ to $\poBot$.
\end{theorem}

A partial order $\poGen$ is said to be \emph{$\omega$-\underline{co}complete} ($\omega$-cocpo) if all \emph{$\omega$-\underline{co}chains}, i.e., non-\emph{increasing} sequences $\poElem(0) \pogeq \poElem(1) \pogeq \ldots$ in $\poDom$ have an infimum $\inf_{i \in \nats}\poElem(i) \in \poDom$.
An $\omega$-cocpo with a greatest element $\poTop$ is called $\omega$-cocpo \emph{with top}.
Note that $\poGen$ is an $\omega$-cocpo (with top) iff the reversed partial order $\po{\poDom}{\pogeq}$ is an $\omega$-cpo (with bottom).
Similarly, a function $\fun \colon \poDom \to \poDomb$ between $\omega$-cocpos $\poGen$ and $\poGenb$ is called \emph{$\omega$-\underline{co}continuous}, if $\inf_{i \in \nats} \fun(\poElem(i)) =  \fun(\inf_{i \in \nats} \poElem(i))$ for all non-increasing $\poElem$.
Note that for an $\omega$-cocpo $\po{\poDom}{\poleq}$ with top, \Cref{thm:kleene} reads as follows:
If $\fun \colon \poDom \to \poDom$ is $\omega$-\underline{co}continuous, then $\fun$ has a \emph{greatest} fixed point $\gfp\fun = \inf_{i \in \nats} \fun^i(\top)$.

A partial order is called \emph{$\omega$-bicomplete} ($\omega$-bicpo), if it is both an $\omega$-cpo and $\omega$-cocpo.
Similarly, a function $\fun$ between $\omega$-bicpos is called \emph{$\omega$-bicontinuous}, if it is both 
$\omega$-continuous and $\omega$-\underline{co}continuous.

A partial order $\poGen$ is a \emph{complete lattice} if for all $\poSubset \subseteq \poDom$ there exists $\sup \poSubset \in \poDom$ and $\inf \poSubset \in \poDom$.
Note that every complete lattice is an $\omega$-bicpo with bottom $\poBot = \sup \emptyset$ and top $\poTop = \inf \emptyset$.

\begin{theorem}[Knaster-Tarski Theorem~\textnormal{\cite[Theorems 5.15 and 5.16]{DBLP:books/daglib/0070910}}]
    \label{thm:knasterTarski}
    Let $\fun \colon \poDom \to \poDom$ be a monotonic function on the complete lattice $\poGen$.
    Then $\fun$ has a \emph{least} and \emph{greatest fixed point} $\lfp\fun \in \poDom$ and $\gfp\fun \in \poDom$.
    Moreover, for all $\poElem \in \poDom$ it holds that
    \begin{align*}
        \fun(\poElem) \poleq \poElem
        \quad\text{implies}\quad
        \lfp\fun \poleq \poElem
        \quad\qqand\quad
        \poElem \poleq \fun(\poElem)
        \quad\text{implies}\quad
        \poElem \poleq \gfp\fun
        ~.
    \end{align*}
\end{theorem}
The above implications are often referred to as \emph{Park (co)induction}~\cite{park1969fixpoint}.

\section{Weakest Pre-Expectations for Probabilistic Programs}
\label{sec:programs}

In this section we define programming language $\pWhile$ and its weakest pre-expectation semantics based on \emph{Lebesgue} integrals as defined in~\cite{DBLP:conf/setss/SzymczakK19}.

\subsection{Program Syntax}
\label{sec:programs:syntax}

For the rest of the paper we fix a finite%
\footnote{Once $\pVars$ is fixed we can only write programs with at most $|\pVars|$ distinct variables. However, as we never make any assumptions about the size of $\pVars$, our theory applies to programs with arbitrarily many variables.
Previous work~\cite{DBLP:conf/setss/SzymczakK19} has considered an infinite $\pVars$, but this requires defining a measure space on the infinite-dimensional $\nonNegReals^\pVars$, which is somewhat more involved.
}
set $\pVars = \{\pVar,\pVarb,\ldots\}$ of program variables.
A \emph{(program) state} is a variable valuation $\pSt \in \pStates$, where $\pStates$ is a shorthand for the set of functions $\pVars \to \nnReals$.
Notice that our program variables range over the \emph{non-negative} reals.
The restriction to non-negative variables is for technical convenience and not essential, see \Cref{rem:nonNeg}.

To obtain a well-defined weakest pre-expectation semantics of programs involving continuous sampling, we need to have some measure-theoretic fundamentals in mind, provided in\iftoggle{arxiv}{ \Cref{sec:prelims:measure}}{~\cite{arxiv}}.
Suffice it to say here that we consider the standard Borel \proseSigmaAlgebra and Lebesgue measure $\lebmes$ on $\pStates$. Lebesgue integrals of a measurable function $\fun \colon \reals \to \ennReals$ over a measurable set $\measurableSet \subseteq \reals$ are denoted by $\int_\measurableSet \fun \, d\lebmes$, or $\int_\measurableSet \fun(x) \, d\lebmes(x)$.
We explicitly allow Lebesgue integrals to evaluate to $\infty$.

Now let $\aExps$ be a set of measurable functions of type $\pStates \to \nnReals$ 
and let $\guards$ be a set of measurable functions of type $\states \to \bools$. 
Elements of $\aExps$ and $\guards$ are called \emph{arithmetic expressions} and \emph{guards}, resp.

\begin{definition}[Probabilistic Programs]
	\label{def:pwhile}
    \newcommand{\syntaxDescr}[1]{\text{\textcolor{gray}{(#1)}}}
    Programs $\prog$ in the set $\pWhileWith{\aExps}{\guards}$ of programs with arithmetic expressions from $\aExps$ and guards from $\guards$ adhere to the following grammar:
	\begin{align*}
		\prog \quad\grammarSymb\quad &\SKIP & &\syntaxDescr{effectless program} \\[-1pt]
		\mmid &\DIVERGE & &\syntaxDescr{nonterminating program} \\[-1pt]
		\mmid &\ASSIGN{\pVar}{\aExp} & &\syntaxDescr{assignment; $\pVar \in \pVars, \aExp \in \aExps$} \\[-1pt]
		\mmid &\OBSERVE{\guard} & &\syntaxDescr{conditioning; $\guard \in \guards $} \\[-1pt]
		\mmid &\UNIFASSIGN{\pVar} & &\syntaxDescr{sample from real interval $\uIval$; $\pVar \in \pVars$} \\[-1pt]
		\mmid &\ITE{\guard}{\prog}{\prog} & &\syntaxDescr{conditional choice; $\guard \in \guards$} \\[-1pt]
		\mmid &\PCHOICE{\prog}{\prob}{\prog} & &\syntaxDescr{probabilistic choice; $\prob \in \uIval \cap \rats$} \\[-1pt]
		\mmid &\SEQ{\prog}{\prog} & &\syntaxDescr{sequential composition} \\[-1pt]
		\mmid &\WHILE{\guard}{\prog} & &\syntaxDescr{while loop; $\guard \in \guards$}
	\end{align*}
	If $\aExps$ and $\guards$ are the sets of \emph{all} measurable functions of the corresponding type, we write $\pWhile$ instead of $\pWhileWith{\aExps}{\guards}$.
	A program not containing while-loops is called \emph{loop-free}.\qedDef
\end{definition}

Let us briefly go over each construct, all of which are standard.
$\SKIP$ does nothing.
$\DIVERGE$ is a non-terminating program, i.e., behaves like $\WHILE{\boolConstTrue}{\SKIP}$.
$\ASSIGN{\pVar}{\aExp}$ assigns the value of the arithmetic expression $\aExp$ evaluated in the current state to variable $\pVar$.
$\UNIFASSIGN{\pVar}$ assigns to $\pVar$ a value drawn from the continuous uniform $\uIval$-distribution.
$\OBSERVE{\guard}$  \emph{conditions} the program execution on the guard $\guard$ being true.
$\PCHOICE{\prog_1}{\prob}{\prog_2}$ executes $\prog_1$ with probability $\prob$, otherwise $\prog_2$.
The assumption $\prob \in \uIval \cap \rats$ is to avoid defining a dedicated syntax for probabilities later on in \Cref{sec:syntax}.
$\ITE{\guard}{\prog_1}{\prog_2}$, $\SEQ{\prog_1}{\prog_2}$, and $\WHILE{\guard}{\prog}$ are standard conditional choices, sequential compositions, and while-loops, respectively.
See\iftoggle{arxiv}{ \Cref{app:redudanciesSyntax}}{~\cite{arxiv}} for additional remarks.

\begin{rem}[More General Distributions]
    \label{rem:generalDist}
    In theory, the restriction to uniform $\unitInterval$-distributions does not limit expressiveness:
    Arbitrary distributions can be simulated by sampling from $\unitInterval$ and applying the inverse
	\emph{cumulative distribution function} (CDF) of the target distribution~\cite{DBLP:conf/setss/SzymczakK19}.
    However, to obtain decidability results, we further restrict the syntax of arithmetic expressions and guards to a class of functions expressible in first-order (FO) real arithmetic, see \Cref{sec:syntax}.
    Distributions whose inverse CDF belongs to this class include triangular, trapezoidal, U-quadratic, and Kumaraswamy distributions.
    On the other hand, distributions with transcendental inverse CDF (Gaussian, Laplace, etc.) do not reside in this class.
    We mention two future directions to address this issue:
    (i) leverage heuristics implemented in modern SMT solvers to discharge the generated verification conditions, even if they do not belong to a decidable theory,
    (ii) soundly over/under-approximate the transcendental inverse CDF by FO-expressible algebraic functions.
\end{rem}

\begin{rem}[Soft Conditioning]
	\label{rem:softConditioning}
    Our syntax has native support for \emph{hard conditioning} ($\KWOBSERVE$). Bounded \emph{soft conditioning} (scoring) --- multiplying the current execution with a weight in $\uIval$ --- can be simulated using $\UNIFASSIGN{\pVar}$ and $\KWOBSERVE$, see~\cite[Lemma 5]{DBLP:conf/setss/SzymczakK19} for details.
\end{rem}

\subsection{Weakest Pre-Expectation Semantics}
\label{sec:programs:wp}

\begin{table}[t]
    \caption{
    	Inductive definition of weakest (liberal) pre-expectations for post-expectation $\ex$~\cite{DBLP:conf/setss/SzymczakK19}.
    }
    \label{tab:original}
    \begin{adjustbox}{max width=\textwidth}
        \def\arraystretch{1.2} 
        \begin{tabular}{l l l} 
            \toprule
            $\prog$ & $\wp{\prog}{\ex}\quad$\textcolor{gray}{where $\ex\in\expsmeas$} & $\wlp{\prog}{\ex}\quad$ ~~\textcolor{gray}{where $\ex\in\bexpsmeas$}  \\
            \midrule
            $\SKIP$ & $\ex$ & $\ex$ \\
            $\DIVERGE$ & $0$ & $1$ \\
            $\ASSIGN{\pVar}{\aExp}$ & $\exSubsGen$ & $\exSubsGen$ \\
            $\UNIFASSIGN{\pVar}$ & $\lam{\st}{\int_\uIval \ex(\pStUpdate{\st}{\pVar}{\xi}) \,d\lebmes(\xi)}$ & $\lam{\st}{\int_\uIval \ex(\pStUpdate{\st}{\pVar}{\xi}) \,d\lebmes(\xi)}$ \\
            $\OBSERVE{\guard}$     &$\iv{\guard} \cdot \ex$ &$\iv{\guard} \cdot \ex$ \\
            $\ITE{\guard}{\prog_1}{\prog_2}$ & $\iv{\guard} \cdot \wp{\prog_1}{\ex}+\iv{\neg\guard} \cdot \wp{\prog_2}{\ex}$ & $\iv{\guard} \cdot \wlp{\prog_1}{\ex} + \iv{\neg\guard} \cdot \wlp{\prog_2}{\ex}$ \\
            $\PCHOICE{\prog_1}{\prob}{\prog_2}$ & $\prob \cdot \wp{\prog_1}{\ex} + (1{-}\prob) \cdot \wp{\prog_2}{\ex}$ & $\prob \cdot \wlp{\prog_1}{\ex} + (1{-}\prob) \cdot \wlp{\prog_2}{\ex}$ \\
            $\SEQ{\prog_1}{\prog_2}$ & $\wp{\prog_1}{\wp{\prog_2}{\ex}}$ & $\wlp{\prog_1}{\wlp{\prog_2}{\ex}}$ \\
            $\WHILE{\guard}{\progBody}$ & $\lfpIn{\lambda\fpVar}{\iv{\neg\guard} \cdot \ex + \iv{\guard} \cdot \wp{\progBody}{\fpVar}}$ & $\gfpIn{\lambda\fpVar}{\iv{\neg\guard} \cdot \ex + \iv{\guard} \cdot \wlp{\progBody}{\fpVar}}$ \\
            \bottomrule
        \end{tabular}
    \end{adjustbox}
\end{table}

We now unify the weakest pre-expectation calculi for continuous probabilistic programs from \cite{DBLP:conf/setss/SzymczakK19} with the calculi proposed in \cite{DBLP:journals/toplas/OlmedoGJKKM18}. The latter calculi take \emph{renormalization} --- for conditional expected outcomes --- into account.
The central objects these calculi operate on are \emph{expectations}:
\begin{definition}[Expectations]
    We distinguish between the following sets of functions:
    \begin{enumerate}
        \item The set of \emph{expectations} is $\exps = \{\ex \mid \ex\colon \pStates \to \exNonNegReals \}$.
        \item The set of \emph{1-bounded expectations} is $\bexps = \{\ex \mid \ex \colon \pStates \to \uIval \}$.
        \item The set of \emph{measurable (1-bounded) expectations} $\expsmeas$ ($\bexpsmeas$) is the subset of $\exps$ ($\bexps$) containing exactly the Borel-measurable functions.
    \end{enumerate}
    We equip all of these sets with the partial order $\eleq$ defined as $\ex \eleq \exb$ iff $\forall \st \in \states \colon \ex(\st) \leq \exb(\st)$.
    \qedDef
\end{definition}
Crucially, we have (see \cite[Lemma 2]{DBLP:conf/setss/SzymczakK19} and\iftoggle{arxiv}{ \Cref{proof:measexpsbicpo}}{~\cite{arxiv}}):
\begin{restatable}{lemma}{measexpsbicpo}
    \label{thm:measexpsbicpo}
    $\po{\exps}{\eleq}$ and $\po{\bexps}{\eleq}$ are complete lattices.
    $\po{\expsmeas}{\eleq}$ and $\po{\bexpsmeas}{\eleq}$ are $\omega$-bicpos with bottom ($0 \gray{{}= \lam{\pSt}{0}}$) and top ($\infty \gray{{}= \lam{\pSt}{\infty}}$ and $1 \gray{{}= \lam{\pSt}{1}}$, respectively).
\end{restatable}
The arithmetic operations $+$ (addition) and $\cdot$ (multiplication) on $\exps$ are defined {pointwise}, i.e., $\forall \st \in \states \colon (\ex + \exb)(\st) = \ex(\st) + \exb(\st)$, and analogously for multiplication.
These operations preserve measurability~\cite[Theorem 11.18]{rudin1953principles}, i.e., $\expsmeas$ is closed under $+$ and $\cdot$.
Moreover, addition (for both arguments) and multiplication by constants%
\footnote{Formally, for every $\exb \in \exps$ these are the functions $\lam{\ex}{\ex + \exb}$ and $\lam{\ex}{\ex \cdot \exb}$.}
are $\omega$-bicontinuous functions.

For state $\pSt \in \pStates$,  program variable $\pVar \in \pVars$, $\xi \in \nnReals$, we define the \emph{updated state as}
\[
    \pStUpdate{\pSt}{\pVar}{\xi}
    \eeq
    \lam{\pVarb}{
        \begin{cases}
        \xi & \text{if $\pVarb = \pVar$} \\
        \pSt(\pVarb) &\text{otherwise}
        \end{cases}
    }
    ~.
\]
Further, given an expectation $\ex \in \exps$, a variable $\pVar \in \pVars$, and an arithmetic expression $\aExp \colon 
\pStates \to \nnReals$ we define the \emph{substitution of $\pVar$ by $\aExp$ in $\ex$} as the expectation $\exSubsGen = \lam{\st}{\ex(\pStUpdate{\st}{\pVar}{\aExp(\st)})}$.
If we have syntactic expressions for $\ex$ and $\aExp$, then we can obtain a syntactic representation of $\exSubs{f}{\pVar}{\aExp}$ by substituting all free occurrences of $\pVar$ in $\ex$ by $\aExp$ in a capture-avoiding manner, see \Cref{sec:syntax}.

\begin{definition}[Weakest (Liberal) Pre-expectation Transformers~\textnormal{\cite{DBLP:conf/setss/SzymczakK19}}]
    \label{def:wpwlp}
    For all $\prog \in \pWhile$, the \emph{weakest (liberal) pre-expectation transformers} $\wpTrans{\prog} \colon \expsmeas \to \expsmeas$ and $\wlpTrans{\prog}\colon \bexpsmeas \to \bexpsmeas$ are defined inductively on the structure of $\prog$ according to the rules in \Cref{tab:original}.
    \qedDef
\end{definition}
Let us briefly explain the inductive definition in \Cref{tab:original}.
The effectless program $\SKIP$ leaves the post-expectation unchanged.
An assignment $\ASSIGN{\pVar}{\aExp}$ substitutes the expression $\aExp$ for the variable $\pVar$ in the post-expectation $\ex$, while uniform assignment $\UNIFASSIGN{\pVar}$ integrates the expectation over all possible values of $\pVar$ in the interval $\uIval$, capturing the averaging effect of drawing $\pVar$ uniformly.
$\OBSERVE{\guard}$ scales the expectation by the Iverson bracket of the guard $\guard$.
Proper renormalization is performed in a second step, see \Cref{sec:wp:cwp}.
For $\ITE{\guard}{\prog_1}{\prog_2}$ and the probabilistic choice $\PCHOICE{\prog_1}{\prob}{\prog_2}$, the weakest (liberal) pre-expectation is a weighted sum of the $\wpSymb$'s from both branches, either weighted by $\iv{\guard}$ and $\iv{\neg\guard}$ in the former case, or by the probabilities $\prob$ and $1-\prob$ in the latter.
The $\wpwlpSymb$ of a sequential composition $\SEQ{\prog_1}{\prog_2}$ is function composition $\wpwlpTrans{\prog_1} \circ \wpwlpTrans{\prog_2}$.
Notably, $\prog_2$ is evaluated \emph{before} $\prog_1$ --- $\wpwlpSymb$'s are thus computed in a backward manner.
Note that $\wpSymb$ and $\wlpSymb$ only differ in the handling of divergence and loops.
The $\DIVERGE$ command, representing non-termination, sets the $\wpSymb$ to $0$ and the $\wlpSymb$ to $1$.
The $\WHILE{\guard}{\prog}$ loop involves computing the \emph{least} fixed point ($\lfp$) for $\wpSymb$ and the \emph{greatest} fixed point ($\gfp$) for $\wlpSymb$.
The difference between $\wpSymb$ and $\wlpSymb$ is further explained in \Cref{sec:wp:wpVsWlp}.

\begin{rem}[On Non-Negativity]
    \label{rem:nonNeg}
    We follow the classic line of work on weakest pre-expectations~\cite{DBLP:conf/stoc/Kozen83,DBLP:series/mcs/McIverM05} and restrict attention to \emph{non-negative expectations} (see~\cite{DBLP:conf/lics/KaminskiK17} for a discussion of the mixed-sign case).
    This means that we can only reason about expected values of non-negative random variables measured in a program's final state.
    As a consequence, in order to reason about the expected final value of a program variable $\pVar$ on termination, we have to assume that $\pVar$ is unsigned.
    We have opted to ensure this by simply requiring \emph{all} variables to be unsigned real numbers.
\end{rem}

Given a loop $\prog = \WHILE{\guard}{\progBody}$ and $\ex \in \expsmeas$, we denote by 
\[
	\charfunwp{\prog}{\ex} \colon \expsmeas \to \expsmeas \,, 
	\qquad 
	\charfunwp{\prog}{\ex}(\exb) \eeq \iv{\guard} \cdot \wp{\progBody}{\exb} + \iv{\neg\guard}\cdot \ex
\]
the \emph{$\wpSymb$-characteristic function of $\prog$ w.r.t.\ $\ex$} (analogously for $\ex' \in \bexpsmeas$ and $\charfunwlp{\prog}{\ex'}$).
Hence, weakest pre-expectations of loops can be denoted more concisely as 
\begin{align*}
	 \wp{\prog}{\ex} \eeq \lfp \charfunwp{\prog}{\ex}
	 \qquad\text{and}\qquad
	  \wlp{\prog}{\ex'} \eeq \gfp \charfunwlp{\prog}{\ex'}~.
\end{align*}
Due to the fixed points in \Cref{tab:original} it is not immediately obvious that $\wpTrans{\prog}$ and $\wlpTrans{\prog}$ are well-defined.
This will we be ensured by Kleene's \Cref{thm:kleene} as shown in the next lemma (parts of which have been proved in~\cite{DBLP:conf/setss/SzymczakK19}).
\begin{restatable}[{Well-definedness of $\wpSymb$ and $\wlpSymb$}]{theorem}{wpWlpWellDefinedAndContinuous}
    \label{thm:wpWlpWellDefinedAndContinuous}
    For all programs $\prog \in \pWhile$, $\wpTrans{\prog}$ and $\wlpTrans{\prog}$ are well-defined.
    In particular, $\wpTrans{\prog}$ is $\omega$-continuous and $\wlpTrans{\prog}$ is $\omega$-\underline{co}continuous.
\end{restatable}

\subsection{Probabilistic Termination: $\wpSymb$ vs.\ $\wlpSymb$}
\label{sec:wp:wpVsWlp}

In general, for all $\prog \in \pWhile$ and $\ex\in\bexpsmeas$ we have $\wp{\prog}{\ex} \eleq \wlp{\prog}{\ex}$ since the former relies on a least and the latter on a greatest fixed point.
More specifically, we have~\cite[Section 5]{DBLP:conf/setss/SzymczakK19}
\begin{align}
    \wp{\prog}{\ex} + \wlp{\prog}{0} \eeq \wlp{\prog}{\ex} \label{eq:wpwlp}
    ~.
\end{align}
\iftoggle{arxiv}{%
Let us make the meaning of $\wpSymb$ and $\wlpSymb$ w.r.t.\ the constant post-expectations $0$ and $1$ explicit.
Assuming $\prog$ is started with initial state $\st$ we have the following:
\footnote{These statements can be formalized by means of an operational semantics for $\pWhile$, see~\cite{DBLP:conf/setss/SzymczakK19}.}%
\begin{itemize}
    \item $\wp{\prog}{0}(\st)$ is always 0.
    \item $\wp{\prog}{1}(\st)$ is the probability of \emph{termination} without violating any $\KWOBSERVE$.
    \item $\wlp{\prog}{0}(\st)$ is the probability of \emph{non-termination} without violating any $\KWOBSERVE$.
    \item $\wlp{\prog}{1}(\st)$ is the probability to never violate any $\KWOBSERVE$.
\end{itemize}
}{}
We call $\prog$ \emph{almost-surely terminating} (AST) if $\wlp{\prog}{0} = 0$, i.e., if the program does not admit any initial state for which the program's infinite runs that do not violate any $\KWOBSERVE$ have positive probability mass.
It follows from equation \eqref{eq:wpwlp} that $\prog$ is AST iff $\wp{\prog}{\ex} = \wlp{\prog}{\ex}$.

\subsection{Conditional Weakest Pre-Expectations}
\label{sec:wp:cwp}

For a program $\prog \in \pWhile$, an expectation $\ex \in \expsmeas$ and a state $\pSt$, following \cite{DBLP:journals/toplas/OlmedoGJKKM18}, we define:
\[
    \cwp{\prog}{\ex}(\pSt)
    \eeq
    \begin{cases}
        \frac{\wp{\prog}{\ex}(\pSt)}{\wlp{\prog}{\onefun}(\pSt)} &\text{ if } \wlp{\prog}{1}(\pSt) \neq 0 \\[0.5em]
        \text{undefined} & \text{ else.}
    \end{cases}
\]
The above definition factors out the probability mass of runs violating an $\KWOBSERVE$, i.e., divides by $\wlp{\prog}{\onefun}(\pSt)$, see \Cref{sec:wp:wpVsWlp}.
$\cwp{\prog}{\ex}(\st)$ is thus the expected value of $\fun$ after termination of $\prog$ started with initial state $\st$, \emph{conditioned} on all $\KWOBSERVE$ statements in $\prog$ being successful.

\section{Approximate Riemann Weakest Pre-Expectations}
\label{sec:approxwp}

We start by recalling lower and upper Riemann sums and integrals.%
\footnote{
    The definition of an integral in terms of these sums is in fact commonly attributed to Darboux, not to Riemann.
    However, Darboux's integral is equivalent to Riemann's which, rather than considering lower and upper sums, relies on evaluating $\fun$ at sample points within the intervals of a partition of the integration domain.
    We refer to~\cite[Chapter 3, Theorems 3.3.1 and 3.3.2]{burk2007garden} for an in-depth comparison.
    In this paper, we consistently use Darboux's definitions, but refer to them nonetheless as Riemann sums and integrals, since the latter terminology is more widespread.
}
Let $\clIvalGen \subseteq \reals$ and $\partitionSize \geq 1$.
A \emph{partition} of $\clIvalGen$ is a tuple of at least $\partitionSize+1$ real numbers $\partition = (x_0,x_1,\ldots,x_{\partitionSize})$ such that
\[
    \ivalL
    =
    x_0
    \llt
    x_1
    \llt
    \ldots
    \llt
    x_{\partitionSize}
    =
    \ivalR
    ~.
\]
The set of all partitions of the interval $\clIvalGen$ is denoted $\partitions{\clIvalGen}$.
For a bounded $\fun \colon \clIvalGen \to \reals$ and a partition $\partition = (x_0,\ldots,x_{\partitionSize}) \in \partitions{\clIvalGen}$, we define the \emph{lower} and \emph{upper sums} of $\fun$ w.r.t.\ $\partition$ as
\[
    \lowerSum{\fun}{\partition}
    \eeq
    \sum_{i=1}^{\partitionSize} (x_{i} - x_{i-1}) \inf_{\xi \in \clIval{x_{i-1}}{x_i}} \fun(\xi)
    \qqand
    \upperSum{\fun}{\partition}
    \eeq
    \sum_{i=1}^{\partitionSize} (x_{i} - x_{i-1}) \sup_{\xi \in \clIval{x_{i-1}}{x_i}} \fun(\xi)
    ~.
\]
Note that $\upperSum{\fun}{\partition}$ and $\lowerSum{\fun}{\partition}$ are well-defined real numbers because $\fun$ is bounded.

\begin{definition}[Riemann integral]
    \label{def:riemannIntegral}
    Let $\fun \colon \clIvalGen \to \reals$ be bounded.
    The \emph{lower-} and \emph{upper Riemann integrals} of $\fun$ are defined as follows:
    \[
        \lowerIntGen \fun(x) \,dx
        \eeq
        \sup \left\{ \lowerSum{\fun}{\partition} \mid \partition \in \partitions{\clIvalGen} \right\}
        \qqand
        \upperIntGen \fun(x) \,dx
        \eeq
        \inf \left\{ \upperSum{\fun}{\partition} \mid \partition \in \partitions{\clIvalGen} \right\}
        ~.
    \]
    If the upper integral equals the lower integral, then the common value is written $\int_{\ivalL}^{\ivalR} \fun(x) \,dx$ and called \emph{the} Riemann integral of $\fun$.
    In this case, $\fun$ is called \emph{Riemann-integrable}.
    \qedDef
\end{definition}

Finally, we  introduce the following terminology:
(i) For $\funb \colon \realDomain \to \reals$, $\realDomain \subseteq \reals$, we say that $\funb$ is Riemann-integrable \emph{on an interval} $\clIval{\ivalLb}{\ivalRb} \subset \realDomain$ if the restriction $\funb \colon \clIval{\ivalLb}{\ivalRb} \to \reals$ is Riemann-integrable.
(ii) In this paper we often consider functions of the form $\funb \colon \realDomain^\pVars \to \reals$ where $\pVars$ is a finite set, e.g.\ the set of program variables.
We then say that $\funb$ is Riemann-integrable on $\clIval{\ivalLb}{\ivalRb} \subseteq \realDomain$ w.r.t.\ some $\pVar \in \pVars$ if \emph{for all} $\pSt \in \realDomain^\pVars$ the function $\lam{\xi}{\funb(\pStUpdate{\pSt}{\pVar}{\xi})}$ of type $\clIval{\ivalLb}{\ivalRb} \to \reals$ is Riemann-integrable.

\iftoggle{arxiv}{
\begin{example}
    To illustrate the concepts defined above consider the following:
    \begin{itemize}
        \item A constant function $\fun \colon \clIvalGen \to \reals, x \mapsto \realConst$ for $\realConst \in \reals$ is Riemann-integrable because for \emph{all} partitions $\partition$ we have $\lowerSum{\fun}{\partition} = \upperSum{\fun}{\partition} = \realConst (\ivalL-\ivalR)$.
        \item $\fun(x,y) = \iv{x^2 + y^2 \leq 1}$ is Riemann-integrable on $\unitInterval$ w.r.t.\ $x$.
        Indeed, for all $y \in \reals$, we have $\int_{0}^{1} \fun(x,y) \, dx = \iv{-1 \leq y \leq 1} \sqrt{1-y^2}$.
        \item The Dirichlet function $\fun \colon \clIvalGen \to \reals, x \mapsto \iv{x \in \rats}$ is not Riemann-integrable because for all partitions $\partition$ we have $\lowerSum{\fun}{\partition} = 0$ and $\upperSum{\fun}{\partition} = 1$.
        Note that $\fun$ is nowhere continuous.
    \end{itemize}
\end{example}
}{}

\subsection{{$\lwpSymb{\N}$ and $\uwpSymb{\N}$: Lower and Upper Riemann Weakest Pre-Expectations}}
\label{sec:approxwp:def}

We are now ready to define our approximate expectation transformers.
They arise from the standard transformers defined in \Cref{tab:original} by replacing the \emph{Lebesgue} integrals in the $\UNIF_{\uIval}$ case by a lower or upper Riemann sum. 
Formally:

\begin{definition}[Lower and Upper Riemann $\wpwlpSymb$-Transformers]
	\label{def:riemannwpwlp}
	For all $\prog \in \pWhile$ and integers $\N \geq 1$, the expectation transformers $\lwpTrans{\N}{\prog} \colon \exps \to \exps$ and $\uwpTrans{\N}{\prog} \colon \exps \to \exps$ are defined by induction over the structure of $\prog$ as in \Cref{tab:original}, with the only exception that if $\prog$ is $\UNIFASSIGN{\pVar}$, then for all $\ex \in \exps$ we set%
	\footnote{Notice that $\inf_{\xi \in \clIvalGen}\exSubs{\ex}{\pVar}{\xi}$ is the same as $\lam{\st}{\inf_{\xi \in \clIvalGen} \ex(\pStUpdate{\st}{\pVar}{\xi})}$.}
	\begin{align*}
		&\lwp{\N}{\UNIFASSIGN{\pVar}}{\ex}
		\eeq
		\frac{1}{\N} \sum_{i=0}^{\N-1} \inf_{\xi \in [\frac{i}{\N}, \frac{i+1}{\N}]} \exSubs{\ex}{\pVar}{\xi} \qquad\text{and, similarly,} \\
		&\uwp{\N}{\UNIFASSIGN{\pVar}}{\ex}
		\eeq
		\frac{1}{\N} \sum_{i=0}^{\N-1} \sup_{\xi \in [\frac{i}{\N}, \frac{i+1}{\N}]} \exSubs{\ex}{\pVar}{\xi}
		~.
	\end{align*}
	The lower and upper Riemann weakest \emph{liberal} pre-expectation transformers $\lwlpTrans{\N}{\prog} \colon \bexps \to \bexps$ and $\uwlpTrans{\N}{\prog} \colon \bexps \to \bexps$ are defined analogously.
	\qedDef
\end{definition}

Our Riemann transformers approximate the Lebesgue integral in the definition of $\wpSymb$ (and $\wlpSymb$) by a lower or upper Riemann sum.
We work with the partitions $0 < \tfrac{1}{\N} < \tfrac{2}{\N} < \ldots < 1$ of the unit interval for the sake of concreteness.
Note that these partitions are \emph{not} successive refinements of each other (see\iftoggle{arxiv}{ \Cref{thm:partitionRefine}}{~\cite{arxiv}} for definitions).
Thus, increasing $\N$ by $1$ does \emph{not necessarily} yield \enquote{better} approximations, i.e., \emph{we might have} $\lwp{\N}{\prog}{\ex} \not\eleq \lwp{\N+1}{\prog}{\ex}$%
\iftoggle{arxiv}{%
, e.g:
\[
    \lwp{2}{\UNIFASSIGN{\pVar}}{\iv{x \geq \tfrac 1 2}} ~\not\eleq~ \lwp{3}{\UNIFASSIGN{\pVar}}{\iv{x \geq \tfrac 1 2}}
    ~.
\]
}{.}

Given a loop $\prog=\WHILE{\guard}{\progBody}$, $\N\geq 1$, $\ex\in\exps$, and $\transSymb\in\{\uwpSymb{\N},\lwpSymb{\N}\}$, we define the \emph{$\transSymb$-characteristic function of $\prog$ w.r.t.\ $\ex$} as 
\[
\charfuntrans{\prog}{\ex} \colon \exps \to \exps, 
\qquad 
\charfuntrans{\prog}{\ex}(\exb) \eeq \iv{\guard} \cdot \somewp{\progBody}{\exb} + \iv{\neg\guard}\cdot \ex~.
\]
For $\exb\in\bexps$ and $\transSymb\in\{\uwlpSymb{\N},\lwlpSymb{\N}\}$, $\charfuntrans{\prog}{\exb} \colon \bexps \to \bexps$ 
\iftoggle{arxiv}{%
we analogously define
\[
	\charfuntrans{\prog}{\exb} \colon \bexps \to \bexps, 
	\qquad 
	\charfuntrans{\prog}{\exb}(\exb) \eeq \iv{\guard} \cdot \somewp{\progBody}{\exb} + \iv{\neg\guard}\cdot \exb~.
\]
}{%
is defined analogously.
}
Note that unlike $\wpSymb$ and $\wlpSymb$, the approximate transformers are defined on the full sets of expectations $\exps$ and $\bexps$, respectively, not just on the measurable ones.
All transformers from \Cref{def:riemannwpwlp} are well-defined: this follows from the Knaster-Tarski \Cref{thm:knasterTarski} using that $\exps$ and $\bexps$ are complete lattices and monotonicity of the transformers, see \Cref{thm:monriemann} below.

\subsection{Healthiness Properties of $\lwpSymb{\N}$ and $\uwpSymb{\N}$}
\label{sec:approxwp:basic}

\begin{restatable}[Monotonicity of the Riemann $\wpwlpSymb$-Transformers]{lemma}{monriemann}
	\label{thm:monriemann}
	For all $\prog \in \pWhile$ and integers $\N \geq 1$, the functions $\lwpTrans{\N}{\prog}$, $\uwpTrans{\N}{\prog}$, $\lwlpTrans{\N}{\prog}$, and $\uwlpTrans{\N}{\prog}$ are monotonic w.r.t.\ $\eleq$.
\end{restatable}

Notably, the Riemann $\wpwlpSymb$-transformers do not possess the same continuity properties as their Lebesgue counterparts from \Cref{sec:programs:wp}.
For instance, due to the presence of infima in the lower Riemann sum, $\lwpTrans{\N}{\prog}$ is \emph{not} $\omega$-continuous in general (see\iftoggle{arxiv}{ \Cref{proof:lwpNotOmegaCont}}{~\cite{arxiv}} for a counter-example).

\begin{restatable}[Soundness of the Riemann $\wpwlpSymb$-Transformers]{lemma}{soundness}
	\label{thm:soundness}
	For all programs $\prog \in \pWhile$, post-expectations $\ex \in \expsmeas$, $\exb \in \bexpsmeas$, and integers $\N \geq 1$,
	\begin{align*}
		\lwp{\N}{\prog}{\ex}
		&\eeleq
		\wp{\prog}{\ex}
		\eeleq
		\uwp{\N}{\prog}{\ex} \qand 
		\\
		\lwlp{\N}{\prog}{\exb}
		&\eeleq
		\wlp{\prog}{\exb}
		\eeleq
		\uwlp{\N}{\prog}{\exb}
		~.
	\end{align*}
\end{restatable}

\noindent For \emph{conditional} weakest pre-expectations (\Cref{sec:wp:cwp}), we immediately get the following:
\begin{corollary}
    \label{thm:cwpSound}
	For all programs $\prog \in \pWhile$, post-expectations $\ex \in \expsmeas$, and integers $\N \geq 1$
	\begin{align*}
		\frac{\lwp{\N}{\prog}{\ex}(\pSt)}{\uwlp{\N}{\prog}{\onefun}(\pSt)}
		\lleq
		\cwp{\prog}{\ex}(\pSt)
		\lleq
		\frac{\uwp{\N}{\prog}{\ex}(\pSt)}{\lwlp{\N}{\prog}{\onefun}(\pSt)}
	\end{align*}
	for all $\pState \in \pStates$ such that none of $\lwlp{\N}{\prog}{1}(\pSt)$, $\wlp{\prog}{1}(\pSt)$, and $\uwlp{\N}{\prog}{1}(\pSt)$ is zero.
\end{corollary}
\noindent We stress that \Cref{thm:soundness} and \Cref{thm:cwpSound} hold even for non-Riemann-integrable post-expectations.

\section{Invariant-Based Reasoning for Loops}
\label{sec:invariants_and_unrolling}
Deductive probabilistic program verification techniques typically bound expected outcomes of loops by means of \emph{quantitative loop invariants}. Intuitively, a quantitative loop invariant $\exI$ is an expectation whose pre-expectation w.r.t.\ \emph{one loop iteration} does not increase (or decrease, depending on whether one wants to establish upper or lower bones, see below).
While quantitative loop invariant-based reasoning can simplify the verification of loops significantly, 
the continuous setting poses challenges: Computing the pre-expectation of $\exI$ w.r.t.\ a loop's body requires reasoning about possibly complex Lebesgue integrals --- involving, e.g., indicator functions of predicates --- arising from the continuous sampling instructions.
We now develop invariant-based proof rules for our \emph{Riemann pre-expectations} (\Cref{def:riemannwpwlp}), which yield sound bounds on the \emph{Lebesgue pre-expectations} (\Cref{def:wpwlp}).
We will see in \Cref{sec:syntax,sec:case_studies} that these proof rules give rise to SMT-based techniques for verifying bounds on Lebesgue pre-expectations of loops in a semi-automated fashion.
Our first insight is that --- due to $(\exps, \, \eleq)$ and $(\bexps, \, \eleq)$ being complete lattices and the Riemann expectation transformers being monotonic by \Cref{thm:monriemann} --- bounds on \emph{Riemann} pre-expectations of loops can be established via Park induction (\Cref{thm:knasterTarski}):
\begin{lemma}
	\label{thm:invariant_approx}
	Let $\prog = \WHILE{\guard}{\progBody} \in\pWhile$ and $\N \geq 1$. Then:
    \begin{enumerate}
		\item For all $\ex,\exI \in \exps$, we have
		$
			\underbrace{\charfunuwp{\N}{\prog}{\ex}(\exI) \eleq \exI}_{\text{$\exI$ is $\uwpSymb{\N}$-superinvariant of $\prog$ w.r.t.\ $\ex$}}
			\qquad\text{implies}\qquad
			\uwp{\N}{\prog}{\ex} \eleq \exI~.
		$
		\item For all $\ex,\exI \in \bexps$, we have
		$
            \underbrace{\exI \eleq \charfunlwlp{\N}{\prog}{\ex}(\exI)}_{\text{$\exI$ is $\lwlpSymb{\N}$-subinvariant of $\prog$ w.r.t.\ $\ex$}}
            \qquad\text{implies}\qquad
            \exI \eleq \lwlp{\N}{\prog}{\ex}~.
		$
	\end{enumerate}
\end{lemma}
More colloquially stated, \emph{super}invariants yield \emph{upper} bounds on \emph{upper} Riemann pre-expectations of loops and, dually, \emph{sub}invariants yield \emph{lower} bounds on \emph{lower liberal} Riemann pre-expectations. Notice that establishing the premise of the above proof rules only requires reasoning about the loop's {body} and avoids explicitly computing Lebesgue integrals.

It thus follows from the soundness of our Riemann expectation transformers (\Cref{thm:soundness}) that the above proof rules yield sound bounds on \emph{Lebesgue} pre-expectations.
\begin{theorem}
	\label{thm:invariant_cont}
	Let $\prog = \WHILE{\guard}{\progBody} \in\pWhile$ and $\N,\N' \geq 1$. We have:
	\begin{enumerate}
		\item If $\exI \in \exps$ is a $\uwpSymb{\N}$-superinvariant of $\prog$ w.r.t.\ $\ex\in\expsmeas$, then
		$\wp{\prog}{\ex} \eleq \exI~$.
		\item If $\exI \in \bexps$ is a $\lwlpSymb{\N}$-subinvariant of $\prog$ w.r.t.\ $\ex\in\bexpsmeas$, then
		$\exI \eleq \wlp{\prog}{\ex}$.
		\item If $\exI \in \exps$ is a $\uwpSymb{\N}$-superinvariant of $\prog$ w.r.t.\ $\ex\in\expsmeas$ and $\exJ \in \bexps$ is a $\lwlpSymb{\N'}$-subinvariant of $\prog$ w.r.t.\ $1$, then, for all $\pState\in\pStates$, $\exJ(\pState) >0$ implies that $\cwp{\prog}{\ex}(\pState)$ is defined and
		\[
			\cwp{\prog}{\ex}(\pState) \lleq \frac{\exI(\pState)}{\exJ(\pState)}~.
		\]
	\end{enumerate}
\end{theorem}
It is important to note that in \Cref{thm:invariant_approx} we admit \emph{arbitrary} ($1$-bounded) post-expectations whereas in \Cref{thm:invariant_cont} we have to restrict to \emph{measurable} ($1$-bounded) post-expectations in order for the Lebesgue pre-expectations $\wp{\prog}{\ex}$ and $\wlp{\prog}{\ex}$ to be well-defined.
The quantitative loop invariant $\exI$, on the other hand, must \emph{not} necessarily be measurable since there is no need to plug $\exI$ into a Lebesgue expectation transformer.
\begin{example}
	\label{ex:invariant_monte_carlo} 
	Consider the Monte Carlo $\pi$-approximator $\prog$ from \Cref{fig:intro}.
	The expectation
	\[
	   \exI \eeq \pVarcount + \iv{\pVari \leq \pVarm} \cdot \big(0.85\cdot ((\pVarm \monus \pVari) + 1)\big)
	\]
	is a $\uwpSymb{16}$-superinvariant of $\prog$ w.r.t.\ $\pVarcount$ (here $\monus$ denote the \emph{monus} operator defined as $\synTerm_1 \synMonus \synTerm_2 = \max(\synTerm_1 - \synTerm_2, 0)$, detailed in \Cref{sec:syntax:defs}).
    Hence, we get by \Cref{thm:invariant_cont},
	$\wp{\prog}{\pVarcount} \eleq \exI$,
	i.e., if initially $\pVarcount = 0, \pVari = 1$, then $0.85\cdot \pVarm $ upper-bounds the expected final value of $\pVarcount$. 
\end{example}

\section{Convergence of The Riemann Weakest Pre-Expectations}
\label{sec:convergence}

\newcommand{\discontinuities}{D}

We now address the question if and under which conditions $\lwp{\N}{\prog}{\ex}$ and $\uwp{\N}{\prog}{\ex}$ converge to $\wp{\prog}{\ex}$ as $\N$ goes to $\infty$.
We focus on \emph{pointwise} convergence.
It is easy to see that in general, pointwise convergence does not hold: consider for example the simple program $\prog = \UNIFASSIGN{\pVar}$ and the expectation $\ex = \iv{\pVar \in \rats}$, for which we have $\wp{\prog}{\ex} = \zerofun$.
However, for every $\N \geq 1$ it holds that $\uwp{\N}{\prog}{\ex} = \onefun$.
Our goal is thus to identify a \emph{Riemann-suitable} subset of our framework, by mildly restricting the class of allowed expectations.

We first recall two fundamental properties of the Riemann integral.
The first is a well-known characterization of Riemann integrability, the second reveals that Lebesgue integrals conservatively generalize Riemann integrals.
In the following we will use $\lebmes$ to denote the Lebesgue measure.
For the formal definition and properties of $\lebmes$ we refer the reader to\iftoggle{arxiv}{ \Cref{app:background}}{~\cite{arxiv}}.

\begin{theorem}[Riemann-Lebesgue Theorem\textnormal{~\cite[Thm.~11.33]{rudin1953principles}}]
    \label{thm:riemannIntegrableIffAECont}
    Let $\fun \colon \clIvalGen \to \reals$ be bounded.
    Consider the set $\discontinuities = \{x \in \clIvalGen \mid \fun \text{ is discontinuous in } x\}$.
    Then $\fun$ is Riemann-integrable if and only if $\lebmes(\discontinuities) = 0$, i.e., $\fun$ is \emph{continuous almost everywhere}.
\end{theorem}

\begin{theorem}[\textnormal{~\cite[Theorem 1.7.1]{ash2000probability}}]
    \label{thm:riemannEqualsLebesgue}
    Let $\fun \colon \clIvalGen \to \reals$ be bounded and Riemann-integrable.
    Then
    \iftoggle{arxiv}{%
    \[
        \int_{\ivalL}^{\ivalR} \fun(x) \,dx \eeq \int_{\clIvalGen} \fun(x) \,d\lebmes(x)~,
    \]
    }{%
    $\int_{\ivalL}^{\ivalR} \fun(x) \,dx = \int_{\clIvalGen} \fun(x) \,d\lebmes(x)$,
    }%
    i.e., the Riemann integral coincides with the Lebesgue integral.
\end{theorem}

\subsection{Convergence of Approximation for Loop-free Programs}
\label{sec:approxwp:convloopfree}

By \Cref{thm:riemannIntegrableIffAECont,thm:riemannEqualsLebesgue}, the Riemann integral coincides with the Lebesgue integral for \enquote{sufficiently continuous} functions.
There is thus hope to establish convergence if we restrict our setting to almost everywhere continuous expectations.
However, there is an additional complication: A function needs to be \emph{bounded} on the integration domain in order to be Riemann-integrable -- otherwise the upper sum could be \enquote{stuck} at $\infty$.
We address this issue using the concept of \emph{local boundedness}:
A function $\fun \colon \pStates \to \reals \cup \{ \pm \infty \}$ is called \emph{locally bounded} if
\[
    \text{for all compact } \compactSet \subseteq \pStates \colon \quad \sup_{\pSt \in \compactSet} |\fun(\pSt)| ~<~ \infty
    ~.
\]
We recall that a subset $\compactSet$ of $\reals^\pVars$ is compact iff it is closed and bounded.
Many common functions such as polynomials are locally bounded.
An example of a function that is \emph{not} locally bounded is $\iv{x > 0} \cdot \tfrac 1 x$.
Based on these considerations, we define the following:

\begin{definition}[Riemann-suitable]
	\label{def:riemannSuitable}
	A combination $(\aExps, \guards, \expsClass)$ of a class $\aExps$ of locally bounded arithmetic expressions, a class $\guards$ of guards, and a class $\expsClass \subseteq \exps$ of locally bounded expectations if
	for all $\ex \in \expsClass$ and $\pVar \in \pVars$, $\ex$ is Riemann-integrable on $\uIval$ w.r.t.\ $\pVar$.
	Moreover, $\expsClass$ contains the constant expectations $0, 1$ and satisfies the following closure properties:
	\begin{enumerate}[(i)]
		\item $\expsClass$ is closed under taking infima and suprema over closed subintervals of $\uIval$, i.e., for all $\ex \in \expsClass$, $\pVar \in \pVars$, and $\clIvalGen \subseteq \uIval$ we have
		$\inf_{\xi \in \clIvalGen}\exSubs{\ex}{\pVar}{\xi} \in \expsClass$ and $\sup_{\xi \in \clIvalGen}\exSubs{\ex}{\pVar}{\xi} \in \expsClass$.
		\item
        For all $\ex \in \expsClass$, $\aExp \in \aExps$ and $\pVar \in \pVars$, we have $\exSubs{\ex}{\pVar}{\aExp} \in \expsClass$.
		\item
        For all $\ex, \exb \in \expsClass$ and $p \in \uIval \cap \rats$, we have $p \cdot \ex + (1-p) \cdot \exb \in \expsClass$.
		\item
		For all $\ex, \exb \in \expsClass$ and $\guard \in \guards$, we have $\iv{\guard} \cdot \ex + \iv{\neg\guard} \cdot \exb \in \expsClass$.
        \qedDef
	\end{enumerate}
\end{definition}
The conditions (i) - (iv) in the above definition ensure that $\expsClass$ is closed under $\lwpTrans{\N}{\prog}$ and $\uwpTrans{\N}{\prog}$ for \emph{loop-free} $\prog$ (see \eqref{eq:approxInClass} in \Cref{thm:convLoopFree} below).
We remark that it is not immediately clear if any \enquote{interesting} Riemann-suitable instantiations exist.
Fortunately, the answer is \emph{yes}, as we show in \Cref{sec:syntax}.
In the Riemann-suitable setting, we can show the following convergence result:

\begin{restatable}[Convergence -- Loop-free Case]{lemma}{convLoopFree}
	\label{thm:convLoopFree}
	Let $(\aExps, \guards, \expsClass)$ be Riemann-suitable.
	Then for all loop-free $\prog \in \pWhileWith{\aExps}{\guards}$ and post-expectations $\ex \in \expsClass$ the following holds:
	\begin{align}
        \forall \N \geq 1 \colon\quad \lwp{\N}{\prog}{\ex} \in \expsClass \qand \uwp{\N}{\prog}{\ex} \in \expsClass
		\label{eq:approxInClass} ~.
    \end{align}
    \begin{align}
        \text{If }\ex \in \expsmeas \colon
        \quad 
        \wp{\prog}{\ex}
		\eeq
		\sup_{\N \geq 1} \lwp{\N}{\prog}{\ex}
		\eeq
		\inf_{\N \geq 1} \uwp{\N}{\prog}{\ex}
		~.
		\label{eq:convLoopFree}
	\end{align}
\end{restatable}
\noindent An analogous $\wlpSymb$-version of \Cref{thm:convLoopFree} is stated in\iftoggle{arxiv}{ \Cref{app:wlpVersion}}{~\cite{arxiv}}.

\subsection{Convergence of Approximation for Programs with Loops}
\label{sec:approxwp:convloops}

Recall the \emph{Dirichlet program} $D$ from \Cref{sec:birdseye}.
For all $\N$ we have $\uwp{\N}{D}{\onefun} = \onefun$, showing that \Cref{thm:convLoopFree} does not hold for loopy programs in general. 
The added difficulty stems from the fact that the semantics of loops is itself a limit.
Specifically, we can imagine to \enquote{unroll} every loop in a program $\prog$ up to a certain depth $\depth \in \nats$, denoted $\unfold{\prog}{\depth}$ (see\iftoggle{arxiv}{ \Cref{app:unrolling}}{~\cite{arxiv}} for the formal definition), and view the semantics $\prog$ as the limit of the semantics of $\unfold{\prog}{\depth}$ when $\depth$ tends to infinity.
Notably, $\wpSymb$ and $\wlpSymb$ behave differently with respect to this limit: when increasing $\depth$ the (non-liberal) pre-expectation of $\ex$ increases, while the liberal pre-expectation decreases.
This is formalized in the next lemma:

\begin{restatable}{lemma}{convUnfolding}
	\label{thm:convUnfolding}
	For all $\prog \in \pWhile$ and post-expectations $\ex \in \expsmeas$ and $\exb \in \bexpsmeas$, $\wp{\unfold{\prog}{\depth}}{\ex}$ and $\wlp{\unfold{\prog}{\depth}}{\exb}$ are non-decreasing (non-increasing, respectively) sequences in $\depth \in \nats$ and it holds that
	\[
	\sup_{\depth \in \nats} \wp{\unfold{\prog}{\depth}}{\ex}
	\eeq
	\wp{\prog}{\ex}
	\qqand
	\inf_{\depth \in \nats} \wlp{\unfold{\prog}{\depth}}{\ex}
	\eeq
	\wlp{\prog}{\ex}
	~.
	\]
\end{restatable}

It should now be clear that, when considering the convergence of $\uwp{\N}{\prog}{\ex}$ for a program $\prog$ with loops, we are implicitly considering two limits, the one over $\N$ and the one over $\depth$. Unfortunately, when the two limits have different monotonic behaviors (i.e., for $\uwpTrans{\N}{\prog}$, where the limit in $\depth$ is the supremum but the limit in $\N$ is the infimum, and symmetrically for $\lwlpTrans{\N}{\prog}$) convergence cannot be guaranteed. Intuitively, this explains why the next result only holds for $\lwpSymb{\N}$ and $\uwlpSymb{\N}$. 

\begin{restatable}[Convergence -- General Case]{theorem}{pointwiseConv}
	\label{thm:pointwiseConv}
	Let $(\aExps, \guards, \expsClass)$ be Riemann-suitable.
	Then for all $\prog \in \pWhileWith{\aExps}{\guards}$ and post-expectations $\ex \in \expsClass \cap \expsmeas$ and $\exb \in \expsClass \cap \bexpsmeas$ it holds that:
	\begin{align*}
		\sup_{n \geq 1} \lwp{n}{\unfold{\prog}{n}}{\ex}
		\eeq
		\wp{\prog}{\ex}
		\qand
		\wlp{\prog}{\exb}
		\eeq
		\inf_{n \geq 1} \uwlp{n}{\unfold{\prog}{n}}{\exb}
		~.
	\end{align*}
\end{restatable}

Note that \Cref{thm:pointwiseConv} ensures that when considering a Riemann suitable class, the L.H.S.\ of the inequality in \Cref{thm:cwpSound} converges to the exact conditional weakest pre-expectation $\cwpSymb$.

\section{Effective Verification of $\pWhile$-Programs}
\label{sec:syntax}

In the previous sections, we have assumed that guards and arithmetical expressions in programs as well as the resulting (pre/post)-expectations are purely mathematical objects.
To enable \emph{effective}, i.e., automated, verification we now define, inspired by \cite{DBLP:journals/pacmpl/BatzKKM21,DBLP:journals/pacmpl/SchroerBKKM23}, concrete syntaxes for these objects.

\subsection{A Formal Language of Expressions}
\label{sec:syntax:defs}

Let $\lvars$ be a countably infinite set of logical variables ranged over by $\lvar, \lvarb, \ldots$, etc.
The sets $\syntacticTerms$ and $\syntacticGuards$ of (syntactic) \emph{terms} and \emph{guards} are defined as follows:
\begin{align*}
    \synTerm
    ~\grammarSymb~
    \ratConst \in \nonNegRats
    \mmid 
    \lvar \in \lvars
    \mmid
    \synTerm + \synTerm
    \mmid
    \synTerm \synMonus \synTerm
    \mmid
    \synTerm \cdot \synTerm
    \qquad
    \qquad
    \synGuard
    ~\grammarSymb~
    \synTerm < \synTerm
    \mmid
    \neg \synGuard
    \mmid
    \synGuard \land \synGuard
\end{align*}
Note that the other standard comparison relations $\leq,\neq, =, >, \geq$ and Boolean connectives $\lor, \limplies, \liff$ can be expressed in terms of $<, \neg, \land$.
We allow further syntactic sugar such as $\synTerm^2$ for $\synTerm \cdot \synTerm$ and $1 \leq \lvar \leq 2$  for $1 \leq \lvar \land \lvar \leq 2$, etc.
Additional parentheses are admitted to clarify the order of precedence; to minimize the use of parentheses we assume that $\cdot$ takes precedence over $+$ and $\synMonus$, and $\neg$ binds stronger than the binary Boolean connectives, as is standard.
We let $\varsIn{\synTerm}$ and $\varsIn{\synGuard}$ be the variables occurring in term $\synTerm$ and guard $\synGuard$.

For every $\lvarsSubset \supseteq \free{\synTerm}$, the semantics $\sem{\synTerm} \colon \nnReals^\lvarsSubset \to \nnReals$ of $\synTerm \in \synTerms$ is standard except that $\synMonus$ is interpreted as \enquote{monus}, i.e., $\sem{\synTerm_1 \synMonus \synTerm_2} = \max(\synTerm_1 - \synTerm_2, 0)$.
The semantics of $\synGuard \in \synGuards$ can be viewed similarly as a function $\sem{\synGuard} \colon \nnReals^\lvarsSubset \to \bools$ for every $\lvarsSubset \supseteq \free{\synGuard}$.
For $\st \in \nnReals^\lvarsSubset$, we write $\st \models \synGuard$ and $\st \not\models \synGuard$ to indicate that $\sem{\synGuard}(\st) = \boolConstTrue$ and $\sem{\synGuard}(\st) = \boolConstFalse$, respectively.

Note that every $\synTerm \in \synTerms$ without monus is a non-negative polynomial in the --- finitely many --- variables $\free{\synTerm}$ with rational coefficients, possibly written in (partially) factorized form, whereas terms with monus can be seen as a piecewise defined non-negative polynomial.
Similarly, every $\synGuard \in \synGuards$ is a Boolean combination of (in)equations between such (piecewise) polynomials.

We are now ready to define our expression language similar to~\cite{DBLP:journals/pacmpl/BatzKKM21}.
The restriction to infima and suprema over \emph{compact} intervals $\clIvalGen$ resembles the definition of lower and upper Riemann sums.
\begin{definition}[The Expression Language $\synExps$]
    The set of $\synExps$ of (syntactic) \emph{expressions} is defined according to the following grammar:
    \begin{align*}
        \synEx
        &\quad\grammarSymb\quad
        \synTerm
        \mmid
        \iv{\synGuard} \cdot \synEx
        \mmid
        \ratConst \cdot \synEx
        \mmid
        \synEx + \synEx
        \mmid
        \synSupBd{\lvar}{\ivalL}{\ivalR} \,\synEx
        \mmid
        \synInfBd{\lvar}{\ivalL}{\ivalR} \,\synEx
    \end{align*}
    Here, $\synTerm \in \syntacticTerms$, $\synGuard \in \synGuards$, $\ratConst \in \nnRats$, $\ivalL, \ivalR \in \nnRats$, $\ivalL \leq \ivalR$, and $\lvar \in \lvars$.
    \qedDef
\end{definition}
We allow $\iv{\synGuard}$ as syntactic sugar for $\iv{\synGuard} \cdot 1$ and adopt the usual rules regarding parentheses and orders of precedence.
For $\synEx \in \synExps$ we let $\free{\synEx}$ be the variables in $\synEx$ that are not bound by a $\supQuantifierSymbol$ (supremum) or $\infQuantifierSymbol$ (infimum) \enquote{quantifier}.
As in standard first-order logic, it is possible for a variable to have both a free and a non-free occurrence in $\synEx$.
We may write $\synEx(\lvar_1,\ldots,\lvar_n)$ to indicate that $\synEx$ contains at most the pairwise distinct free variables $\lvar_1,\ldots,\lvar_n \in \lvars$.

\begin{definition}[Semantics of Expressions]
    \label{def:exprSem}
    We define the semantics of expressions $\sem{\synEx} \colon \nnReals^\lvarsSubset \to \nonNegReals$ for every finite $\lvarsSubset \subseteq \lvars$ with $\free{\synEx} \subseteq \lvarsSubset$ inductively as follows.
    For all $\st \in \nnReals^\lvarsSubset$:
    \begin{itemize}
        \item $\sem{\synTerm}(\st)$ is defined in the standard manner, see above.
        \item $\sem{\iv{\synGuard} \cdot \synEx}(\st) = \sem{\synEx}(\st)$ if $\st \models \synGuard$; $\sem{\iv{\synGuard} \cdot \synEx}(\st) = 0$ if $\st \not\models \synGuard$.
        \item $\sem{\ratConst \cdot \synEx}(\st) = \ratConst \cdot \sem{\synEx}(\st)$.
        \item $\sem{\synEx_1 + \synEx_2}(\st) = \sem{\synEx_1}(\st) + \sem{\synEx_2}(\st)$
        \item $\sem{\synSupBd{\lvar}{\ivalL}{\ivalR} \,\synEx}(\st) = \sup_{\xi \in \clIvalGen} \sem{\synEx}(\pStUpdate{\st}{\lvar}{\xi})$.
        \item $\sem{\synInfBd{\lvar}{\ivalL}{\ivalR} \,\synEx}(\st) = \inf_{\xi \in \clIvalGen} \sem{\synEx}(\pStUpdate{\st}{\lvar}{\xi})$.
    \end{itemize}
    In the last two cases, if $\lvar$ is not already in the domain of $\st$, we tacitly assume that the operation $\pStUpdate{\st}{\lvar}{\xi}$ extends the domain of $\st$ by $\lvar$.
    \qedDef
\end{definition}

It is not immediately clear that $\sem{\synEx}$ is well-defined due to the suprema in \Cref{def:exprSem} --- we have to ensure that those exist in $\nonNegReals$.
We show this now.

\begin{restatable}{lemma}{expWellLocallyBounded}
    \label{thm:expWellLocallyBounded}
    For every $\synEx \in \synExps$ and every finite $\lvarsSubset \subseteq \lvars$ with $\free{\synEx} \subseteq \lvarsSubset$ the semantics $\sem{\synEx} \colon \nnReals^\lvarsSubset \to \nonNegReals$ is a well-defined locally bounded function (cf.~\Cref{sec:approxwp:convloopfree}).
\end{restatable}

\begin{example}
    \label{ex:syntax}
    $\synEx = \synSupBd{\lvar}{0}{5} \iv{\lvard = \lvar^2} \cdot \lvar$ denotes the function $\sem{\synEx} = \lam{\st}{\iv{0 \leq \st(\lvard) \leq 25} \cdot \sqrt{\st(\lvard)}}$.
\end{example}

\subsection{Properties of Syntactic Expressions}
\label{sec:syntax:properties}

The set $\firstOrderFormulas$ of first-order formulae\footnote{In this paper, \emph{first-order formula} always refers to first-order logic over the signature of polynomial arithmetic with comparison; we interpret such formulae over the fixed structure $\reals$.} is defined according to the following grammar:
\begin{align*}
    \foForm
    \quad\grammarSymb\quad
    \synGuard
    \mmid
    \foExists \lvar \colon \foForm
    \mmid 
    \foForall \lvar \colon \foForm
    \mmid
    \neg \foForm
    \mmid
    \foForm \land \foForm
    \mmid
    \foForm \lor \foForm
    \mmid
    \foForm \limplies \foForm
\end{align*}
where $\synGuard$ is a Boolean combination of polynomial (in)equalities (similar to the guards defined in \Cref{sec:syntax:defs} but with standard minus instead of \enquote{monus}), and $\lvar \in \lvars$.
As usual, $\free{\foForm}$ is the set of free variables in $\foForm$, and we write $\foForm(\lvar_1,\ldots,\lvar_n)$ to indicate that $\foForm$ contains at most the (pairwise distinct) free variables  $\lvar_1,\ldots,\lvar_n \in \lvars$.
An FO formula $\foForm$ is called \emph{quantifier-free} if it does not contain any $\foExists$ and $\foForall$ quantifiers.
Moreover, $\foForm$ is said to be in \emph{prenex normal form} (PNF) if $\foForm = \foQuantifierSymbol_1\lvar_1 \colon \ldots \foQuantifierSymbol_n\lvar_n \colon \foForm'$ where $n \geq 0$, $\foQuantifierSymbol_1, \ldots, \foQuantifierSymbol_n \in \{\foExists, \foForall\}$, and $\foForm'$ is quantifier-free.
Similarly, $\foForm$ is in \emph{existential (universal) prenex form} if for all $1 \leq i \leq n$ we have $\foQuantifierSymbol_i = \foExists$ ($\foQuantifierSymbol_i = \foForall$, respectively).

The semantics of a formula $\foForm \in \foForms$ is standard and can be viewed as a function $\sem{\foForm} \colon \reals^\lvarsSubset \to \bools$ for all $\lvarsSubset \subseteq \lvars$ with $\free{\foForm} \subseteq \lvarsSubset$.
For $\st \in \reals^{\lvarsSubset}$ we write $\st \models \foForm$ to indicate that $\st$ is a model of $\foForm$, i.e., the \emph{sentence} (formula without free variables) obtained by substituting every free occurrence of $\lvar \in \free{\foForm}$ in $\foForm$ by $\st(\lvar)$ evaluates to $\boolConstTrue$ in $\reals$.
A formula $\foForm$ is called  \emph{satisfiable} if it has a model, \emph{unsatisfiable} if it does not have a model, and \emph{valid} if $\neg\foForm$ is unsatisfiable.

The decision problem of checking whether a given \emph{quantifier-free} $\foForm \in \foForms$ is satisfiable is called $\QFNRA$\footnote{Quantifier-free non-linear real arithmetic, a term coined by the SMT-community~\cite{BarFT-SMTLIB}.}.
It is known that $\QFNRA \in \PSPACE$.
Checking the satisfiability of general $\foForms$ formulae is decidable as well, but has higher complexity.

\begin{restatable}[{$\synExps$ to $\foForms$}]{lemma}{expToFo}
    \label{thm:expToFo}
    For every $\synEx(\lvar_1,\ldots,\lvar_n) \in \synExps$ there exists an FO formula $\foForm_\synEx(\lvar_1,\ldots,\lvar_n,\lvarb)$ encoding $\sem{\synEx}$ in the sense that for all $\st \in \nnReals^\free{\synEx}$ and $\resVal \in \nnReals$, we have that $\st,\resVal \models \foForm_\synEx$ iff $\sem{\synEx}(\st) = \resVal$.
    For quantifier-free $\synEx$, we can construct $\foForm_\synEx$ in existential prenex form in linear time.
\end{restatable}

FO-expressible functions in the sense of \Cref{thm:expToFo} have been extensively studied~\cite{coste2000introduction}.
Let $\foForms_\reals$ be defined like $\foForms$ with the only difference that \emph{arbitrary real numbers} are allowed as coefficients in the polynomials.
A set $\saSet \subseteq \reals^n$, $n \geq 1$, is called \emph{semi-algebraic} if there exists an $\foForms_\reals$ formula $\foForm(\lvar_1,\ldots,\lvar_n)$ such that $\saSet = \{ \st \in \reals^n \mid \st \models \foForm\}$.
A function $\fun \colon \reals^n \to \reals$ is called semi-algebraic if its graph $\{ ( \st, \fun(\st) ) \mid \st \in \reals^n \}$ is a semi-algebraic set.

\begin{lemma}[\textnormal{e.g.,~\cite[Exercise 2.22]{coste2000introduction}}]
    \label{thm:semiAlgebraicFunctionAlmostContinuous}
    Every semi-algebraic function $\fun \colon \reals \to \reals$ has at most finitely many discontinuities.
    In particular, $\fun$ is almost everywhere continuous and hence Riemann-integrable on every interval $\clIvalGen \subset \reals$ where $\fun$ is bounded.
\end{lemma}

By combining local boundedness from \Cref{thm:expWellLocallyBounded}, the translation to FO from \Cref{thm:expToFo}, and the fact that semi-algebraic functions are Riemann-integrable (\Cref{thm:semiAlgebraicFunctionAlmostContinuous}), we obtain:

\begin{restatable}{theorem}{expBoundedAndRiemannIntegrable}
    \label{thm:expBoundedAndRiemannIntegrable}
    For all $\synEx \in \synExps$, intervals $\clIvalGen \subset \nnReals$, finite $\lvarsSubset \subseteq \lvars$ with $\free{\synEx} \subseteq \lvarsSubset $, and variables $\lvar \in \lvarsSubset$, the function $\sem{\synEx} \colon \nnReals^\lvarsSubset \to \nonNegReals$ is Riemann-integrable on $\clIvalGen$ w.r.t.\ $\lvar$.
\end{restatable}

In analogy to standard FO, we say that an expression $\synEx \in \synExps$ is in \emph{prenex normal form} (PNF) if
\begin{align*}
    \synEx \eeq \synVarQuantBd{\lvar_1}{\ivalL_1}{\ivalR_1}{\quantitativeQuantifierSymbol^{\!1}} \ldots \synVarQuantBd{\lvar_n}{\ivalL_n}{\ivalR_n}{\quantitativeQuantifierSymbol^{\!n}} \synEx'
\end{align*}
where $n \geq 0$, $\quantitativeQuantifierSymbol^{\!1}, \ldots,\quantitativeQuantifierSymbol^{\!n} \in \{\infQuantifierSymbol, \supQuantifierSymbol\}$, and $\synEx'$ is quantifier-free.

\begin{lemma}[Prenex Normal Form\textnormal{~\cite[Lemma 8.1]{DBLP:journals/pacmpl/BatzKKM21}}]
    \label{thm:pnf}
    Every $\synEx \in \synExps$ can be transformed into an equivalent $\synEx' \in \synExps$ in PNF with the same free variables.
    Moreover, if $\synEx$ is $\infQuantifierSymbol$-free ($\supQuantifierSymbol$-free), then $\synEx'$ is $\infQuantifierSymbol$-free ($\supQuantifierSymbol$-free, respectively) as well.
    This transformation can be done in linear time.
\end{lemma}

Next, we formalize one of our key insights:
In order to check a \enquote{\emph{quantitative entailment}} $\sem{\synEx} \eleq \sem{\synExb}$ such that suprema occur only in $\synEx$ and infima only in $\synExb$, it is \emph{not} actually necessary to evaluate them exactly (which amounts to solving optimization problems).
Instead, one can, loosely speaking, replace $\sup$ and $\inf$ by $\foForall$-quantifiers.
This idea is based on the following elementary observation:
For arbitrary sets $A, B \subseteq \reals$, we have $\sup A \leq \inf B$ iff $\forall a \in A , b \in B \colon~  a \leq b$.

\begin{restatable}[Checking Quantitative Entailments\textnormal{, cf.~\cite[Section 5.2]{DBLP:journals/pacmpl/SchroerBKKM23}}]{lemma}{checkEntailment}
    \label{thm:checkEntailment}
    Let
    \begin{align*}
        \synEx \eeq \synSupBd{\lvar_1}{\ivalL_1}{\ivalR_1} \ldots \synSupBd{\lvar_n}{\ivalL_n}{\ivalR_n} \synEx'
        \qqand
        \synExb \eeq \synInfBd{\lvarc_1}{\ivalLb_1}{\ivalRb_1} \ldots \synInfBd{\lvarc_m}{\ivalLb_m}{\ivalRb_m} \synExb'
    \end{align*}
    be $\infQuantifierSymbol$-free ($\supQuantifierSymbol$-free, respectively) expressions in PNF.
    Further, let $\foForm_{\synEx'}(\free{\synEx'},\lvarb_\synEx)$ and $\foForm_{\synExb'}(\free{\synExb'},\lvarb_\synExb)$ be the FO formulae in existential prenex form encoding $\synEx'$ and $\synExb'$ from \Cref{thm:expToFo}.
    Then:
    \begin{align}
        & \sem{\synEx} \eeleq \sem{\synExb} \notag \\
        \text{iff}\quad & 
        \Bigg(
            \bigwedge_{\substack{\lvar \,\in\, \free{\synEx} \\ \cup \, \free{\synExb}}} \lvar {\geq} 0
            ~\land~
            \bigwedge_{i=1}^n \ivalL_i {\leq} \lvar_i {\leq} \ivalR_i
            ~\land~
            \bigwedge_{i=1}^m \ivalLb_i {\leq} \lvarc_i {\leq} \ivalRb_i
            ~\land~
            \foForm_{\synEx'} \land \foForm_{\synExb'}
        \Bigg)
        ~\limplies~
        \lvarb_\synEx \leq \lvarb_\synExb
        \quad\text{is valid}
        \label{eq:smt-query}
    \end{align}
    As a consequence, the decision problem \enquote{ $\sem{\synEx} \eleq \sem{\synExb}$?} for $\synEx$ and $\synExb$ is linear-time reducible to $\QFNRA$.
\end{restatable}

\begin{example}
    Reconsider the expression $\synEx = \synSupBd{\lvar}{0}{5} \iv{\lvard = \lvar^2} \cdot \lvar$ from \Cref{ex:syntax}.
    Suppose we wish to check whether $\sem{\synEx}$ is upper-bounded by the constant function $\sem{\synExb}$, $\synExb = 4$.
    Since $\synEx$ and $\synExb$ have the form required by \Cref{thm:checkEntailment} we can achieve this by considering the FO formula
    \begin{align*}
        \left(
            \lvard \geq 0 \land
            0 \leq \lvar \leq 5 \land (\lvard = \lvar^2 \limplies \lvarb = x) \land (\lvard \neq \lvar^2 \limplies \lvarb = 0) \land \lvarb' = 4
        \right)
        ~\limplies~
        \lvarb \leq \lvarb' 
    \end{align*}
    which is \emph{not} valid as witnessed by the assignment $\{\lvar \mapsto 5, \lvard \mapsto 25, \lvarb \mapsto 5, \lvarb' \mapsto 4\}$.
    Hence $\sem{\synEx} \not\eleq 4$.
\end{example}

\subsection{Decidability and Complexity of $\pWhile$ Verification}

We now assume (w.l.o.g.) that the set of logical variables $\lvars$ used in syntactic expressions contains our fixed set $\pVars$ of program variables.

\begin{definition}[Representable and Syntactic Expectations]
    We define the following terminology:
    \begin{itemize}
        \item An expectation $\ex \in \exps$ is called \emph{representable} if $f = \sem{\synEx}$ for some $\synEx \in \synExps$ with $\free{\synEx} \subseteq \pVars$.
        The set of all representable (1-bounded) expectations is denoted $\repExps$ ($\brepExps$, respectively).
        \item A \emph{syntactic expectation} is an expression $\synEx \in \synExps$ whose free variables are program variables, i.e., $\free{\synEx} \subseteq \pVars$.
        A \emph{1-bounded} syntactic expectation  $\bsyntacticExpectations$ is $\synEx \in \synExps$ such that $\sem{\synEx} \in \boundedExpectations$.
        The set of all 1-bounded syntactic expectations is denoted $\bsyntacticExpectations$.
        \qedDef
    \end{itemize}
\end{definition}
Membership in $\bsyntacticExpectations$ is decidable by encoding $\synEx(\lvar_1,\ldots,\lvar_n)$ as the FO formula $\foForm_\synEx(\lvar_1,\ldots,\lvar_n,\lvarb)$ from \Cref{thm:expToFo} and checking validity of $(\lvar_1 \geq 0 \land \ldots \land \lvar_n \geq 0 \land \foForm_\synEx(\lvar_1,\ldots,\lvar_n, \lvarb)) \limplies \lvarb \leq 1$.

Our next theorem implies that if we restrict our framework to representable expectations and only allow the $\synTerms$ and $\synGuards$ defined in \Cref{sec:syntax:defs} as the arithmetic and Boolean, respectively, expressions in our programs, then we obtain the convergence guarantees from \Cref{sec:convergence}.

\begin{restatable}[Riemann-suitability of $\synExps$]{theorem}{exprRiemannSuitable}
    \label{thm:exprRiemannSuitable}
    The combination $(\expsClassExp, \synTerms, \synGuards)$ with $\synTerms$ and $\synGuards$ as defined in \Cref{sec:syntax:defs} is Riemann-suitable (see~\Cref{def:riemannSuitable}).
\end{restatable}

\subsubsection{Verifying Loop-free Programs}
\label{sec:verify_loop_free}

Since $\synExps$ is effectively closed under weakest (liberal) Riemann pre-expectation (formalized in\iftoggle{arxiv}{ \Cref{proof:computeSyntacticExpression}}{~\cite{arxiv}}), it follows that we can effectively prove upper and lower bounds on weakest (liberal) pre-expectations of loop-free programs:

\begin{restatable}[Bounds on Loop-free $\wpwlpSymb$]{theorem}{boundsOnLoopFreeWpWlp}
    \label{thm:boundsOnLoopFreeWpWlp}
    For all loop-free $\prog \in \pWhileWith{\synTerms}{\synGuards}$ and integers $\N \geq 1$ the following decision problems can be effectively translated%
    \footnote{In polynomial time if we assume $\prog$ to be of constant size and $\N$ given in unary. In general, an expression encoding $\lwp{\N}{\prog}{\sem{\synEx}}$ for a given $\synEx$ can be exponential in the size of $\prog$ --- consider $n$ sequential $\KWIF$-$\KWELSE$ constructs. However, our reduction is of practical interest as $\QFNRA$ is supported by many SMT solvers.}
    to $\QFNRA$:
    Given ...
    \begin{enumerate}
        \setlength\itemsep{0.2em}
        \item ... an $\supQuantifierSymbol$-free $\synEx \in \syntacticExpectations$ and an $\infQuantifierSymbol$-free $\synExb \in \syntacticExpectations$, is \qquad $\sem{\synExb} \eleq \lwp{\N}{\prog}{\sem{\synEx}}$ ~? \\
        \gray{If yes, then $\sem{\synExb} \eleq \wp{\prog}{\sem{\synEx}}$, provided $\sem{\synEx} \in \expsmeas$.}
        \item ... an $\infQuantifierSymbol$-free $\synEx \in \syntacticExpectations$ and an $\supQuantifierSymbol$-free $\synExb \in \syntacticExpectations$, is \qquad $\uwp{\N}{\prog}{\sem{\synEx}} \eleq \sem{\synExb}$ ~? \\
        \gray{If yes, then $\wp{\prog}{\sem{\synEx}} \eleq \sem{\synExb}$, provided $\sem{\synEx} \in \expsmeas$.}
        \item ... an $\supQuantifierSymbol$-free $\synEx \in \bsyntacticExpectations$ and an $\infQuantifierSymbol$-free $\synExb \in \bsyntacticExpectations$, is \quad $\sem{\synExb} \eleq \lwlp{\N}{\prog}{\sem{\synEx}}$ ~? \\
        \gray{If yes, then $\sem{\synExb} \eleq \wlp{\prog}{\sem{\synEx}}$), provided $\sem{\synEx} \in \bexpsmeas$.}
        \item ... an $\infQuantifierSymbol$-free $\synEx \in \bsyntacticExpectations$ and an $\supQuantifierSymbol$-free $\synExb \in \bsyntacticExpectations$, is \quad $\uwlp{\N}{\prog}{\sem{\synEx}} \eleq \sem{\synExb}$ ~? \\
        \gray{If yes, then $\wlp{\prog}{\sem{\synEx}} \eleq \sem{\synExb}$, provided $\sem{\synEx} \in \bexpsmeas$.}
    \end{enumerate}
\end{restatable}
\noindent
Moreover, for fixed $\prog$ and $\synEx$, the size of the syntactic expectation representing $\somewp{\prog}{\sem{\synEx}}$ for $\transSymb \in \{\lwpSymb{N},\uwpSymb{N},\lwlpSymb{N},\uwlpSymb{N}\}$ is in $\mathcal{O}\left( N^k\right)$, where $k$ is the number of $\UNIF$-statements \mbox{in $\prog$.}

Using \Cref{thm:boundsOnLoopFreeWpWlp} and \Cref{thm:cwpSound}, it is straightforward to derive bounds on $\cwpSymb$ as well.

\subsubsection{Verifying Loops}

We now state how to apply the proof rules from \Cref{sec:invariants_and_unrolling} effectively:

\begin{restatable}[Verification of Invariants]{theorem}{verificationInvariants}
    \label{thm:verificationInvariants}
    Let $\prog = \WHILE{\guard}{\progBody} \in \pWhileWith{\synTerms}{\synGuards}$ such that $\progBody$ is loop-free.
    Further, let $\synExI, \synEx \in \synExps$ and $\synExJ,\synExb \in \bsyntacticExpectations$ all be quantifier-free, and let $\N \geq 1$ be an integer.
    Then we can ...
    \begin{enumerate}
        \item ... decide if $\sem{\synExI}$ is a $\uwpSymb{\N}$-superinvariant of $\prog$ w.r.t.\ $\sem{\synEx}$ by a reduction to $\QFNRA$.\\
        \gray{If the superinvariant property holds, then $\wp{\prog}{\sem\synEx} \eleq \sem{I}$, provided $\sem{\synEx} \in \expsmeas$}
        \item ... decide if $\sem{\synExJ}$ is a $\lwlpSymb{\N}$-subinvariant of $\prog$ w.r.t.\ $\sem{\synExb}$ by a reduction to $\QFNRA$.\\
        \gray{If the subinvariant property holds, then $\sem{J} \eleq \wlp{\prog}{\sem\synExb}$, provided $\sem{\synExb} \in \bexpsmeas$.}
    \end{enumerate}
\end{restatable}

\begin{restatable}[Complexity of Verification Problems]{theorem}{complexity}
    \label{thm:complexity}
    The following decision problems are $\coRE$-complete\footnote{I.e., $\Pi_1^0$-complete in the arithmetical hierarchy.}:
    Given an arbitrary (not necessarily loop-free) $\prog \in \pWhileWith{\synTerms}{\synGuards}$ and ...
    \begin{enumerate}
        \item\label{it:problemUpperBoundOnWp} ... $\synEx, \synExb \in \syntacticExpectations$ such that $\sem{\synEx} \in \expsmeas$, does it hold that $\wp{\prog}{\sem{\synEx}} \eleq \sem{\synExb}$\,?
        \item\label{it:problemLowerBoundOnWlp} ... $\synEx, \synExb \in \bsyntacticExpectations$ such that $\sem{\synEx} \in \bexpsmeas$, does it hold that $\sem{\synExb} \eleq  \wlp{\prog}{\sem{\synEx}}$\,?
    \end{enumerate}
\end{restatable}
\noindent The proof of membership in $\coRE$ heavily relies on the convergence results established in \Cref{thm:pointwiseConv}.
The $\coRE$-hardness proofs are standard~\cite{DBLP:conf/mfcs/KaminskiK15} and follow from $\coRE$-hardness of the non-halting problem.
See\iftoggle{arxiv}{ \Cref{proof:complexity}}{~\cite{arxiv}} for details.

\section{Implementation and Case Studies}
\label{sec:case_studies}
We have automated our techniques using a modern deductive verifier for probabilistic programs. We first describe how to integrate our techniques within that verifier in \Cref{sec:case_studies:implementation}. We then describe our case studies, i.e., programs and specifications we have verified, in \Cref{sec:case_studies:case_studies}. Finally, we evaluate our approach empirically on these case studies in \Cref{sec:empirical_eval}.
\subsection{Implementation}
\label{sec:case_studies:implementation}
\toolcaesar\footnote{\url{https://www.caesarverifier.org/}} \cite{DBLP:journals/pacmpl/SchroerBKKM23} is a recent expectation-based automated deductive verifier targeting \emph{discrete} and possibly \emph{nondeterministic\footnote{Here we refer to pure nondeterminism which is to be resolved angelically or demonically.}} probabilistic programs. As input, \toolcaesar takes programs written in an intermediate verification language called $\heyVL$ --- a quantitative analogue of \toolboogie \cite{leinoThisBoogie2008}. $\heyVL$ provides a programmatic means to describe verification conditions such as the validity of quantitative loop invariants. These verification conditions are offloaded to an SMT solver, which enables the semi-automated verification of said discrete probabilistic programs. 
\toolcaesar does \emph{not} support continuous sampling instructions. We will, however, now demonstrate that with our approach based on Riemann sums, \toolcaesar \emph{can} readily be used to verify bounds on expected outcomes of continuous probabilistic programs in a semi-automated fashion.

Our key insight is that our Riemann expectation transformers can be expressed as expectation transformers denoted by \emph{discrete} probabilistic programs featuring angelic (resp.\ demonic) \emph{nondeterministic choices} for $\nonNegReals$-valued intervals. The latter features \emph{are} supported by \toolcaesar: $\heyVL$ supports the discrete fragment of $\pWhile$ and, amongst others, the two statements

\begin{enumerate}
	\item  $\NONDETASSIGN{\pVar}{\aExp_1}{\aExp_2}$, which, on program state $\pState$, assigns a nondeterministically chosen value from the $\nonNegReals$-valued interval $[\aExp_1(\pState), \aExp_2(\pState)]$ to the variable $\pVar$, and
	\item $\DISCRETEUNIFASSIGN{\pVar}{\N}$, which, given a (constant) natural number $\N\geq 1$, samples a value from the  set $\{0,\ldots,\N-1\}$ uniformly at random and assigns the result to the variable $\pVar$.
\end{enumerate}
Notice that the latter statement is syntactic sugar for $\pWhile$ as it can be simulated by binary probabilistic choices. The statement $\NONDETASSIGN{\pVar}{\aExp_1}{\aExp_2}$, on the other hand, is not syntactic sugar as nondeterminism is not supported in $\pWhile$. 

Now, given a (discrete but possibly nondeterministic) program $\prog$, \toolcaesar employs a transformer $\vcTrans{\prog} \colon \exps \to \exps$, which is defined\footnote{We do not need to require $\ex$ to be measurable since $\prog$ is does not contain continuous sampling.
    Moreover, as is standard for deductive verifiers, loops need to be annotated with quantitative loop invariants; see \Cref{ex:caesar_loop} on page \pageref{ex:caesar_loop}.} as in \Cref{tab:original} and where for the additional statements, we have
\begin{align*}
	\vc{\NONDETASSIGN{\pVar}{\aExp_1}{\aExp_2}}{\ex}
	\eeq& \lam{\pState}{\sup_{\xi \in [\aExp_1(\pState), \aExp_2(\pState)]} \ex(\pStUpdate{\st}{\lvar}{\xi})} \\
	\vc{\DISCRETEUNIFASSIGN{\pVar}{\N}}{\ex} \eeq & \frac{1}{\N} \cdot \sum\limits_{i=0}^{\N-1} \exSubs{\ex}{\pVar}{i}~.
\end{align*}
Notice that $\vc{\NONDETASSIGN{\pVar}{\aExp_1}{\aExp_2}}{\ex}$ resolves the nondeterministic choice of $\pVar$ \emph{angelically} by returning the \emph{supremum} of all values obtained from evaluating $\ex$ at $\pState$ where $\pVar$ is set to some value from $[\aExp_1(\pState), \aExp_2(\pState)]$. The demonic counterpart is obtained from replacing $\sup$ by $\inf$.
Now let $\pVarj$ be a fresh program variable and observe that for all $\N \geq 1$ and all $\ex \in \exps$, we have
\[
	\underbrace{\uwp{\N}{\UNIFASSIGN{\pVar}}{\ex}}_{\text{defined in this paper}}
	\eeq
	\underbrace{\vc{\SEQ{\DISCRETEUNIFASSIGN{\pVarj}{\N}}{\NONDETASSIGN{\pVar}{\tfrac{\pVarj}{\N}}{\tfrac{\pVarj + 1}{\N}}}}{\ex}}_{\text{expressible in $\heyVL$ and thus supported by \toolcaesar}}~.
\]
Hence, the upper Riemann pre-expectation of $\UNIFASSIGN{\pVar}$ w.r.t.\ $\ex$ corresponds to the maximal --- under all possible resolutions of the nondeterminism --- expected final value of $\ex$ obtained from (i) sampling one of the $\N$ intervals occurring in the Riemann sums uniformly at random and (ii) assigning to $\pVar$ a nondeterministically chosen value from that sampled interval. \emph{Lower} Riemann pre-expectations can be expressed analogously by resolving the nondeterminism \emph{demonically}.
\begin{example}
	\label{ex:caesar_loop}
	Reconsider the Monte Carlo $\pi$-approximator $\prog$ and the superinvariant $\exI$ of $\prog$ w.r.t.\ $\pVarcount$ from \Cref{ex:invariant_monte_carlo}. We can provide \toolcaesar with the following annotated loop:
		\begin{align*}
	&\INVARIANTANNOTATE{\pVarcount + \iv{\pVari \leq \pVarm} \cdot \big(0.85\cdot ((\pVarm \monus \pVari) + 1)\big)} \\
	&\WHILENOBODY{\pVari \leq \pVarm} \\
	&\qquad \SEQ{\DISCRETEUNIFASSIGN{\pVarj_1}{16}}{\NONDETASSIGN{\pVar}{\tfrac{\pVarj_1}{16}}{\tfrac{\pVarj_1 + 1}{16}}}\,; \\
	&\qquad \SEQ{\DISCRETEUNIFASSIGN{\pVarj_2}{16}}{\NONDETASSIGN{\pVarb}{\tfrac{\pVarj_2}{16}}{\tfrac{\pVarj_2 + 1}{16}}}\,; \\
	&\qquad \ITE{\pVar^2 + \pVarb^2 \leq 1}{\ASSIGN{\pVarcount}{\pVarcount + 1}}{\SKIP}\,; \,
	%
	\ASSIGN{\pVari}{\pVari + 1} \quad \}
\end{align*}
	We can then instruct \toolcaesar to check whether $\exI$ is a superinvariant of this loop w.r.t.\ $\pVarcount$. Since the loop body encodes appropriate upper Riemann pre-expectations, \toolcaesar offloads the quantitative entailments corresponding to \Cref{thm:verificationInvariants} to an SMT solver to check them automatically.
	For lower Riemann pre-expectations and subinvariants (for $\wlpSymb$), \toolcaesar can be used analogously.
\end{example}

\subsection{Case Studies}
\label{sec:case_studies:case_studies}

In what follows, we describe the programs and specifications we have verified using \toolcaesar with the SMT solver \toolzt \cite{DBLP:conf/tacas/MouraB08} as back-end. All case studies (including \Cref{ex:caesar_loop}) have been verified within $12$ seconds on an Apple M2. The precise inputs to \toolcaesar are provided in\iftoggle{arxiv}{ \Cref{app:caesar_inputs}}{~\cite{arxiv}}. Both the verified bounds and the required parition sizes $\N$ where determined manually by guessing the respective weakest pre-expectations and increasing $\N$ until \toolcaesar reports success.

\subsubsection{The Irwin-Hall Distribution}
\begin{figure}[t]
    \small
	\begin{subfigure}[b]{0.4\textwidth}
		\input{case_studies_programs/irwin_hall_standard}
	\end{subfigure}
	\begin{subfigure}[b]{0.58\textwidth}
		\input{case_studies_programs/irwin_hall_conditioning}
	\end{subfigure}
	\caption{Generating the standard Irwin-Hall distribution (left) and a variant with conditioning (right).}
	\label{fig:irwin_hall}
\end{figure}
The Irwin-Hall distribution (parameterized in $\pVarm$) \cite{johnson1995continuous} is the sum of $\pVarm$ independent and identically distributed random variables, each of which is distributed uniformly over $[0,1]$. The loop $\prog$ depicted in \Cref{fig:irwin_hall} (left) models this family of distributions: On each iteration, $\prog$ samples a value for $\pVarb$ uniformly at random and adds the result to the variable $\pVar$. Hence, if initially $\pVari= 1$, $\pVar= 0$, and $\pVarm \in \nats$, then the final distribution of $\pVar$ indeed corresponds to the Irwin-Hall distribution with parameter $\pVarm$.

We now aim to upper-bound the expected final value of $\pVar$, i.e., the expectation of the Irwin-Hall distribution, for \emph{arbitrary} values of $\pVarm$. Towards this end, we employ our encoding from \Cref{sec:case_studies:implementation} and use \toolcaesar to automatically verify that the expectation 
$\exI =  \pVar + \iv{\pVari \leq \pVarm}\cdot \big(1.1 \cdot \tfrac{(\pVarm \monus \pVari) + 1}{2} \big) $
is a $\uwpSymb{10}$-superinvariant of $\prog$ w.r.t.\ $\pVar$. Hence, by \Cref{thm:invariant_cont}, we get $\wp{\prog}{\pVar} \eleq \exI$ and thus
\[
	\text{for initial $\pState$ with $\pState(\pVari) = 1$, $\pState(\pVar) = 0$, $\pState(\pVarm)\in \nats$ with $\pState(\pVarm) \geq 1$:}~
	\wp{\prog}{\pVar}(\pState) \leq  1.1\cdot \frac{\pState(\pVarm)}{2}~.
\]
Now consider a variant of the Irwin-Hall distribution where we condition each of the $\pVarm$ random variables to take a value in $[0,\nicefrac{1}{2}]$.
This situation is modeled by the loop $\prog$ depicted in \Cref{fig:irwin_hall} (right), where we employ conditioning inside of the loop. We aim to upper-bound $\cwp{\prog}{\pVar}$, i.e., the \emph{conditional expected final value of $\pVar$}.
For that, we use \toolcaesar to verify that the following expectation is a $\uwpSymb{19}$-superinvariant of $\prog$ w.r.t.\ $\pVarcount$ (resp.\ $\lwlpSymb{2}$-subinvariant of $\prog$ w.r.t.\ $1$):
\begin{align*}
	\exI \eeq  \pVar + \iv{\pVari \leq \pVarm}\cdot \big(1.5 \cdot \frac{(\pVarm \monus \pVari) + 1}{8} \big)  
	 \qquad \text{and}\qquad 
	 \exJ \eeq  \iv{\pVari \leq \pVarm} \cdot 0.5^{ (\pVarm \monus \pVari) +1 } + \iv{\pVari > \pVarm}\cdot 1~.
\end{align*}
Hence, we get by \Cref{thm:invariant_cont} that $\cwp{\prog}{\pVar}(\pState)$ is well-defined for all $\pState \in \pStates$ and
\[
    \text{for initial $\pState$ with $\pState(\pVari) = 1$, $\pState(\pVar) = 0$, $\pState(\pVarm)\in \nats$ with $\pState(\pVarm) \geq 1$:}~
    \cwp{\prog}{\pVar}(\pState) \leq  \frac{1.5\cdot\pState(\pVarm)}{8 \cdot 0.5^{\pState(\pVarm)}}~.
\]
Notice that, even though exponential functions like $0.5^\pVarm$ are not supported by our decidable language of syntactic expectations from \Cref{sec:syntax}, we can use \toolcaesar to reason about it by means of standard user-defined domain declaration for exponentials (see \cite[Section 5.1]{DBLP:journals/pacmpl/SchroerBKKM23} for details). Decidability of the corresponding entailments involving such domains is, however, not guaranteed.

\subsubsection{Reasoning about Diverging Programs}
\begin{figure}[t]
    \small
    \begin{subfigure}[t]{0.3\textwidth}
        \input{case_studies_programs/diverging}
    \end{subfigure}
    \begin{subfigure}[t]{0.53\textwidth}
        \input{case_studies_programs/tortoise_hare}
    \end{subfigure}
    \caption{A probably diverging loop (left) and a race between tortoise and hare (right) \cite{DBLP:conf/sas/ChakarovS14}.}
    \label{fig:div_and_race}
\end{figure}
Weakest \emph{liberal} pre-expectations enable us to lower-bound divergence probabilities of programs involving continuous sampling. 
Consider the loop $\prog$ depicted in \Cref{fig:div_and_race} (left).
In each iteration, we sample uniformly from the interval $[\pVaraa, \pVarbb]$ (notice that $\pVaraa$, $\pVarbb$ are uninitialized variables and hence correspond to parameters) and assign the result to $\pVarb$.
Whenever we sample a value in $[0, \tfrac{\pVaraa + \pVarbb}{2}]$, the program diverges.
Using \toolcaesar, we verify that 
\[
    \exI \eeq \iv{\pVaraa \leq \pVarbb}\cdot (1 \monus 0.5^{\pVar})
\]
is a $\lwlpSymb{2}$-subinvariant of $\prog$ w.r.t.\ $0$. Hence, we have $\exI \eeleq \wlp{\prog}{0}$ by \Cref{thm:invariant_cont}, i.e., whenever $\pVaraa \leq \pVarbb$, then $\prog$ diverges on initial state $\pState$ with probability at least $1 - 0.5^{\pState(\pVar)}$.

\subsubsection{Race between Tortoise and Hare}
The loop $\prog$ depicted in \Cref{fig:div_and_race} (right) is adapted from \cite{DBLP:conf/sas/ChakarovS14} and models a race between a tortoise ($\pVart$) and a hare ($\pVarh$).
As long as the hare did not overtake the tortoise, the hare flips a fair coin to decide whether to move or not. If the hare decides to move, it samples a distance uniformly at random from $[0,10]$.
The tortoise always moves exactly one step.

We now upper-bound the expected final value of variable $\pVarcount$. Towards this end, we use \toolcaesar to verify that the expectation
$\exI = \pVarcount + \iv{\pVarh \leq \pVart} \cdot 3.012\cdot\big((\pVart \monus \pVarh) +2\big)$
is a $\uwpSymb{16}$-superinvariant of $\prog$ w.r.t.\ $\pVarcount$. Hence, we get $\wp{\prog}{\pVarcount} \eleq \exI$ by \Cref{thm:invariant_cont}, i.e., if $\prog$ is executed on an initial state $\pState$ with $\pState(\pVarcount) = 0$ and $\pState(\pVarh) \leq \pState(\pVart)$, then the expected final value of $\pVarcount$ is at most $3.012 \cdot (\pState(\pVart) - \pState(\pVarh) + 2)$.

\subsection{Experimental Evaluation}
\label{sec:empirical_eval}
\begin{table}
	\caption{Experimental results. Time in seconds. TO = \SI{180}{s}, MO = \SI{8}{GB}.}
	\label{tab:experiments}
    \renewcommand{\arraystretch}{1.0}
	\begin{tabular}{l | c r@{\quad}@{\quad} r@{\quad}@{\quad} c r}%
		\toprule
		\colprog & \coln & \colt & \colformsize & \colpost & \colbound \\
		\midrule
		\multirow{5}{*}{\progmontecarlo~(Fig.\ \ref{fig:intro})} & $2$ & $0.03$ & $552$ & \multirow{5}{*}{$\pVarcount$}  & $1.000\cdot \pVarm$ \\
& $4$ & $0.05$ & $1634$ & & $0.938\cdot \pVarm$ \\
& $8$ & $0.38$ & $5448$ & & $0.875\cdot \pVarm$ \\
& $16$ & $10.05$ & $19676$ & & $0.837\cdot \pVarm$ \\
& $32$ & TO & $74526$ &  & -- \\
\midrule
\multirow{6}{*}{\progirwin~(Fig.\ \ref{fig:irwin_hall} (left))} & $2$ & $0.02$ & $187$ & \multirow{6}{*}{$\pVar$} & $0.750\cdot \pVarm$ \\
& $4$ & $0.02$ & $309$ & & $0.625\cdot \pVarm$ \\
& $8$ & $0.03$ & $553$ & & $0.565\cdot \pVarm$ \\
& $10$ & $0.08$ & $716$ & & $0.55\cdot \pVarm$ \\
 & $16$ & $2.15$ & $1041$ & & $0.532\cdot \pVarm$ \\
& $32$ & TO & $2017$ & & -- \\
\midrule
\progdiverging~(Fig.\ \ref{fig:div_and_race} (left)) & $2$ & $0.22$ & $162$ & $0$ & $\iv{\pVaraa \leq \pVarbb} \cdot (1 \monus 0.5^{\pVar})$ \\
\midrule
\multirow{5}{*}{\progtortoise~(Fig.\ \ref{fig:div_and_race} (right)) } & $2$ & $0.02$ & $2$ &  \multirow{5}{*}{$\pVarcount$}  & $\infty$ \\
& $4$ & $.03$ & $400$ & & $8\cdot (\pVart \monus \pVarh + 2)$ \\
& $8$ & $.015$ & $680$ & & $3.38\cdot (\pVart \monus \pVarh + 2)$ \\
& $16$ & $100.09$ & $1240$ & & $3.012\cdot (\pVart \monus \pVarh + 2)$ \\
& $32$ & TO & $2160$ &  &  --
\\
        \bottomrule
	\end{tabular}
\end{table}
In this section, we empirically evaluate the scalability of our approach based on the programs from the previous sections, i.e., our case studies and the Monte-Carlo approximator from \Cref{fig:intro}.
For each of these programs and post-expectations, we verify the tightest upper bound (up to a precision of $3$ decimal places) on the respective weakest pre-expectation for increasing values of $N$ (except for \progdiverging~where $N=2$ already yields the precise weakest pre-expectation).

Our empirical results are depicted in \Cref{tab:experiments}.
Column \colprog~depicts the respective program, \coln~denotes the partition size for the Riemann pre-expectation, \colt~denotes the verification time in seconds (i.e., verification condition generation and time for SMT solving) required by \toolcaesar, \colformsize~denotes the size of the formula offloaded to the SMT solver \toolzt (i.e., the number of nodes in the formula's abstract syntax tree), \colpost~denotes the post-expectation, and \colbound~depicts the upper (lower) bound on the (liberal) weakest pre-expectation. For the \progdiverging~benchmark, we lower-bound a liberal weakest pre-expectation.
All other benchmarks upper-bound a non-liberal weakest pre-expectation.
The bounds were inferred manually in a binary search-like fashion.

With our encoding of Riemann pre-expectations described in \Cref{sec:case_studies:implementation}, \toolcaesar allows verifying increasingly sharper bounds when increasing the partition size $N$
The size of the generated formula for the verification conditions mostly increases moderately when increasing $N$.
This is to be expected since the formula size grows polynomially in $N$ for a fixed program and a fixed post-expectation (see \Cref{sec:verify_loop_free}).
The time required for verification, however, can increase drastically when doubling the value of $N$ (see entries of \Cref{tab:experiments} where $N=32$).

\section{Related Work}
\label{sec:relwork}

\emph{Integral Approximation.}
Our approach relies on under- and over-approximation of Lebesgue integrals via Riemann sums, achieved by discretizing the domain of the continuous uniform distributions in the program.
Recently, \cite{DBLP:journals/pacmpl/GargHBM24} proposed an efficient \emph{bit blasting} discretization method for continuous mixed-gamma distributions (which subsume the uniform distributions) for programs with bounded ($\mathtt{for}$) loops.
In \cite{DBLP:conf/pldi/BeutnerOZ22}, integral approximation is used to derive guaranteed bounds on a program's posterior distribution.
Following up on this, \cite{wang2024static} extended the computation of such bounds to a class of programs with soft conditioning (\emph{score-at-end} programs).
Previously, \cite{DBLP:conf/pldi/SankaranarayananCG13} and \cite{DBLP:journals/pacmpl/AlbarghouthiDDN17} used integral approximations to bound the satisfaction probabilities of properties like safety and fairness.
However, these approaches are not able to deal with unbounded loops in full generality.
AQUA \cite{DBLP:journals/isse/HuangDM22} uses a different discretization method, where the posterior of a given program is approximated by quantizing the state space.
However, this approximation does not necessarily provide an upper or lower bound of the true distribution, and unbounded loops are not supported.
Finally, \cite{DBLP:conf/pldi/Wang0R21} propose a method to derive upper and lower bounds on central moments of cost accumulators.

\emph{Loop Invariant Analysis.} 
The main purpose of the current work is to enable automatic verification of inductive invariants for continuous probabilistic programs. Reasoning about probabilistic loops via quantitative invariants was first introduced by McIver and Morgan \cite{DBLP:series/mcs/McIverM05} for probabilistic programs with discrete probabilistic choice and extended by various authors to different settings \cite{DBLP:conf/setss/SzymczakK19,DBLP:journals/pe/GretzKM14,DBLP:conf/lics/KaminskiK17}.
A generalization of this framework to programs featuring continuous distribution was proposed in \cite{DBLP:conf/cav/ChakarovS13} and uses super-martingale ranking functions as the probabilistic analogue of classic ranking functions.
The super-martingale approach has been specialized to qualitative and quantitative termination problems for various classes of programs such as polynomial programs \cite{DBLP:conf/cav/ChatterjeeFG16}, affine programs admitting a linear ranking super-martingale \cite{DBLP:conf/popl/ChatterjeeFNH16}, and non-deterministic probabilistic programs \cite{DBLP:conf/popl/ChatterjeeNZ17,DBLP:journals/pacmpl/AgrawalC018}.
Besides termination, this approach has been successfully applied also to cost analysis \cite{DBLP:conf/pldi/Wang0GCQS19} and equivalence refutation \cite{DBLP:journals/pacmpl/ChatterjeeGNZ24}.
While building on the martingale approach, \cite{DBLP:conf/cav/ChatterjeeGMZ22} overcomes the necessity to compute ranking martingales, by introducing the notion of stochastic invariants indicators, that can be used to prove termination in a sound and relatively complete way.
A different generalization was proposed in \cite{DBLP:conf/sas/ChakarovS14}, where expectation invariants are proposed to introduce a generalized fixed-point framework for programs with continuous distributions and unbounded loops.
However, this analysis rules out invariants expressed as Iverson brackets and is restricted to piecewise linear assertion guards and updates. Similarly, \cite{DBLP:conf/pldi/WangS0CG21}, restricts to affine programs and invariants to find exponential bounds on assertion violation probabilities.

\emph{Exact Inference.}
Exact inference on probabilistic programs involving continuous distributions is especially hard because of the challenge of computing densities via integrals.
Exact symbolic inference engines like \cite{DBLP:conf/cav/GehrMV16, DBLP:conf/pldi/GehrSV20, narayanan2016probabilistic, DBLP:conf/pldi/SaadRM21} only support deterministically bounded loops.
Even within this constrained context, the occurrence of non-simplified integrals in the output can significantly impede the potential to perform appropriate analysis or further computations.
A different, yet still exact approach is pursued in \cite{DBLP:journals/pacmpl/MoosbruggerSBK22}, where the focus is on the computation of moments as functions of the loop iteration.
In order to compute the moments exactly, the syntax of the programs is restricted quite severely, though parametric models are supported. Under the same syntactic restrictions, it has been shown in \cite{DBLP:journals/fmsd/MoosbruggerBKK22} that termination probabilities can be computed automatically.  

\section{Conclusion}
\label{sec:conclusion}

We tackled the semi-automatic verification of user-provided quantitative invariants for probabilistic programs involving continuous distributions, conditioning, and potentially unbounded loops.
Our approach is to approximate the Lebesgue integrals appearing in a program's weakest pre-expectation semantics by lower and upper Riemann sums.
This simple idea allows us to prove sound bounds on expected outcomes by verifying invariants automatically with SMT-solvers.
We demonstrated that this method can be readily integrated in the \toolcaesar verification infrastructure~\cite{DBLP:journals/pacmpl/SchroerBKKM23}.

There are many opportunities for future work.
A natural direction is to integrate our approach to invariant verification with techniques for \emph{fully automatic invariant generation} such as~\cite{DBLP:conf/tacas/BatzCJKKM23}.
While we can already verify invariants for non-trivial benchmarks, there is some room for improvement on the engineering side.
This includes more sophisticated, multi-dimensional, and adaptive partitions of the integration domain, and a combination with direct integration methods, where applicable.
We also plan to extend the front end of \toolcaesar with dedicated support for continuous distributions.

Due to its simplicity, we are confident that our approach is amenable to generalizations such as non-determinism and strategy synthesis as in~\cite{DBLP:journals/pacmpl/BatzBKW24}, and programs with pointers and data structures.

\begin{acks}
    The authors thank Philipp Schroer for his support with \toolcaesar, Mirco Tribastone for his help with the initial conceptualization of the paper, and the anonymous referees for their detailed comments.
    This work was partially funded by the \grantsponsor{DFG}{DFG}{https://www.dfg.de} \grantnum{DFG}{GRK 2236 UnRAVeL}, the EU's Horizon 2020 programme under the Marie Skłodowska-Curie grant No.\ 101008233 (MISSION), the ERC AdG No.\ 787914 (FRAPPANT), and the PNRR project iNEST (Interconnected Nord-Est Innovation Ecosystem) funded by the European Union Next-GenerationEU (Piano Nazionale di Ripresa e Resilienza (PNRR) – Missione 4 Componente 2, Investimento 1.5 – D.D. 1058 23/06/2022, ECS\_00000043).
\end{acks}

\iftoggle{arxiv}{}{%
\section*{Data Availability Statement}
A software artifact comprising the \toolcaesar verifier, the encodings of the case studies from \Cref{sec:case_studies}, and instructions for running the experiments is available~\cite{zenodo}.
}
\bibliographystyle{ACM-Reference-Format}
\bibliography{references}

\iftoggle{arxiv}{
\newpage
\appendix
\section*{\centering Appendix}
\bigskip
\allowdisplaybreaks
\section{Additional Preliminary Material}
\label{app:background}
\subsection{Fixed Point Theory}

The following observation is key in several proofs in this paper:
\begin{restatable}[\textnormal{cf.~\cite[Proposition 2.1.12]{abramsky1994domain}}]{lemma}{oneVsTwoSups}
    \label{thm:oneVsTwoSups}
    Let $\poGen$ be an $\omega$-cpo and let $\poElem \colon \nats \times \nats \to \poDom$ be monotonic (with regards to the usual order on $\nats$ lifted pointwise to $\nats \times \nats$).
    Then
    \[
        \sup_{i \in \nats} \sup_{j \in \nats} \poElem(i,j)
        \eeq
        \sup_{i \in \nats} \poElem(i,i)
        \eeq
        \sup_{j \in \nats} \sup_{i \in \nats} \poElem(i,j)
        ~.
    \]
\end{restatable}
\begin{proof}
    \Cref{proof:oneVsTwoSups}.
\end{proof}

\subsection{Measure Theory and Lebesgue Integrals}
\label{sec:prelims:measure}

\subsubsection{Basic Notions of Measure Theory}
\label{sec:prelims:measure:basics}

Most of the following can be found in~\cite{ash2000probability}.
Let $\mUniv$ be a set.
A \proseSigmaAlgebra on $\mUniv$ is a set $\sAlg \subseteq 2^\mUniv$ of subsets of $\mUniv$ such that $\mUniv \in \sAlg$ and $\sAlg$ is closed under taking countable unions and complements relative to $\mUniv$.
If $\sAlg$ is a \proseSigmaAlgebra on $\mUniv$ then $(\mUniv, \sAlg)$ is called a \emph{measurable space} and the sets in $\sAlg$ are called \emph{measurable sets} (this is not to be confused with a \emph{measure space} which is introduced further below).
For an arbitrary set $\sigmaAlgebraGenerators \subseteq 2^\mUniv$ of subsets of $\mUniv$ we let $\sigmaAlgebraGeneratedBy{\sigmaAlgebraGenerators} = \bigcap \{\mathcal F \mid \sigmaAlgebraGenerators \subseteq \sAlg, \sAlg\text{ a \proseSigmaAlgebra on } \mUniv \}$ be the smallest \proseSigmaAlgebra (on $\mUniv$) containing $\sigmaAlgebraGenerators$.

The \emph{Borel \proseSigmaAlgebra} $\borelSets{\exReals}$ on the extended real numbers $\exReals$ is the smallest \proseSigmaAlgebra containing all intervals of the form $(\ivalL, \ivalR]$, $\ivalL, \ivalR \in \exReals$, i.e., $\borelSets{\exReals} = \sigmaAlgebraGeneratedBy{\{(\ivalL, \ivalR] \mid \ivalL, \ivalR \in \exReals\}}$.
The sets in $\borelSets{\exReals}$ are called \emph{Borel sets}.
Note that many common sets including various types of intervals (bounded, unbounded, open, closed, half-open, and so on), all countable subsets of $\exReals$, and many others are Borel sets.
However, $\borelSets{\exReals} \neq 2^\exReals$, i.e., there exist non-Borel sets.
For $\measurableSet \in \borelSets{\exReals}$ we define $\borelSets{\measurableSet} = \borelSets{\exReals} \cap \measurableSet$, where the intersection is taken element-wise.
It can be shown that $\borelSets{\measurableSet}$ is a \proseSigmaAlgebra on $\measurableSet$ and $\borelSets{\measurableSet} \subseteq \borelSets{\exReals}$, see, e.g.,~\cite[page 5]{ash2000probability}.
In this way we may obtain a \proseSigmaAlgebra on, say, $\uIval$.

Given a measurable space $(\mUniv, \sAlg)$, a \emph{measure} on $\sAlg$ is a function $\measure \colon \sAlg \to \exNonNegReals$ such that for all countable collections $\countableCollectionOfMeasurableSets \subseteq \sAlg$ of pairwise disjoint subsets of $\sAlg$ (i.e., $\forall \measurableSet,  \measurableSetb \in \countableCollectionOfMeasurableSets \colon \measurableSet = \measurableSetb \lor \measurableSet \cap \measurableSetb = \emptyset$) it holds that
$\measure\left( \bigcup_{\measurableSet \in \countableCollectionOfMeasurableSets} \measurableSet \right) = \sum_{\measurableSet \in \countableCollectionOfMeasurableSets} \measure(\measurableSet)$, where the infinite sum is allowed to take value $\infty$.
If $\measure(\mUniv) = 1$ then $\measure$ is a \emph{probability measure}.
We denote by $\lebmes$ the \emph{Lebesgue measure}%
\footnote{We remark that it is also common to define $\lebmes$ on $\lebesgueSets{\reals}$, the \emph{completion} of $\borelSets{\reals}$ w.r.t.\ $\lebmes$ (the sets in $\lebesgueSets{\reals}$ are often called \emph{Lebesgue measurable}), but this construction is not necessary for our purposes.}
on $\borelSets{\reals}$
is the unique measure $\lebmes$ satisfying $\lebmes((\ivalL, \ivalR]) = \ivalR - \ivalL$ for all $\ivalL \leq \ivalR \in \reals$.
Note that we define this only on the reals as we do not need it for the extended reals.
In fact, the only measure we consider (explicitly) in this paper is the Lebesgue measure on $\borelSets{\uIval}$.

\subsubsection{Lebesgue Integrals}
\label{sec:prelims:measure:integrals}

Let $(\mUniv_i, \sAlg_i)$, $i \in \{1,2\}$, be measurable spaces.
A function $\fun \colon \mUniv_1 \to \mUniv_2$ is called \emph{measurable} w.r.t.\ $\sAlg_1$ and $\sAlg_2$ if for all $\measurableSet \in \sAlg_2$ it holds that $\fun^{-1}(\measurableSet) \in \sAlg_1$.

For our purpose we only need to define integrals of non-negative functions.
Let $(\mUniv, \sAlg, \measure)$ be a measure space and $\fun \colon \mUniv \to \exNonNegReals$ be an arbitrary function.
A measurable function $\simpleFun \colon \mUniv \to \nonNegReals$ w.r.t.\ $\sAlg$ and $\borelSets{\nonNegReals}$ is called \emph{simple} if its image $\simpleFun(\mUniv)$ is a finite set $\{a_1,\ldots,a_n\}$; such an $\simpleFun$ can be written as $\simpleFun(x) = \sum_{i=1}^n a_i \cdot \iv{x \in \simpleFun^{-1}(a_i)}$ for all $x \in \mUniv$.
The Lebesgue integral of the simple function $\simpleFun$ on $\measurableSet \in \sAlg$ is defined as 
$
\int_{\measurableSet} \simpleFun \,d\measure
=
\sum_{i=1}^n a_i \cdot \measure(\simpleFun^{-1}(a_i) \cap \measurableSet)
$
.
The Lebesgue integral of $\fun$ on $\measurableSet$ is then defined as
\[
\int_{\measurableSet} \fun \,d\measure
\eeq
\sup \left\{\int_{A} \simpleFun \,d\measure \mid \simpleFun \text{ is simple, } \simpleFun \leq \fun \text{ (pointwise)} \right\}
\]
which can be any number in $\nonNegReals$ or $\infty$.
Note that $\fun$ itself does not have to be measurable; however, many fundamental properties of Lebesgue integrals  break for non-measurable $\fun$, hence it is customary to consider Lebesgue integrals only for measurable $\fun$.

\subsubsection{Multi- vs.\ Single-dimensional Integrals}
\label{sec:prelims:measure:multi}

Given two measurable spaces $(\mUniv_i, \sAlg_i)$, $i \in \{1,2\}$, define $\sAlg_1 \otimes \sAlg_2 = \sigmaAlgebraGeneratedBy{\{\measurableSet_1 \times \measurableSet_2 \mid \measurableSet_i \in \sAlg_i\}}$.
The measurable space $(\mUniv_1 \times \mUniv_2, \sAlg_1 \otimes \sAlg_2)$ is called the \emph{product} of $(\mUniv_1, \sAlg_1)$ and $(\mUniv_2, \sAlg_2)$.
For $n > 1$ and $\measurableSet \in \borelSets{\exReals}$, we obtain the Borel \proseSigmaAlgebra $\borelSets{\measurableSet^n}$ on $\measurableSet^n$ as the $n$-fold product of $(\measurableSet, \borelSets{\measurableSet})$ with itself.
The sets in $\borelSets{\measurableSet^n}$ are called Borel sets as well.

We call a function $\fun \colon \measurableSet \to \measurableSetb$ \emph{Borel measurable} if both $A$ and $B$ are Borel sets (in any dimension) and $\fun$ is measurable w.r.t.\ the respective Borel \proseSigmaAlgebras on $A$ and $B$. 
In this paper, we often consider Borel measurable functions of type $\fun \colon \reals^n \to \exNonNegReals$ and we would like to do something like \enquote{take a Lebesgue integral in one variable}.
To justify that this makes sense (and preserves measurability) we rely on the following:

\begin{theorem}[Part of Fubini's Theorem\textnormal{~\cite[Theorems 8.5 and 8.8]{10.5555/26851}}]
    \label{thm:fubini}
    Let $(\mUniv_i, \sAlg_i)$, $i \in \{1,2\}$, be measurable spaces and $\fun \colon \mUniv_1 \times \mUniv_2 \to \exNonNegReals$ be measurable w.r.t.\ the product space.
    Then:
    \begin{enumerate}
        \item For every $x \in \mUniv_1$, the function $\fun_x \colon \mUniv_2 \to \exNonNegReals, y \mapsto \fun(x,y)$ is measurable w.r.t.\ $\sAlg_2$.
        \item For every $\measurableSet \in \sAlg_2$ and measure $\measure$ on $\sAlg_2$, the function $F \colon \mUniv_1 \to \exNonNegReals, x \mapsto \int_{\measurableSet} \fun_x \,d\measure$ is measurable w.r.t.\ $\mathcal F _1$.
    \end{enumerate}
\end{theorem}
For a given $x \in \mUniv_1$, we also write $\int_{\measurableSet}\fun(x,y) \,d\measure(y)$ instead of $\int_{\measurableSet}\fun_{x} \,d\measure$.

\subsection{Riemann Integrals}

A useful aspect of (lower and upper) Riemann integrals is that one can choose partitions that induce sequences of upper (resp.\ lower) sums that converge \emph{monotonically}.
For partitions $\partition = (x_0,\ldots,x_{\partitionSize})$ and $\partitionb = (y_0,\ldots,y_{\partitionSizeb})$,
 both of the same interval, we say that $\partition$ \emph{refines} $\partitionb$, in symbols $\partition \partitionRefines \partitionb$, if for all indices $0 \leq j \leq \partitionSizeb$ there exists $0 \leq i \leq \partitionSize$ such that $x_i = y_j$.
In other words, the partition $\partition$ can be obtained from $\partitionb$ by subdividing some of the intervals in $\partitionb$ further.
\begin{lemma}[\textnormal{see,~e.g.,~\cite[Lemma 6.2]{fitzpatrick2009advanced}}]
    \label{thm:partitionRefine}
    Let $\fun \colon \clIvalGen \to \reals$ be bounded.
    Suppose that $\partition \partitionRefines \partitionb$.
    Then $\lowerSum{\fun}{\partitionb} \leq \lowerSum{\fun}{\partition}$
    and
    $\upperSum{\fun}{\partition} \leq \upperSum{\fun}{\partitionb}$.
\end{lemma}

To obtain the Riemann-Darboux integral, it is not actually necessary to consider \emph{all} partitions.
Indeed, it suffices to consider a set of partitions whose \emph{norm} (\enquote{fineness}) becomes arbitrarily small.
We now formalize this.
For a partition $\partition = (x_0,\ldots,x_{\partitionSize}) \in \partitions{\clIvalGen}$, let its norm be defined as
\[
\partitionNorm{\partition}
\eeq
\max_{0 \leq i < \partitionSize} x_{i+1} - x_i
~.
\]

\begin{restatable}[\textnormal{see,~e.g.,~\cite[Theorem 7.12]{fitzpatrick2009advanced}}]{theorem}{smallNormSuffices}
    \label{thm:smallNormSuffices}
    Let $\fun \colon \clIvalGen \to \reals$ be bounded and let $(\partition_n)_{n\in\nats}$ be a sequence of partitions of $\clIvalGen$ satisfying $\lim_{n \to \infty} \partitionNorm{\partition_n} = 0$.
    Then \[
    \lowerIntGen \fun(x) \,dx
    \eeq
    \lim_{n \to \infty} \lowerSum{\fun}{\partition_n}
    \eeq
    \sup_{n \in \nats} \lowerSum{\fun}{\partition_n}
    \qqand
    \upperIntGen \fun(x) \,dx
    \eeq
    \lim_{n \to \infty} \upperSum{\fun}{\partition_n}
    \eeq
    \inf_{n \in \nats} \upperSum{\fun}{\partition_n}
    ~.
    \]
\end{restatable}
\begin{proof}
    \Cref{proof:smallNormSuffices}
\end{proof}

\section{Omitted Remarks and Explanations}
\label{app:misc}
\subsection{On Redundancies in the Program Syntax}
\label{app:redudanciesSyntax}

Some of the programming constructs from \Cref{def:pwhile} could be regarded syntactic sugar.
However, this would have some unsatisfactory effects.
For example, $\PCHOICE{\prog_1}{\prob}{\prog_2}$ is equivalent to $\SEQ{\UNIFASSIGN{\pVarTmp}}{\ITE{\pVarTmp < \prob}{\prog_1}{\prog_2}}$, but this has the disadvantages that (i) it requires an additional program variable, (ii) it assumes that $(\pVarTmp < \prob) \in \guards$, and (iii) it replaces a simple convex combination by a way more difficult integration operation, which is a major drawback in practice.
Further, one could define $\DIVERGE$ as syntactic sugar for $\WHILE{true}{\SKIP}$.
However, in the context of this paper, it is more convenient to distinguish between loopy programs and loop-free programs possibly containing $\DIVERGE$.

\subsection{Formal Definition of Loop Unrolling}
\label{app:unrolling}

The following definition appears in \Cref{sec:approxwp:convloops}.
Let $\depth \in \nats$.
We define the \emph{$\depth$-fold unrolling} $\unfold{\prog}{\depth}$ of program $\prog$ inductively as follows:
For all atomic programs $\prog$, we let $\unfold{\prog}{\depth} = \prog$.
Further, we set (using some parentheses -- which are of course not part of the program syntax -- for clarity)
\begin{itemize}
    \item $\unfold{(\SEQ{\prog_1}{\prog_2})}{\depth} = \SEQ{\unfold{\prog_1}{\depth}}{\unfold{\prog_2}{\depth}}$,
    \item $\unfold{(\ITE{\guard}{\prog_1}{\prog_2})}{\depth} = \ITE{\guard}{\unfold{\prog_1}{\depth}}{\unfold{\prog_2}{\depth}}$, and
    \item $\unfold{(\PCHOICE{\prog_1}{\prob}{\prog_2})}{\depth} = \PCHOICE{\unfold{\prog_1}{\depth}}{\prob}{\unfold{\prog_2}{\depth}}$.
\end{itemize}
For loops, we let
\begin{itemize}
    \item $\unfold{(\WHILE{\guard}{\progBody})}{0} = \DIVERGE$, and for $\depth \geq 0$,
    \item $\unfold{(\WHILE{\guard}{\progBody})}{\depth+1} = \ITE{\guard}{\SEQ{\unfold{\progBody}{\depth}}{\unfold{(\WHILE{\guard}{\progBody})}}{\depth}}{\SKIP}$.
\end{itemize}

\section{Proofs}
\label{app:proofs}
\subsection{Proof of \Cref{thm:oneVsTwoSups}}
\label{proof:oneVsTwoSups}
\oneVsTwoSups*
\begin{proof}
    The suprema and infima exist because $\poGen$ is an $\omega$-cpo.
    The fact that the order of the two suprema can be exchanged is a standard property of suprema, so we only show the leftmost equality:
    For all $i \in \nats$, $\sup_{j} \poElem(i,j) \geq \poElem(i,i)$, so $\sup_{i} \sup_{j} \poElem(i,j) \geq \sup_{i} \poElem(i,i)$.
    For the other inequality we argue as follows:
    Let $i \in \nats$ be arbitrary.
    \begin{align*}
        & \sup_j \poElem(i,j) \\
        \eeq & \sup_{i \leq j} \poElem(i,j) \tag{because $\poElem$ is non-decreasing in the 2nd argument} \\
        \lleq & \sup_{i \leq j} \poElem(j,j) \tag{because $i \leq j$ and $\poElem$ is non-decreasing in the 1st argument} \\
        \eeq & \sup_{j} \poElem(j,j)
    \end{align*}
    Since $i$ was arbitrary it follows that $\sup_i \sup_j \poElem(i,j) \leq \sup_{j} \poElem(j,j)$.
\end{proof}

\subsection{Proof of \Cref{thm:smallNormSuffices}}
\label{proof:smallNormSuffices}
\smallNormSuffices*
\begin{proof}
    We only show the equalities for the lower integral.
    The statement for the upper integral follows with \Cref{thm:upperLowerDuality}.
    
    We first show that $\lowerIntGen \fun(x) \,dx = \lim_{n \to \infty} \lowerSum{\fun}{\partition_n}$.
    Let us fix $\eps > 0$.
    We want to prove that there exists $n_0 \in \nats$ such that $\forall n \geq n_0$ it holds that $ \lowerIntGen \fun - \lowerSum{\fun}{\partition_n} < \eps$.
    
    By definition of $\lowerIntGen \fun$ as the supremum over all partitions there exists a partition $\partition_\eps$ such that 
    $\lowerIntGen \fun -\lowerSum{\fun}{\partition_\eps} < \frac{\eps}{2}.$ 
    Let $\partitionSize$ be the number of points in $\partition_\eps$ and suppose that $\fun$ is bounded by $\boundedFunBound$, i.e., $|\fun| \leq \boundedFunBound$.
    Since $\lim_{n \to \infty} \partitionNorm{\partition_n} = 0$ there exists $n_0$ such that $\forall n \geq n_0$ we have $\partitionNorm{\partition_n} < \frac{\eps}{4 \partitionSize \boundedFunBound}$.
    Let $n \geq n_0$ be fixed for the rest of the proof.
    
    We set $\partition_n^* = \partition_n \cup \partition_\eps$, i.e., the partition obtained considering as extrema of the sub-intervals all the points in $\partition_n$ and $\partition_\eps$.
    Then $\partition_n^* \partitionRefines \partition_n$ and $\partition_n^* \partitionRefines \partition_\eps$,
    and thus $\lowerSum{\fun}{\partition_n^*} \geq \lowerSum{\fun}{\partition_n}$ and $\lowerSum{\fun}{\partition_n^*} \geq \lowerSum{\fun}{\partition_\eps}$.
    Let us consider the difference $\lowerSum{\fun}{\partition_n^*} - \lowerSum{\fun}{\partition_n}$.
    Let $I$ be the set of indexes $i$ such that $x_i \in \partition_n$ and $[x_i, x_{i+1}]$ contains at least one point $x_i^{\eps}$ of $\partition_\eps$.
    Then $|I| \leq \partitionSize$ and 
    \begin{align*}
        & \lowerSum{\fun}{\partition_n^*} - \lowerSum{\fun}{\partition_n} \\
        \eeq & \sum_{i \in I} \left( (x_{i+1} - x_i^{\eps}) \inf_{[x_i^{\eps},x_{i+1}]} \fun(x) + (x_i^{\eps} - x_i) \inf_{[x_i, x_i^{\eps}]} \fun(x) \mminus (x_{i+1} - x_i) \inf_{[x_i, x_{i+1}]} \fun(x) \right) \\
        \begin{split}
            \eeq & \sum_{i \in I} \left( (x_{i+1} - x_i^{\eps}) \inf_{[x_i^{\eps}, x_{i+1}]} \fun(x) + (x_i^{\eps} - x_i) \inf_{[x_i, x_i^{\eps}]} \fun(x) \right. \\
            & \qquad\qquad\qquad\qquad \left. \mminus  (x_{i+1} - x_i^{\eps}) \inf_{[x_i, x_{i+1}]} \fun(x) - (x_i^{\eps} - x_i) \inf_{[x_i, x_{i+1}]} \fun(x) \right)
        \end{split} \\
        \lleq & \sum_{i \in I} (x_{i+1} - x_i^{\eps}) \underbrace{\left(\inf_{[x_i^{\eps}, x_{i+1}]} \fun(x) - \inf_{[x_i, x_{i+1}]} \fun(x)\right)}_{\leq 2\boundedFunBound} \pplus \sum_{i \in I}  (x_i^{\eps} - x_i) \underbrace{\left( \inf_{[x_i, x_i^{\eps}]} \fun(x) - \inf_{[x_i, x_{i+1}]} \fun(x) \right)}_{\leq 2\boundedFunBound} \\
        \lleq & \partitionSize \frac{\eps}{4 \partitionSize \boundedFunBound} 2\boundedFunBound \eeq \frac{\eps}{2} ~.
    \end{align*}
    It follows that
    \[
    \lowerIntGen \fun - \lowerSum{\fun}{\partition_n}
    \eeq
    \lowerIntGen \fun - \lowerSum{\fun}{\partition_n^*} + \lowerSum{\fun}{\partition_n^*} - \lowerSum{\fun}{\partition_n}
    \lleq
    \lowerIntGen \fun - \lowerSum{\fun}{\partition_\eps} + L_{f,\partition_n^*} - \lowerSum{\fun}{\partition_n} < \frac{\eps}{2} + \frac{\eps}{2} = \eps
    ~.
    \]
    It remains to show that $\lim_{n \to \infty} \lowerSum{\fun}{\partition_n} = \sup_{n \geq 0} \lowerSum{\fun}{\partition_n}$.
    By definition of the lower integral as the supremum over the lower sums w.r.t.\ \emph{every possible} partition it follows that $\sup_{n \in \nats} \lowerSum{\fun}{\partition_n} \leq \lowerIntGen \fun$.
    On the other hand, the limit of a convergent sequence is always bounded by its supremum, i.e., $\lim_{n \to \infty} \lowerSum{\fun}{\partition_n} \leq \sup_{n \in \nats} \lowerSum{\fun}{\partition_n}$
\end{proof}

\begin{lemma}
    \label{thm:upperLowerDuality}
    For bounded $\fun \colon \clIvalGen \to \reals$ and $\partition \in \partitions{\clIvalGen}$ it holds that
    \begin{enumerate}
        \item $\lowerSum{-\fun}{\partition} = -\upperSum{\fun}{\partition}$, and
        \item $\lowerIntGen -\fun(x) \,dx = -\upperIntGen \fun(x) \,dx$.
    \end{enumerate}
\end{lemma}
\begin{proof}
    Let $\partition = (x_0,\ldots,x_{\partitionSize})$.
    For (1) note that
    \begin{align*}
        \lowerSum{-\fun}{\partition} \eeq & \sum_{i=0}^{\partitionSize-1} (x_{i+1} - x_i) \inf_{x \in \clIval{x_i}{x_{i+1}}} -\fun(x) \\
        \eeq & \sum_{i=0}^{\partitionSize-1} (x_{i+1} - x_i) \cdot \left(-\sup_{x \in \clIval{x_i}{x_{i+1}}} \fun(x) \right) \\
        \eeq & -\sum_{i=0}^{\partitionSize-1} (x_{i+1} - x_i) \sup_{x \in \clIval{x_i}{x_{i+1}}} \fun(x) \\
        \eeq & -\upperSum{\fun}{\partition} ~.
    \end{align*}
    For (2) we argue similarly:
    \begin{align*}
        \lowerIntGen-\fun(x) \,dx \eeq & \sup \left\{ \lowerSum{-\fun}{\partition} \mid \partition \in \partitions{\clIvalGen} \right\} \\
        \eeq & \sup \left\{ -\upperSum{\fun}{\partition} \mid \partition \in \partitions{\clIvalGen} \right\} \tag{by (1)} \\
        \eeq & -\inf \left\{ \upperSum{\fun}{\partition} \mid \partition \in \partitions{\clIvalGen} \right\} \\
        \eeq & - \upperIntGen \fun(x) \,dx ~.
    \end{align*}
\end{proof}
\subsection{Proof of \Cref{thm:measexpsbicpo}}
\label{proof:measexpsbicpo}
\measexpsbicpo*
\begin{proof}
    
    To see that $\po{\exps}{\eleq}$ and $\po{\bexps}{\eleq}$ are complete lattices for any $\poSubset \subseteq \exps$ (resp. $\poSubset \subseteq \bexps$) set
    $$ \sup \poSubset = \lambda \st . \sup_{\ex \in \poSubset} \ex(\st), \hspace{1cm} \inf \poSubset = \lambda \st . \inf_{\ex \in \poSubset} \ex(\st).$$
    Then, $\sup \poSubset \in \exps$ and $\inf \poSubset \in \exps$ (resp. $\sup \poSubset \in \bexps$ and $\inf \poSubset \in \bexps$).
    
    Theorem 11.17 from \cite{rudin1953principles} ensures that for any sequence of measurable functions the supremum and infimum of the sequence is measurable. Specializing the theorem in the case of Borel measurability ensures that $\po{\expsmeas}{\eleq}$ and $\po{\bexpsmeas}{\eleq}$ are $\omega$-bicpos. 
\end{proof}

\subsection{Proof of \Cref{thm:wpWlpWellDefinedAndContinuous}}
\label{proof:wpWlpWellDefinedAndContinuous}
\wpWlpWellDefinedAndContinuous*
\begin{proof}
    Parts of this statement have been shown in in \cite[p.\ 88 (Proof of Lemma 3)]{DBLP:conf/setss/SzymczakK19} and \cite[Lemma 34]{DBLP:journals/pe/GretzKM14}. 
    Here we provide a mostly self-contained proof.
    Note that for well-definedness, one must show various points:
    preservation of measurable expectations, existence of the Lebesgue integrals, and existence of the least and greatest fixed points -- for the latter, $\omega$-(co)continuity is needed, which is why we show it simultaneously.
    
    In the following, we let $\ex \in \expsmeas$ be an arbitrary measurable expectation and $\ex_0 \eleq \ex_1 \eleq \ldots$ a (non-decreasing) $\omega$-chain in $\expsmeas$.
    Similarly, $\exb \in \bexpsmeas$ is a 1-bounded measurable expectation and $\exb_0 \egeq \exb_1 \egeq \ldots$ a (non-increasing) $\omega$-cochain in $\bexpsmeas$.
    We will show by structural induction that for all $\prog \in \pWhile$ it holds that $\wp{\prog}{\ex}$ and $\wlp{\prog}{\exb}$ are measurable and that $\sup_{i \in \nats}\wp{\prog}{\ex_i} = \wp{\prog}{\sup_{i \in \nats} \ex_i}$ as well as $\inf_{i \in \nats}\wlp{\prog}{\exb_i} = \wlp{\prog}{\inf_{i \in \nats} \exb_i}$.
    Throughout this proof we write $\sup_i$ as a shorthand for $\sup_{i\in \nats}$.

    \inductionCase{$\prog = \SKIP$}
    $\wpTrans{\SKIP} = \wlpTrans{\SKIP}$ is the identify function, hence it trivially preserves measurable expectations and is $\omega$-bicontinuous.

    \inductionCase{$\prog = \DIVERGE$}
    $\wpTrans{\DIVERGE}$ are $\wlpTrans{\DIVERGE}$ are constant expectations, hence they are both measurable and $\omega$-bicontinuous.

    \inductionCase{$\prog = \ASSIGN{\pVar}{\aExp}$}
    \newcommand{\auxiliaryStateTransformer}{U}
    \newcommand{\auxStTrans}{\auxiliaryStateTransformer}
    By definition, $\wp{\ASSIGN{\pVar}{\aExp}}{\ex} = \exSubsGen$ which can also be written as $\ex \circ \auxStTrans$ where $\auxStTrans \colon \states \to \states, \st \mapsto \pStUpdate{\st}{\pVar}{\aExp(\st)}$ is a function updating the program state according to $\ASSIGN{\pVar}{\aExp}$.
    Since $\aExp$ is measurable, $\auxStTrans$ is measurable as well, and thus $\exSubsGen$ is measurable as a composition of measurable functions. The argument for measurability of $\wlp{\ASSIGN{\pVar}{\aExp}}{\exb}$ is exactly the same. 
    
    For $\omega$-continuity we have:
    \begin{align*}
        & \wp{\ASSIGN{\pVar}{\aExp}}{\sup_i \ex_i}  \\
        \eeq & \exSubs{\left( \sup_i \ex_i \right)}{\pVar}{\aExp} \tag{definition of $\wpSymb$} \\
        \eeq & \sup_i \exSubs{\ex_i}{\pVar}{\aExp} \tag{$\sup$ defined pointwise}\\
        \eeq & \sup_i \wp{\ASSIGN{\pVar}{\aExp}}{\ex_i} \tag{definition of $\wpSymb$}
        ~.
    \end{align*} 
    An analogous argument holds for $\omega$-cocontinuity of $\wlpTrans{\ASSIGN{\pVar}{\aExp}}$.

    \inductionCase{$\prog = \UNIFASSIGN{\pVar}$}
    Measurability of $\wp{\UNIFASSIGN{\pVar}}{\ex} = \lam{\pSt}{\int_\uIval \ex (\pStUpdate{\st}{\pVar}{\xi}) \, d\lebmes(\xi)}$ and its liberal counterpart is a direct application of \Cref{thm:fubini}
    
    For $\omega$-continuity of $\wpTrans{\UNIFASSIGN{\pVar}}$ we argue as follows:
    \begin{align*}
        & \wp{\UNIFASSIGN{\pVar}}{\sup_i \ex_i} \\
        \eeq & \lam{\pSt}{\int_\uIval \left(\sup_i \ex_i\right)(\pStUpdate{\st}{\pVar}{\xi}) \, d\lebmes(\xi) } \tag{definition of $\wpSymb$} \\
        \eeq & \lam{\pSt}{\int_\uIval \sup_i \ex_i (\pStUpdate{\st}{\pVar}{\xi})  \, d\lebmes(\xi) } \tag{$\sup$ defined pointwise} \\
        \eeq & \lam{\pSt}{\sup_i \int_\uIval \ex_i (\pStUpdate{\st}{\pVar}{\xi}) \, d\lebmes(\xi)} \tag{Monotone Convergence Theorem, e.g.,~\cite[Thm.\ 12.1]{schilling2017measures}} \\
        \eeq & \sup_i \lam{\pSt}{\int_\uIval \ex_i (\pStUpdate{\st}{\pVar}{\xi}) \, d\lebmes(\xi)}  \tag{$\sup$ defined pointwise} \\
        \eeq & \sup_i \wp{\UNIFASSIGN{\pVar}}{\ex_i} \tag{definition of $\wpSymb$}
    \end{align*}
    Note the application of the the Monotone Convergence Theorem, which ensures that for a non-decreasing sequence of measurable $\exNonNegReals$-valued functions it is possible to switch the supremum over the sequence and the Lebesgue integral. Co-continuity for $\wlpSymb$ follows by expressing $\inf_i \exb_i = 1 - \sup_i (1-\exb_i)$ and observing that the functions $(1-\exb_i) \in \bexps$ are a non-decreasing sequence, so Monotone Converge Theorem can be applied to them.

    \inductionCase{$\prog = \OBSERVE{\guard}$}
    Since we require the functions in $\guards$ to be Borel measurable, $\iv{\guard}$ is Borel measurable for any $\guard \in \guards$. Measurability of $\wpSymb$ and $\wlpSymb$ then follows from the fact that pointwise product of measurable functions is measurable. 
    For $\omega$-continuity we have:
    \begin{align*}
        & \wp{\OBSERVE{\guard}}{\sup_i \ex_i} \\ 
        & \eeq \iv{\guard} \cdot \sup_i \ex_i \tag{definition of $\wpSymb$} \\
        & \eeq \sup_i \left( \iv{\guard} \cdot \ex_i \right) \tag{$\sup$ defined pointwise} \\
        & \eeq \sup_i \wp{\OBSERVE{\guard}}{\ex_i} \tag{Definition of $\wpSymb$}
    \end{align*}
    where, again, we are using the properties of the supremum taken with respect to pointwise ordering as in \cite{DBLP:conf/setss/SzymczakK19}.  The same holds for $\omega$-cocontinuity of $\wlpSymb$ using the properties of the infimum taken with respect to pointwise ordering.

    \inductionCase{$\prog = \ITE{\guard}{\prog_1}{\prog_2}$}
    By I.H.\ we know that $\wp{\prog_1}{\ex}$ and $\wp{\prog_2}{\ex}$ are measurable and $\omega$-continuous. Then, measurability of $\wp{\ITE{\guard}{\prog_1}{\prog_2}}{\ex} = \iv{\guard} \cdot \wp{\prog_1}{\ex} + \iv{\neg \guard} \cdot \wp{\prog_2}{\ex}$ follows from the fact that pointwise sum and product of measurable functions is measurable.        
    For $\omega$-continuity we have:
    \begin{align*}
        & \wp{\ITE{\guard}{\prog_1}{\prog_2}}{\sup_i \ex_i} \\
        & \eeq \iv{\guard} \cdot \wp{\prog_1}{\sup_i \ex_i} + \iv{\neg \guard} \cdot \wp{\prog_2}{\sup_i \ex_i}  \tag{definition of $\wpSymb$} \\
        & \eeq \iv{\guard}  \cdot \sup_i \wp{\prog_1}{\ex_i} + \iv{\neg \guard} \cdot \sup_i \wp{\prog_2}{\ex_i} \tag{I.H.}  \\
        & \eeq \sup_i \iv{\guard}  \cdot \wp{\prog_1}{\ex_i} + \sup_i \iv{\neg \guard} \cdot  \wp{\prog_2}{\ex_i} \tag{pointwise ordering}  \\
        & \ggeq \sup_i \iv{\guard}  \cdot \wp{\prog_1}{\ex_i} + \iv{\neg \guard} \cdot  \wp{\prog_2}{\ex_i} \tag{properties of $\sup$}
    \end{align*}
    To prove that inverse inequality holds, we proceed by contradiction, supposing $>$ holds. Then, for some $i_1, i_2 \in \nats$ it holds:
    $$  \iv{\guard}  \cdot \wp{\prog_1}{\ex_{i_1}} + \iv{\neg \guard} \cdot  \wp{\prog_2}{\ex_{i_2}}  > \sup_i \iv{\guard} \cdot \wp{\prog_1}{\ex_i} + \iv{\neg \guard} \cdot \wp{\prog_2}{\ex_i}.$$
    However, since $\ex_i$ is a $\omega$-chain we can take $k \in \nats$ such that $k > i_i$ and $k > i_2$ for which we have
    \begin{align*}
        & \iv{\guard}  \cdot \wp{\prog_1}{\ex_k} + \iv{\neg \guard} \cdot  \wp{\prog_2}{\ex_k} \\
        & \ge \iv{\guard}  \cdot \wp{\prog_1}{\ex_{i_1}} + \iv{\neg \guard} \cdot  \wp{\prog_2}{\ex_{i_2}} \tag{monotonicity of $\wpTrans{\prog_1}, \wpTrans{\prog_2}$} \\
        & > \sup_i \iv{\guard}  \cdot \wp{\prog_1}{\ex_i} + \iv{\neg \guard} \cdot  \wp{\prog_2}{\ex_i} \tag{hypothesis} 
    \end{align*}
    which is a contradiction since
    $$\iv{\guard}  \cdot \wp{\prog_1}{\ex_k} + \iv{\neg \guard} \cdot  \wp{\prog_2}{\ex_k} \le \sup_i \iv{\guard}  \cdot \wp{\prog_1}{\ex_i} + \iv{\neg \guard} \cdot  \wp{\prog_2}{\ex_i}.$$
    For $\wlpSymb$ both measurability and $\omega$-cocontinuity can be proved using the same arguments.

    \inductionCase{$\prog = \PCHOICE{\prog_1}{\prob}{\prog_2}$}
    As in the previous case, replacing $\iv{\guard}$ with $\prob$ and $\iv{\neg \guard}$ with $1 - \prob$.

    \inductionCase{$\prog = \SEQ{\prog_1}{\prog_2}$}
    Measurability of $\wp{\SEQ{\prog_1}{\prog_2}}{\ex}$ and $\wlp{\SEQ{\prog_1}{\prog_2}}{\exb}$ follows immediately from the I.H.. 
    For $\omega$-continuity we have:
    \begin{align*}
        & \wp{\SEQ{\prog_1}{\prog_2}}{\sup_i \ex_i} \\
        & \eeq \wp{\prog_1}{\wp{\prog_2}{\sup_i \ex_i}} \tag{definition of $\wpSymb$} \\
        & \eeq \wp{\prog_1}{\sup_i \wp{\prog_2}{\ex_i}} \tag{I.H.\ on $\wpTrans{\prog_2}$} \\
        & \eeq \sup_i \wp{\prog_1}{\wp{\prog_2}{\ex_i}} \tag{I.H.\ on $\wpTrans{\prog_1}$} \\
        & \eeq \sup_i \wp{\SEQ{\prog_1}{\prog_2}}{\ex_i} \tag{definition of $\wpSymb$}
    \end{align*}
    where $\wp{\prog_1}{\sup_i \wp{\prog_2}{\ex_i}}$ is well-defined since by I.H. $\wp{\prog_2}{\ex_i}$ is measurable for every $i$ and therefore $\sup_i \wp{\prog_2}{\ex_i}$ is measurable as well.

    \inductionCase{$\prog = \WHILE{\guard}{\progBody}$}
    \begin{align*}
        & \wp{\prog}{\sup_i \ex_i} \\
        & = \slfp{\lam{g}{\exlfp{\progBody}{\sup_i \ex_i}{\guard}(g)}} = \tag{definition of $\wpSymb$}\\
        & = \sup_n \left( \exlfp{\progBody}{\sup_i \ex_i}{\guard} \right)^n(\zerofun) = \tag{I.H. continuity of $\wpTrans{\progBody}$}  \\
        & = \sup_n \left( \sup_i \exlfp{\progBody}{\ex_i}{\guard} \right)^n(\zerofun) = \tag{continuity of $\lambda \ex . \iv{\neg \guard}\ex + \iv{\guard}\wp{\progBody}{\exb}$} \\
        & = \sup_n \sup_i \left(\exlfp{\progBody}{\ex_i}{\guard} \right)^n(\zerofun) = \tag{$*$} \\
        & = \sup_i \sup_n (\exlfp{\progBody}{\ex_i}{\guard})^n(\zerofun) = \tag{$**$} \\
        & = \sup_i \slfp{\lam{g}{\exlfp{\progBody}{\ex_i}{\guard}(g)}} = \tag{I.H. continuity of $\wpTrans{\progBody}$} \\
        &= \sup_i \wp{\prog}{\ex_i}. \tag{definition of $\wpSymb$}
    \end{align*}
    where $(*)$ is justified by proving with an inner induction on $n$ that:
    $$ \left(\sup_i \exlfp{\prog}{\ex_i}{\guard} \right)^n(\zerofun) = \sup_i \left( \exlfp{\prog}{\ex_i}{\guard}\right)^n(\zerofun).$$
    In fact, for $n=1$ the equality holds trivially, and assuming it holds for $n-1$ we get:
    \begin{align*}
        \text{L.H.S.} & = \left( \sup_i \exlfp{\prog}{\ex_i}{\guard} \right)\left(\sup_j \exlfp{\prog}{\ex_j}{\guard} \right)^{n-1}(\zerofun) =  \\
        & = \left( \sup_i \exlfp{\prog}{\ex_i}{\guard} \right)\left(\sup_j \left(\exlfp{\prog}{\ex_j}{\guard} \right)^{n-1}(\zerofun)\right) = \tag{case $n-1$} \\
        & = \sup_j \left( \sup_i \exlfp{\prog}{\ex_i}{\guard} \right)\left(\left(\exlfp{\prog}{\ex_j}{\guard} \right)^{n-1}(\zerofun)\right) = \tag{outer I.H.} \\
        & = \sup_j \sup_i \left( \exlfp{\prog}{\ex_i}{\guard} \right)\left(\left(\exlfp{\prog}{\ex_j}{\guard} \right)^{n-1}(\zerofun)\right) = \tag{case $n=1$}\\
        & = \sup_i \left( \exlfp{\prog}{\ex_i}{\guard} \right)\left(\left(\exlfp{\prog}{\ex_i}{\guard} \right)^{n-1}(\zerofun)\right) = \tag{Lemma \ref{thm:oneVsTwoSups}}  \\
        & = \sup_i \left(\left(\exlfp{\prog}{\ex_i}{\guard} \right)^{n}(\zerofun)\right) = \text{R.H.S.}.
    \end{align*}
    For $(**)$ switching the order of suprema is justified by the fact that both terms exists.
    
    To prove co-continuity of $\wlpSymb$ we can use the same arguments: $(*)$ and $(**)$ still hold if $\sup$ is replaced by $\inf$, since Lemma \ref{thm:oneVsTwoSups} still holds and the order of infima can be switched if both terms of the equality exist.
\end{proof}
\subsection{Proof of \Cref{thm:monriemann}}
\label{proof:monriemann}
\monriemann*
\begin{proof}
    We will show that the statement holds for $\lwpTrans{\N}{\prog}$ as the proof in the remaining cases is analogous. We proceed by structural induction on $\prog$ . For $\prog \in \{ \SKIP, \DIVERGE, \ASSIGN{\pVar}{\aExp}, \OBSERVE{\guard} \}$ $\lwpTrans{\N}{\prog}$ is equal to $\wpTrans{\prog}$ so the conclusion holds by monotonicity of the latter \cite{DBLP:series/mcs/McIverM05}. 
    
    For $\prog = \UNIFASSIGN{\pVar}$ suppose $\ex \eleq \exb$, then:
    \begin{align*}
        & \lwp{\N}{\prog}{\ex} \\
        \eeq & \frac{1}{\N} \sum_{i=0}^{\N-1} \inf_{\xi \in [\frac{i}{\N}, \frac{i+1}{\N}]} \exSubs{\ex}{\pVar}{\xi} \tag{definition of $\lwpSymb{\N}$}  \\
        \eleq & \frac{1}{\N} \sum_{i=0}^{\N-1} \inf_{\xi \in [\frac{i}{\N}, \frac{i+1}{\N}]} \exSubs{\exb}{\pVar}{\xi} \tag{$\ex \eleq \exb$} \\
        \eeq & \lwp{\N}{\prog}{\ex} \tag{definition of $\lwpSymb{\N}$}
    \end{align*}
    
    For $\prog \in \{ \ITE{\guard}{\prog_1}{\prog_2}, \PCHOICE{\prob}{\prog_1}{\prog_2}, \SEQ{\prog_1}{\prog_2} \}$ follows by I.H. on $\prog_1$ and $\prog_2$.
    
    For $\prog = \WHILE{\guard}{\progBody}$ consider the characteristic functions $\charfunuwp{\N}{\prog}{\ex}$ and $\charfunuwp{\N}{\prog}{\exb}$.
    Let us denote by $Y_\ex$ and $Y_\exb$ the least fixed points of $\charfunuwp{\N}{\prog}{\ex}$ and $\charfunuwp{\N}{\prog}{\exb}$ respectively. From I.H. it follows that $\charfunuwp{\N}{\prog}{\ex}(Y) \eleq \charfunuwp{\N}{\prog}{\exb}(Y)$ for all $Y$. Then, $ \charfunuwp{\N}{\prog}{\ex}(Y_\exb) \eleq \charfunuwp{\N}{\prog}{\exb}(Y_\exb) = Y_\exb.$ By Theorem \ref{thm:knasterTarski} the previous implies $Y_\ex \eleq Y_\exb$.
\end{proof}

\subsection{Counter-example to $\omega$-continuity of $\lwpSymb{\N}$}
\label{proof:lwpNotOmegaCont}

Consider $\prog = \UNIFASSIGN{\pVar}$ and the sequence of functions $\ex_i = \onefun - \sum_{k=i}^\infty \iv{\pVar = \frac{1}{k}}$. Then $\sup_i \ex_i = \onefun$ and therefore $\lwpTrans{\N}{\prog}\left( \sup_i \ex_i \right) = \onefun. $
On the other hand,
\begin{align*}
    & \sup_i \lwpTrans{\N}{\prog}(\ex_i) \\
    \eeq & \sup_i \frac{1}{\N} \sum_{j=0}^{\N -1} \inf_{\xi \in \left[ \frac{j}{\N}, \frac{j+1}{\N} \right]} \exSubs{\ex_i}{\pVar}{\xi} \tag{definition of $\lwpSymb{\N}$} \\
    \eeq & \sup_i \left( \frac{1}{\N} \inf_{\xi \in \left[ \frac{0}{\N}, \frac{1}{\N} \right]} \exSubs{\ex_i}{\pVar}{\xi} + \frac{1}{\N} \sum_{j=0}^{\N -1} \inf_{\xi \in \left[ \frac{j}{\N}, \frac{j+1}{\N} \right]} \exSubs{\ex_i}{\pVar}{\xi} \right) \tag{splitting the summation} \\
    \le & \sup_i \left( \frac{\N - 1}{\N} \right) \tag{$\inf_{\xi \in \left[ \frac{0}{\N}, \frac{1}{\N} \right]} \exSubs{\ex_i}{\pVar}{\xi} = 0 \, \forall i$} \\
    \eeq & \frac{\N-1}{\N} \tag{sup of a constant sequence}
\end{align*}

\subsection{Proof of \Cref{thm:soundness}}
\label{proof:soundness}
\soundness*
\begin{proof}
    We prove the statement for $\lwpTrans{\N}{\prog}$ since the remaining cases are analogous.
    
    We proceed by structural induction on $\prog$.
    
    For $\prog \in \{ \SKIP, \DIVERGE, \ASSIGN{\pVar}{\aExp}, \OBSERVE{\guard} \}$ the definitions of $\lwpTrans{\N}{\prog}$ and $\wpTrans{\prog}$ coincide so pointwise inequality holds trivially.
    
    For $\prog = \UNIFASSIGN{\pVar}$ we have
    \begin{align*}
        & \lwp{\N}{\prog}{\ex} \\
        \eeq & \frac{1}{\N} \sum_{i=0}^{\N-1} \inf_{\xi \in \left[ \frac{i}{\N}, \frac{i+1}{\N} \right]} \exSubs{\ex}{\pVar}{\xi} \tag{definition of $\lwpSymb{\N}$} \\
        \eeq & \sum_{i=0}^{\N -1} \int_{\frac{i}{\N}}^{\frac{i+1}{\N}}  \inf_{\xi \in \left[ \frac{i}{\N}, \frac{i+1}{\N} \right]} \exSubs{\ex}{\pVar}{\xi} \, d\lebmes(\xi) \tag{integral of constant function} \\
        \eleq & \sum_{i=0}^{\N - 1} \int_{\frac{i}{\N}}^{\frac{i+1}{\N}} \exSubs{\ex}{\pVar}{\xi} \, d\lebmes{\xi} \tag{monotonicity of integral} \\
        \eeq & \int_0^1 \exSubs{\ex}{\pVar}{\xi} \, d\lebmes(\xi) \tag{elementary property of integral} \\
        \eeq & \wp{\prog}{\ex}. \tag{definition of $\wpSymb$}
    \end{align*}

    For $\prog \in \{ \ITE{\guard}{\prog_1}{\prog_2}, \PCHOICE{\prob}{\prog_1}{\prog_2} \}$ conclusion follows applying the I.H. on $\prog_1$ and $\prog_2$.
    
    For $\prog = \SEQ{\prog_1}{\prog_2}$ we have
    \begin{align*}
        & \lwp{\N}{\prog}{\ex} = \\
        \eeq & \lwp{\N}{\prog_1}{\lwp{\N}{\prog_2}{\ex}} \tag{definition of $\lwpSymb{\N}$} \\
        \eleq \, & \lwp{\N}{\prog_1}{\wp{\prog_2}{\ex}} \tag{I.H. on $\prog_2$ + Lemma \ref{thm:monriemann}} \\
        \eleq \, & \wp{\prog_1}{\wp{\prog_2}{\ex}} \tag{I.H. on $\prog_1$} \\
        \eeq & \wp{\prog}{\ex}. \tag{definition of $\wpSymb$}
    \end{align*}
    
    Finally, for $\prog = \WHILE{\guard}{\progBody}$ consider the characteristic functions $\charfunwp{\prog}{\ex}$ and $\charfunlwp{\N}{\prog}{\ex}$.
    Let $Y^*$ and $Y^*_\N$ be the least fixed points of $\charfunwp{\prog}{\ex}$ and $\charfunlwp{\N}{\prog}{\ex}$ respectively.
    From I.H. on $\progBody$ it follows $\charfunlwp{\N}{\prog}{\ex}(Y) \eleq \charfunwp{\prog}{\ex}(Y)$ for all $Y$.
    Then, $\charfunlwp{\N}{\prog}{\ex}(Y^*) \eleq \charfunwp{\prog}{\ex}(Y^*) = Y^*$ and from Theorem \ref{thm:knasterTarski} it follows $Y^*_\N \eleq Y^*$.
\end{proof}

\subsection{Proof of \Cref{thm:powerTwoMonotonic}}
\label{proof:powerTwoMonotonic}

\begin{restatable}{lemma}{powerTwoMonotonic}
    \label{thm:powerTwoMonotonic}
    For all programs $\prog \in \pWhile$ and post-expectations $\ex \in \exps$, $\lwp{2^\N}{\prog}{\ex}$ and $\uwp{2^\N}{\prog}{\ex}$ are non-increasing (non-decreasing, respectively) sequences in $\N \geq 1$.
\end{restatable}
\begin{proof}
    Let $\N \geq 1$ be fixed.
    We show by induction on the structure of $\prog$ that for all $\ex \in \exps$ it holds that $\lwp{2^\N}{\prog}{\ex} \leq \lwp{2^{\N+1}}{\prog}{\ex}$.
    The base case $\prog = \UNIFASSIGN{\pVar}$ follows with \Cref{thm:partitionRefine} because the partition $0 < \frac{1}{2^\N} < \frac{2}{2^\N} \ldots  < 1$ is refined by $0 < \frac{1}{2^{\N+1}} < \frac{2}{2^{\N+1}} = \frac{1}{2^\N} < \ldots < 1$.
    For the other base cases there is nothing to show.
    The inductive cases are straightforward using monotonicity of $\lwpTrans{2^\N}{\prog}$.
    We exemplify this for sequential composition:
    \begin{align*}
        & \lwp{2^\N}{\SEQ{\prog_1}{\prog_2}}{\ex} \\
        \eeq & \lwp{2^\N}{\prog_1}{\lwp{2^\N}{\prog_2}{\ex}} \tag{definition of $\lwpSymb{2^\N}$} \\
        \eeleq & \lwp{2^\N}{\prog_1}{\lwp{2^{\N+1}}{\prog_2}{\ex}} \tag{I.H. on $\prog_2$ and monotonicity of $\lwpTrans{2^\N}{\prog_1}$ (\Cref{thm:monriemann})} \\
        \eeleq & \lwp{2^{\N+1}}{\prog_1}{\lwp{2^{\N+1}}{\prog_2}{\ex}} \tag{I.H. on $\prog_1$} \\
        \eeq & \lwp{2^{\N+1}}{\SEQ{\prog_1}{\prog_2}}{\ex} \tag{definition of $\lwpSymb{2^{\N+1}}$} ~.
    \end{align*}
    The argument for $\uwpSymb{}$ is similar.
\end{proof}

\subsection{Proof of \Cref{thm:wpLoopFreeCocontLocBounded}}
\label{proof:wpLoopFreeCocontLocBounded}

Let $\expsmeaslb$ be the set of measurable locally bounded expectations.

\begin{restatable}{lemma}{wpLoopFreeCocontLocBounded}
    \label{thm:wpLoopFreeCocontLocBounded}
    Let $\aExps$ be a set of locally bounded measurable arithmetic expressions (\enquote{right hand sides}) and $\guards$ be an arbitrary set of guards.
    Let $\prog \in \pWhileWith{\aExps}{\guards}$ be loop-free.
    Then $\wpTrans{\prog}$ preserves $\expsmeaslb$.
    Moreover, $\wpTrans{\prog}$ preserves limits of non-increasing sequences from $\expsmeaslb$.
\end{restatable}
\begin{proof}
    By induction on the structure of $\prog$.
    Let $\ex \in \expsmeaslb$ be locally bounded and let $\ex_0 \egeq \ex_1 \egeq \ldots$ be a non-increasing sequence of expectations from $\expsmeaslb$.
    
    \inductionCase{$\prog = \SKIP$}
    Trivial.
    
    \inductionCase{$\prog = \DIVERGE$}
    Trivial.
    
    \inductionCase{$\prog = \OBSERVE{\guard}$}
    Straightforward.
    
    \inductionCase{$\prog = \ASSIGN{\pVar}{\aExp}$}
    We have to show that $\wp{\ASSIGN{\pVar}{\aExp}}{\ex} = \exSubsGen = \lam{\pSt}{\ex(\pStUpdate{\pSt}{\pVar}{\aExp(\pSt)})}$ is locally bounded, i.e.,
    that
    $\sup_{\pSt \in \compactSet} \ex(\pStUpdate{\pSt}{\pVar}{\aExp(\pSt)}) < \infty$
    for all compact $\compactSet \subseteq \pStates$.
    
    Let us fix a compact $\compactSet \subseteq \pStates$.
    Since $\aExp$ is locally bounded by assumption, we have $\aExp(\compactSet) \subseteq [0,b]$ for some $b \in \nonNegReals$.
    Now let $\compactSet' \supseteq \compactSet$ be a compact set such that for all $\pSt \in \compactSet$ and $\xi \in [0,b]$ we have $\pStUpdate{\pSt}{\pVar}{\xi} \in \compactSet'$. To see that $\compactSet'$ exists consider the projection of $\compactSet$ on $\mathbb{R}^{\pVars \setminus \{ \pVar \}}$, denoted by $\pi(\compactSet)$. By Theorem 39 in \cite{pugh2002real} $\pi(\compactSet)$ is compact. Then set $\compactSet' = \pi(\compactSet) \times [0,b]$. By Theorem 31 in \cite{pugh2002real} $\compactSet'$ is compact and $\pStUpdate{\pSt}{\pVar}{\xi} \in \compactSet'$ for all $\pSt \in \compactSet$ and $\xi \in [0,b]$.
    
    Then we have
    \[
    \{\ex(\pStUpdate{\pSt}{\pVar}{\aExp(\pSt)}) \mid \pSt \in \compactSet\}
    ~\subseteq~
    \{\ex(\pStUpdate{\pSt}{\pVar}{\xi}) \mid \pSt \in \compactSet, \xi \in [0,b]\}
    ~\subseteq~
    \{\ex(\pSt) \mid \pSt \in \compactSet'\}
    \]
    and thus $\sup_{\pSt \in \compactSet} \ex(\pStUpdate{\pSt}{\pVar}{\aExp(\pSt)}) \leq \sup_{\pSt \in \compactSet'} \ex(\pSt) < \infty$ because $\ex$ is locally bounded.
    
    \inductionCase{$\prog = \UNIFASSIGN{\pVar}$}
    We first show that
    \[
    \wp{\UNIFASSIGN{\pVar}}{\ex} = \lam{\st}{\int_\uIval \ex(\pStUpdate{\st}{\pVar}{\xi}) \,d\lebmes(\xi)}
    \]
    remains locally bounded.
    To this end let $\compactSet \subseteq \pStates$ be compact.
    Recall that we have to show that $\sup \wp{\UNIFASSIGN{\pVar}}{\ex}(\compactSet) < \infty$.
    Let $\compactSet' \supseteq \compactSet$ be a compact set such that for all $\pSt \in \compactSet$ and $\xi \in \uIval$ we have $\pStUpdate{\pSt}{\pVar}{\xi} \in \compactSet'$, where $\compactSet'$ can be constructed as above.
    Since $\ex$ is locally bounded we have $\sup \ex(\compactSet') = b \in \nonNegReals$.
    Hence for all $\pSt \in \compactSet'$,
    \begin{align*}
        & \int_\uIval \ex(\pStUpdate{\st}{\pVar}{\xi}) \,d\lebmes(\xi) \\
        \lleq & \int_\uIval \sup \ex(\compactSet') \,d\lebmes(\xi) \\
        \eeq & \int_\uIval b \,d\lebmes(\xi) \\
        \eeq & b 
    \end{align*}
    It follows that $\sup \wp{\UNIFASSIGN{\pVar}}{\ex}(\compactSet) \leq \sup \wp{\UNIFASSIGN{\pVar}}{\ex}(\compactSet') \leq b < \infty$ and hence the claim.
    
    The fact that $\wpSymb$ preserve infima of non-increasing sequences follows from the properties of pointwise $\inf$.
    
    The fact that $\wpTrans{\UNIFASSIGN{\pVar}}$ preserves infima of non-increasing sequences $( \ex_i )$ follows from an extended version of the Monotone Convergence \cite[Theorem 12.1 (ii)]{schilling2017measures} adapted to non-increasing sequences of functions. Notably, the corollary requires that the integral of the first function in the sequence is smaller than $\infty$, but this is guaranteed by the restriction to locally bounded expectations, since for every $\pSt \in \pStates$ we can consider the compact set $\compactSet_\pSt = \{\pStUpdate{\pSt}{\pVar}{\xi} \mid \xi \in \uIval\}$, and since $\ex_0$ is locally bounded, we find that $\int_\uIval \ex_0(\pStUpdate{\st}{\pVar}{\xi}) \,d\lebmes(\xi) \leq \sup \compactSet_\pSt < \infty$.
    
    \inductionCase{$\prog = \ITE{\guard}{\prog_1}{\prog_2}$}
    Local boundedness follows by applying the I.H.\ to $\prog_1$ and $\prog_2$, while preservation of the limit can be proved using the same argument as in the proof \Cref{thm:wpWlpWellDefinedAndContinuous}.
    
    \inductionCase{$\prog = \PCHOICE{\prog_1}{\prob}{\prog_2}$}
    Local boundedness follows by applying the I.H.\ to $\prog_1$ and $\prog_2$, while preservation of the limit can be proved using the same argument as in the proof \Cref{thm:wpWlpWellDefinedAndContinuous}.
    
    \inductionCase{$\prog = \SEQ{\prog_1}{\prog_2}$}
    Since $\wp{\SEQ{\prog_1}{\prog_2}}{\ex} = \wp{\prog_1}{\wp{\prog_2}{\ex}}$ it follows immediately from the I.H.\ that $\wp{\prog_1}{\wp{\prog_2}{\ex}} \in \expsmeaslb$.
    
    To conclude the proof observe that
    \begin{align*}
        &\inf_{i \in \nats} \wp{\SEQ{\prog_1}{\prog_2}}{\ex_i} \\
        \eeq &\inf_{i \in \nats} \wp{\prog_1}{\wp{\prog_2}{\ex_i}} \\
        \eeq & \wp{\prog_1}{\inf_{i \in \nats} \wp{\prog_2}{\ex_i}} \tag{I.H.\ on $\prog_1$, using that $\wp{\prog_2}{\ex_i}$ is a non-decreasing sequence in $\expsmeaslb$ by I.H.\ on $\prog_2$} \\
        \eeq & \wp{\prog_1}{\wp{\prog_2}{\inf_{i \in \nats}  \ex_i}} \tag{I.H.\ on $\prog_2$} \\
        \eeq & \wp{\SEQ{\prog_1}{\prog_2}}{\inf_{i \in \nats}  \ex_i}
    \end{align*}
\end{proof}

\subsection{Proof of \Cref{thm:convLoopFree}}
\label{proof:convLoopFree}
\convLoopFree*
\begin{proof}
    In the rest of the proof we write  $\sup_\N$ and $\inf_\N$ instead of $\sup_{\N \geq 1}$ and $\inf_{\N \geq 1}$, respectively.
    We will show the following \emph{alternative} statement:
    For all loop-free $\prog \in \pWhileWith{\aExps}{\guards}$ and post-expectations $\ex \in \expsClass$:
    \eqref{eq:approxInClass} from above holds, and
    \begin{align}
        \wp{\prog}{\ex}
        \eeq
        \sup_\N \lwp{2^\N}{\prog}{\ex}
        \eeq
        \inf_\N \uwp{2^\N}{\prog}{\ex}
        ~.
        \tag{\ref{eq:convLoopFree}'}
        \label{eq:convLoopFreeAlt}
    \end{align}
    Note that the difference between \eqref{eq:convLoopFreeAlt} and $\eqref{eq:convLoopFree}$
    is that we consider the supremum over powers of two.
    This simplifies the situation as $\lwp{2^\N}{\prog}{\ex}$ and $\uwp{2^\N}{\prog}{\ex}$ are non-decreasing (non-increasing, respectively) sequences in $\N$ (see \Cref{thm:powerTwoMonotonic}), which allows us to apply \Cref{thm:oneVsTwoSups} at a crucial point in the proof.
    
    To show that the alternative statement implies the claim of \Cref{thm:convLoopFree} we argue as follows:
    By properties of suprema, $\sup_\N \lwp{2^\N}{\prog}{\ex} \leq \sup_\N \lwp{\N}{\prog}{\ex}$.
    We also have
    \begin{align*}
        & \sup_\N \lwp{\N}{\prog}{\ex} \\
        \lleq & \wp{\prog}{\ex} \tag{by \Cref{thm:soundness}} \\
        \eeq & \sup_\N \lwp{2^\N}{\prog}{\ex} \tag{by the alternative statement \eqref{eq:convLoopFreeAlt}} ~,
    \end{align*}
    and thus $\sup_\N \lwp{2^\N}{\prog}{\ex} = \sup_\N \lwp{\N}{\prog}{\ex}$.
    The argument for $\inf_\N \uwp{2^\N}{\prog}{\ex} = \inf_\N \uwp{\N}{\prog}{\ex}$ is analogous.
    
    The proof of the alternative statement \eqref{eq:convLoopFreeAlt} is by induction over the structure of $\prog$.
    In the following, we let $\ex \in \expsClass$ and $\N \in \nats$ be arbitrary.
    
    We consider the base cases first.
    For the base cases $\prog \in \{\SKIP, \DIVERGE, \ASSIGN{\pVar}{\aExp}, \OBSERVE{\guard} \}$ the definitions of $\lwpSymb \N$, $\uwpSymb \N$, and $\wpSymb$ coincide, so we only have to show that $\wp{\prog}{\ex} \in \expsClass$.

    \inductionCase{$\prog = \SKIP$}
    By definition, $\wp \SKIP \ex = \ex \in \expsClass$.
    
    \inductionCase{$\prog = \DIVERGE$}
    By definition, $\wp \DIVERGE \ex = 0$.
    By assumption, $0 \in \expsClass$.
    
    \inductionCase{$\prog = \ASSIGN{\pVar}{\aExp}$}
    By definition, $\wp{\ASSIGN{\pVar}{\aExp}}{f} = \exSubs{\ex}{\pVar}{\aExp}$.
    By assumption, $\exSubs{\ex}{\pVar}{\aExp} \in \expsClass$.
    
    \inductionCase{$\prog = \OBSERVE{\guard}$}
    By definition, $\wp{\OBSERVE{\guard}}{f} = \iv{\guard} \cdot \ex$.
    By assumption, $\iv{\guard} \cdot \ex + \iv{\neg\guard} \cdot 0 \in \expsClass$.
    
    \inductionCase{$\prog = \UNIFASSIGN{\pVar}$}
    This base case is more interesting and crucially relies on Riemann-integrability.
    We first show that $\lwp{\N}{\UNIFASSIGN{\pVar}}{\ex} \in \expsClass$.
    By definition of $\lwpSymb{\N}$,
    \begin{align*}
        \lwp{\N}{\UNIFASSIGN{\pVar}}{\ex}
        \eeq
        \frac{1}{\N} \sum_{i=0}^{\N-1} \inf_{\xi \in [\frac{i}{\N}, \frac{i+1}{\N}]} \exSubs{\ex}{\pVar}{\xi}
        ~.
    \end{align*}
    The right-hand side is a finite convex combination of suprema over subintervals of $\uIval$, hence it is in $\expsClass$ by assumption.
    Next we show
    $\sup_\N \lwp{2^\N}{\UNIFASSIGN{\pVar}}{\ex} = \wp{\UNIFASSIGN{\pVar}}{\ex}$:
    \begin{align*}
        & \sup_\N \lwp{2^\N}{\UNIFASSIGN{\pVar}}{\ex} \\
        \eeq & \sup_N \frac{1}{2^\N} \sum_{i=0}^{2^\N-1} \inf_{\xi \in [\frac{i}{2^\N}, \frac{i+1}{2^\N}]}
        \exSubs{\ex}{\pVar}{\xi} \tag{definition of $\lwpSymb{2^\N}$} \\
        \eeq & \lam{\st}{} \lowerInt{0}{1} \ex(\pStUpdate{\st}{\pVar}{\xi}) \,d\xi \tag{by \Cref{thm:smallNormSuffices}} \\
        \eeq & \lam{\st}{} \int_{0}^{1} \ex(\pStUpdate{\st}{\pVar}{\xi}) \,d\xi \tag{by assumption $\ex \in \expsClass$ is bounded and Riemann-integrable w.r.t.\ $\pVar$} \\
        \eeq & \lam{\st}{} \int_\uIval  \ex(\pStUpdate{\st}{\pVar}{\xi}) \,d\lebmes(\xi) \tag{Riemann integral = Lebesgue integral, see \Cref{thm:riemannEqualsLebesgue}} \\
        \eeq & \wp{\UNIFASSIGN{\pVar}}{\ex} \tag{definition of $\wpSymb$}
    \end{align*}
    The proof for $\uwpSymb{}$ is exactly analogous (replace all infima by suprema and vice versa).
    
    \inductionCase{$\prog = \ITE{\guard}{\prog_1}{\prog_2}$}
    $\lwp{\N}{\ITE{\guard}{\prog_1}{\prog_2}}{\ex}
    = [\guard] \cdot \lwp{\N}{\prog_1}{\ex} + [\neg\guard] \cdot \lwp{\N}{\prog_2}{\ex} \in \expsClass$ by I.H. on $\prog_1$ and $\prog_2$ and assumption.
    Furthermore,
    \begin{align*}
        & \sup_\N \lwp{2^\N}{\ITE{\guard}{\prog_1}{\prog_2}}{\ex} \\
        \eeq & \sup_\N \left(\iv{\guard} \cdot \lwp{2^\N}{\prog_1}{\ex} + \iv{\neg\guard} \cdot \lwp{2^\N}{\prog_2}{\ex}\right) \tag{definition of $\lwpSymb{\N}$} \\
        \eeq & \iv{\guard} \cdot \sup_\N \lwp{2^\N}{\prog_1}{\ex} + \iv{\neg\guard} \cdot \sup_N\lwp{2^\N}{\prog_2}{\ex} \tag{addition and multiplication by Iverson brackets is $\omega$-continuous} \\
        \eeq & \iv{\guard} \cdot \wp{\prog_1}{\ex} + \iv{\neg\guard} \cdot \wp{\prog_2}{\ex} \tag{by I.H. on $\prog_1$ and $\prog_2$} \\
        \eeq & \wp{\ITE{\guard}{\prog_1}{\prog_2}}{\ex} \tag{definition of $\wpSymb$}
        ~.
    \end{align*}
    Again, the proof for $\uwpSymb{}$ is analogous (consider infima instead of suprema, and use that addition and multiplication by Iverson brackets is $\omega$-\underline{co}continuous).
    
    \inductionCase{$\prog = \PCHOICE{\prog_1}{\prob}{\prog_2}$}
    $\lwp{\N}{\PCHOICE{\prog_1}{\prob}{\prog_2}}{\ex}
    = \prob \cdot \lwp{\N}{\prog_1}{\ex} + (1-\prob) \cdot \lwp{\N}{\prog_2}{\ex} \in \expsClass$ by I.H. on $\prog_1$ and $\prog_2$ and assumption.
    Furthermore,
    \begin{align*}
        & \sup_\N \lwp{2^\N}{\PCHOICE{\prog_1}{\prob}{\prog_2}}{\ex} \\
        \eeq & \sup_\N \left(\prob \cdot \lwp{2^\N}{\prog_1}{\ex} + (1-\prob) \cdot \lwp{2^\N}{\prog_2}{\ex}\right) \tag{definition of $\lwpSymb{\N}$} \\
        \eeq & \prob \cdot \sup_\N \lwp{2^\N}{\prog_1}{\ex} + (1-\prob) \cdot \sup_\N\lwp{2^\N}{\prog_2}{\ex} \tag{addition and multiplication by constants is $\omega$-continuous} \\
        \eeq & \prob \cdot \wp{\prog_1}{\ex} + (1-\prob) \cdot \wp{\prog_2}{\ex} \tag{by I.H. on $\prog_1$ and $\prog_2$} \\
        \eeq & \wp{\PCHOICE{\prog_1}{\prob}{\prog_2}}{\ex} \tag{definition of $\wpSymb$}
        ~.
    \end{align*}
    As before, the proof for $\uwpSymb{}$ is analogous (consider infima instead of suprema, and use that addition and multiplication by constants is $\omega$-\underline{co}continuous).
    
    \inductionCase{$\prog = \SEQ{\prog_1}{\prog_2}$}
    First, note that $\lwp{\N}{\SEQ{\prog_1}{\prog_2}}{\ex} = \lwp{\N}{\prog_1}{\lwp{\N}{\prog_2}{\ex}} \in \expsClass$ by I.H. on $\prog_2$ and $\prog_1$.
    For convergence we argue as follows:
    \begin{align*}
        & \sup_\N \lwp{2^\N}{\SEQ{\prog_1}{\prog_2}}{\ex} \\
        \eeq & \sup_\N \lwp{2^\N}{\prog_1}{\lwp{2^\N}{\prog_2}{\ex}}  \tag{definition of $\lwpSymb{\N}$} \\
        \eeq & \sup_{\Nb} \sup_\N \lwp{2^\N}{\prog_1}{\lwp{2^\Nb}{\prog_2}{\ex}} \tag{\Cref{thm:oneVsTwoSups} and \Cref{thm:powerTwoMonotonic}} \\
        \eeq & \sup_{\Nb} \wp{\prog_1}{\lwp{2^\Nb}{\prog_2}{\ex}} \tag{I.H.\ on $\prog_1$, using that $\lwp{2^\Nb}{\prog_2}{\ex} \in \expsClass$ (also by I.H.)} \\
        \eeq & \wp{\prog_1}{\sup_{\Nb} \lwp{2^\Nb}{\prog_2}{\ex}} \tag{$\omega$-continuity of $\wpTrans{\prog_1}$, see \Cref{thm:wpWlpWellDefinedAndContinuous}} \\
        \eeq & \wp{\prog_1}{\wp{\prog_2}{\ex}}{\ex} \tag{I.H.\ on $\prog_2$} \\
        \eeq & \wp{\SEQ{\prog_1}{\prog_2}}{\ex} \tag{definition of $\wpSymb$}
    \end{align*}
    The convergence proof of $\uwpSymb{}$ is not fully analogous as $\wpSymb$ is not $\omega$-\underline{co}continuous in general.
    Instead, we invoke \Cref{thm:wpLoopFreeCocontLocBounded}:
    \begin{align*}
        & \inf_\N \uwp{2^\N}{\SEQ{\prog_1}{\prog_2}}{\ex} \\
        \eeq & \inf_\N \uwp{2^\N}{\prog_1}{\uwp{2^\N}{\prog_2}{\ex}}  \tag{definition of $\uwpSymb{\N}$} \\
        \eeq & \inf_{\Nb} \inf_\N \uwp{2^\N}{\prog_1}{\uwp{2^\Nb}{\prog_2}{\ex}} \tag{\enquote{$\inf$-version} of \Cref{thm:oneVsTwoSups} and \Cref{thm:powerTwoMonotonic}} \\
        \eeq & \inf_{\Nb} \wp{\prog_1}{\uwp{2^\Nb}{\prog_2}{\ex}} \tag{I.H.\ on $\prog_1$, using that $\uwp{2^\Nb}{\prog_2}{\ex} \in \expsClass$ (also by I.H.)} \\
        \eeq & \wp{\prog_1}{\inf_{\Nb} \uwp{2^\Nb}{\prog_2}{\ex}} \tag{\Cref{thm:wpLoopFreeCocontLocBounded}} \\
        \eeq & \wp{\prog_1}{\wp{\prog_2}{\ex}}{\ex} \tag{I.H.\ on $\prog_2$} \\
        \eeq & \wp{\SEQ{\prog_1}{\prog_2}}{\ex} \tag{definition of $\wpSymb$}
    \end{align*}
\end{proof}

\subsection{{$\wlpSymb$-Version of \Cref{thm:convLoopFree}}}
\label{app:wlpVersion}
\begin{lemma}
    \label{thm:convLoopFreeWlp}
    Let $(\aExps, \guards, \expsClass)$ be Riemann-suitable.
    Then for all loop-free $\prog \in \pWhileWith{\aExps}{\guards}$ and post-expectations $\exb \in \expsClass  \cap \bexpsmeas$ the following holds:
    \begin{align*}
        \forall \N \geq 1 \colon\quad \lwlp{\N}{\prog}{\exb} \in \expsClass \qand \uwlp{\N}{\prog}{\exb} \in \expsClass
    \end{align*}
    and
    \begin{align*}
        \wlp{\prog}{\exb}
        \eeq
        \sup_{\N \geq 1} \lwlp{\N}{\prog}{\exb}
        \eeq
        \inf_{\N \geq 1} \uwlp{\N}{\prog}{\exb}
        ~.
    \end{align*}
\end{lemma}
\begin{proof}
    Very similar to the proof of \Cref{thm:convLoopFree}.
    A minor difference is the straightforward base case $\prog = \DIVERGE$.
    Also recall that $\wlpSymb$ only applies to 1-bounded expectations only.
\end{proof}

\subsection{Proof of \Cref{thm:convUnfolding}}
\label{proof:convUnfolding}
\convUnfolding*
\begin{proof}
    The proof is by induction on the structure of $\prog$.
    We focus on $\wpSymb$, the proof for $\wlpSymb$ is analogous.
    For the bases cases $\prog = \SKIP$, $\prog = \DIVERGE$, $\prog = \ASSIGN{\pVar}{\aExp}$, $\prog = \UNIFASSIGN{\pVar}$ and $\prog = \OBSERVE{\guard}$ there is nothing to show as these atomic programs are unaffected by unfolding.
    We now consider the composite cases.
    Let $\ex \in \expsmeas$ be arbitrary but fixed throughout the following.
    
    \inductionCase{$\prog = \ITE{\guard}{\prog_1}{\prog_2}$ and $\prog = \PCHOICE{\prog_1}{\prob}{\prog_2}$}
    These cases follow straightforwardly by induction and basic properties of $\wpSymb$.
    
    \inductionCase{$\prog = \SEQ{\prog_1}{\prog_2}$}
    To show that for all $\depth \in \nats$ it holds that $\wp{\unfold{\prog}{\depth}}{\ex} \eleq \wp{\unfold{\prog}{\depth+1}}{\ex}$ we argue as follows:
    \begin{align*}
        & \wp{(\SEQ{\prog_1}{\prog_2})^\depth}{\ex} \\
        \eeq & \wp{\prog_1^\depth}{\wp{\prog_2^\depth}{\ex}} \tag{definition of $\wpSymb$ and of $(\SEQ{\prog_1}{\prog_2})^\depth$} \\
        \eeleq & \wp{\prog_1^\depth}{\wp{\prog_2^{\depth + 1}}{\ex}} \tag{I.H. on $\prog_2$, monotonicity of $\wpSymb$}\\
        \eeleq & \wp{\prog_1^{\depth + 1}}{\wp{\prog_2^{\depth + 1}}{\ex}} \tag{I.H. on $\prog_1$} \\
        \eeq & \wp{(\SEQ{\prog_1}{\prog_2})^{\depth + 1}}{\ex} \tag{definition of $\wpSymb$ and of $(\SEQ{\prog_1}{\prog_2})^{\depth+1}$}
    \end{align*}
    To see that the supremum over all unrollings is equal to the exact $\wpSymb$ consider the following:
    \begin{align*}
        & \sup_{\depth \in \nats} \wp{(\SEQ{\prog_1}{\prog_2})^\depth}{\ex} \\
        \eeq & \sup_{\depth \in \nats} \wp{\prog_1^\depth}{\wp{\prog_2^\depth}{\ex}} \tag{definition of $\wpSymb$ and of $(\SEQ{\prog_1}{\prog_2})^\depth$} \\
        \eeq & \sup_{\depth_1 \in \nats} \sup_{\depth_2 \in \nats} \wp{\prog_1^{\depth_1}}{\wp{\prog_2^{\depth_2}}{\ex}} \tag{\Cref{thm:oneVsTwoSups}, monotonicity of $\wpTrans{\unfold{\prog_1}{\depth_1}}$, I.H.\ on $\prog_1$ and $\prog_2$} \\
        \eeq & \sup_{\depth_1 \in \nats} \wp{\prog_1^{\depth_1}}{\sup_{\depth_2 \in \nats} \wp{\prog_2^{\depth_2}}{\ex}} \tag{$\omega$-continuity of $\wpTrans{\unfold{\prog_1}{\depth_1}}$, see \Cref{thm:wpWlpWellDefinedAndContinuous}} \\
        \eeq & \sup_{\depth_1 \in \nats} \wp{\prog_1^{\depth_1}}{\wp{\prog_2}{\ex}} \tag{I.H.\ on $\prog_2$} \\
        \eeq & \wp{\prog_1}{\wp{\prog_2}{\ex}} \tag{I.H.\ on $\prog_1$} \\
        \eeq & \wp{\SEQ{\prog_1}{\prog_2}}{\ex} \tag{definition of $\wpSymb$} 
    \end{align*}
    
    \inductionCase{$\prog = \WHILE{\guard}{\progBody}$}
    We first show by (an inner) induction on $\depth \in \nats$ that
    \[
    \wp{\unfold{(\WHILE{\guard}{\progBody})}{\depth}}{\ex}
    \eeleq
    \wp{\unfold{(\WHILE{\guard}{\progBody})}{\depth + 1}}{\ex}
    ~.
    \]
    For $\depth = 0$ this follows as $\wp{\unfold{(\WHILE{\guard}{\progBody})}{0}}{\ex} = \wp{\DIVERGE}{\ex} = 0$ by definition.
    For $\depth > 0$:
    \begin{align*}
        & \wp{\unfold{(\WHILE{\guard}{\progBody})}{\depth}}{\ex} \\
        \eeq & \wp{\ITE{\guard}{\SEQ{\unfold{\progBody}{\depth-1}}{\unfold{(\WHILE{\guard}{\progBody})}}{\depth-1}}{\SKIP}}{\ex} \tag{definition of unfolding} \\
        \eeq & \iv{\guard} \cdot \wp{\unfold{\progBody}{\depth-1}}{\wp{\unfold{(\WHILE{\guard}{\progBody})}{\depth-1}}{\ex}} + \iv{\neg\guard} \cdot \wp{\SKIP}{\ex} \tag{definition of $\wpSymb$} \\
        \eeleq & \iv{\guard} \cdot \wp{\unfold{\progBody}{\depth-1}}{\wp{\unfold{(\WHILE{\guard}{\progBody})}{\depth}}{\ex}} + \iv{\neg\guard} \cdot \wp{\SKIP}{\ex} \tag{inner I.H., monotonicity of $\wpSymb$} \\
        \eeleq & \iv{\guard} \cdot \wp{\unfold{\progBody}{\depth}}{\wp{\unfold{(\WHILE{\guard}{\progBody})}{\depth}}{\ex}} + \iv{\neg\guard} \cdot \wp{\SKIP}{\ex} \tag{outer I.H. on $\progBody$} \\
        \eeq & \wp{\ITE{\guard}{\SEQ{\unfold{\progBody}{\depth}}{\unfold{(\WHILE{\guard}{\progBody})}}{\depth}}{\SKIP}}{\ex} \tag{definition of $\wpSymb$} \\
        \eeq & \wp{\unfold{(\WHILE{\guard}{\progBody})}{\depth+1}}{\ex} \tag{definition of unfolding}
    \end{align*}
    To see that the supremum over all unrollings is equal to the exact $\wpSymb$ consider the following:
    \begin{align*}
        & \sup_{\depth \in \nats} \wp{\unfold{(\WHILE{\guard}{\progBody})}{\depth}}{\ex} \\
        \eeq & \sup_{\depth \in \nats} \wp{\unfold{(\WHILE{\guard}{\progBody})}{\depth+1}}{\ex} \tag{non-decreasing, see above} \\
        \eeq & \sup_{\depth \in \nats} \wp{\ITE{\guard}{\SEQ{\unfold{\progBody}{\depth}}{\unfold{(\WHILE{\guard}{\progBody})}}{\depth}}{\SKIP}}{\ex} \tag{definition of unfolding} \\
        \eeq & \sup_{\depth \in \nats} \left( \iv{\guard} \cdot \wp{\unfold{\progBody}{\depth}}{\wp{\unfold{(\WHILE{\guard}{\progBody})}{\depth}}{\ex}} + \iv{\neg\guard} \cdot \ex \right) \tag{definition of $\wpSymb$} \\
        \eeq & \iv{\guard} \cdot \sup_{\depth \in \nats}  \wp{\unfold{\progBody}{\depth}}{\wp{\unfold{(\WHILE{\guard}{\progBody})}{\depth}}{\ex}} + \iv{\neg\guard} \cdot \ex \tag{$\omega$-continuity of $\cdot$ and $+$} \\
        \eeq & \iv{\guard} \cdot \sup_{\depth_2 \in \nats} \sup_{\depth_1 \in \nats} \wp{\unfold{\progBody}{\depth_1}}{\wp{\unfold{(\WHILE{\guard}{\progBody})}{\depth_2}}{\ex}} + \iv{\neg\guard} \cdot \ex \tag{\Cref{thm:oneVsTwoSups}} \\
        \eeq & \iv{\guard} \cdot \sup_{\depth_2 \in \nats}  \wp{\progBody}{\wp{\unfold{(\WHILE{\guard}{\progBody})}{\depth_2}}{\ex}} + \iv{\neg\guard} \cdot \ex \tag{I.H. on $\progBody$} \\
        \eeq & \iv{\guard} \cdot  \wp{\progBody}{\sup_{\depth_2 \in \nats} \wp{\unfold{(\WHILE{\guard}{\progBody})}{\depth_2}}{\ex}} + \iv{\neg\guard} \cdot \ex \tag{$\omega$-continuity of $\wpSymb$}
    \end{align*}
    By the definition of $\wpTrans{\WHILE{\guard}{\progBody}}$ in terms of a least fixed point, the above equation implies
    \[
    \wp{\WHILE{\guard}{\progBody}}{\ex}
    \eeleq
    \sup_{\depth \in \nats} \wp{\unfold{(\WHILE{\guard}{\progBody})}{\depth}}{\ex}
    ~. 
    \]
    
    To complete the proof we show by (an inner) induction on $\depth \in \nats$ that
    \[
    \wp{\unfold{(\WHILE{\guard}{\progBody})}{\depth}}{\ex}
    \eeleq
    \wp{\WHILE{\guard}{\progBody}}{\ex}
    ~.
    \]
    This implies that $\sup_{\depth \in \nats} \wp{\unfold{(\WHILE{\guard}{\progBody})}{\depth}}{\ex} \eleq \wp{\WHILE{\guard}{\progBody}}{\ex}$.
    For $\depth = 0$ this holds as $\wp{\unfold{(\WHILE{\guard}{\progBody})}{0}}{\ex} = \wp{\DIVERGE}{\ex} = 0$ by definition.
    For $\depth \geq 0$ the argument is as follows:
    \begin{align*}
        & \wp{\unfold{(\WHILE{\guard}{\progBody})}{\depth+1}}{\ex}  \\
        \eeq &  \wp{\ITE{\guard}{\SEQ{\unfold{\progBody}{\depth}}{\unfold{(\WHILE{\guard}{\progBody})}}{\depth}}{\SKIP}}{\ex} \tag{definition of unfolding} \\
        \eeq & \iv{\guard} \cdot  \wp{\unfold{\progBody}{\depth}}{\wp{\unfold{(\WHILE{\guard}{\progBody})}{\depth}}{\ex}} + \iv{\neg\guard} \cdot \ex  \tag{definition of $\wpSymb$} \\
        \eeleq & \iv{\guard} \cdot \wp{\unfold{\progBody}{\depth}}{\wp{\WHILE{\guard}{\progBody}}{\ex}} + \iv{\neg\guard} \cdot \ex  \tag{inner I.H., monotonicity of $\wpTrans{\unfold{\progBody}{\depth}}$} \\
        \eeleq & \iv{\guard} \cdot \wp{\progBody}{\wp{\WHILE{\guard}{\progBody}}{\ex}} + \iv{\neg\guard} \cdot \ex  \tag{outer I.H. on $\progBody$} \\
        \eeq & \wp{\WHILE{\guard}{\progBody}}{\ex} \tag{definition of $\wpSymb$}
    \end{align*}
\end{proof}

\subsection{Proof of \Cref{thm:pointwiseConv}}
\label{proof:pointwiseConv}
\pointwiseConv*
\begin{proof}
    We only show the equality involving $\wpSymb$ (the one for $\wlpSymb$ is analogous).    
    For all $n \in \nats$ we have
    \begin{align*}
        & \lwp{n}{\unfold{\prog}{n}}{\ex} \\
        \lleq & \wp{\unfold{\prog}{n}}{\ex} \tag{by \Cref{thm:soundness}} \\
        \lleq & \wp{\prog}{\ex} \tag{by \Cref{thm:convUnfolding}}
    \end{align*}
    and thus $\sup_{n \in \nats} \lwp{n}{\unfold{\prog}{n}}{\ex} \leq \wp{\prog}{\ex}$.
    But we also have
    \begin{align*}
        & \sup_{n \in \nats} \lwp{n}{\unfold{\prog}{n}}{\ex} \\
        \ggeq & \sup_{n \in \nats} \lwp{2^n}{\unfold{\prog}{2^n}}{\ex} \tag{elementary property of supremum} \\
        \eeq & \sup_{\depth \in \nats} \sup_{\N \in \nats} \lwp{2^\N}{\unfold{\prog}{2^\depth}}{\ex} \tag{by \Cref{thm:oneVsTwoSups} and \Cref{thm:powerTwoMonotonic}} \\
        \eeq & \sup_{\depth \in \nats} \wp{\unfold{\prog}{2^\depth}}{\ex} \tag{by \Cref{thm:convLoopFree}, using that $\unfold{\prog}{2^\depth}$ is loop-free and that $f \in \expsClass$} \\
        \eeq & \wp{\prog}{\ex} \tag{by \Cref{thm:convUnfolding}} ~.
    \end{align*}
\end{proof}

\subsection{Proof of \Cref{thm:expWellLocallyBounded}}
\label{proof:expWellLocallyBounded}
\expWellLocallyBounded*
\begin{proof}
    We show, by induction on the structure of $\synEx$, that $\sem{\synEx}$ is locally bounded when viewed as a function of type $\sem{\synEx} \colon \reals^{\lvarsSubset} \to \reals$.
    
    \inductionCase{$\synEx = \synTerm$}
    $\sem{\synTerm}$ is a piece-wise defined polynomial in finitely many variables, hence $\sem{\synTerm}$ is locally bounded.
    
    \inductionCase{$\synEx = \iv{\synGuard}\cdot \synEx$}
    It is $\sem{\iv{\synGuard}\cdot \synEx} \leq \sem{\synEx}$ and $\sem{\synEx}$ is locally bounded by the I.H.
    Hence $\sem{\iv{\synGuard}\cdot \synEx}$ is locally bounded, too.
    
    \inductionCase{$\synEx = \ratConst \cdot \synExb$}
    By the I.H., $\synExb$ is locally bounded, i.e., for every $\compactSet \subseteq \reals^\lvarsSubset = \reals^\lvarsSubset$ we have $\sup_{\st \in \compactSet} |\sem{\synExb}| < \infty$.
    But then we also have $\sup_{\st \in \compactSet} | \sem{\ratConst \cdot \synExb}| = \sup_{\st \in \compactSet} |\ratConst \cdot \sem{\synExb}|  = \sup_{\st \in \compactSet} |\ratConst|  \cdot | \sem{\synExb}| =  \sup_{\st \in \compactSet} \ratConst \cdot  | \sem{\synExb}| = \ratConst \cdot   \sup_{\st \in \compactSet} | \sem{\synExb}| < \infty$.
    
    \inductionCase{$\synEx = \synExb + \synExc$}
    Similar to the previous case, using I.H.\ for both $\synExb$ and $\synExc$.
    
    \inductionCase{$\synEx = \synSupBd{\lvar}{\ivalL}{\ivalR} \synExb$}
    Let $\compactSet \subseteq \reals^\lvarsSubset$ be compact.
    Let $\compactSet'$ be the (compact) cross product of $\compactSet$ and $\clIvalGen$, i.e., $\compactSet' = \{\pStUpdate{\st}{\lvar}{\xi} \in \reals^{\lvarsSubset \cup \{\lvar\}} \mid \st \in \compactSet, \xi \in \clIvalGen\}$.
    We have to show that $\sup_{\st \in \compactSet} |\sem{f}(\st)| < \infty$.
    We argue as follows:
    \begin{align*}
        & \sup_{\st \in \compactSet} \left|\sem{f}(\st) \right| \\ 
        \eeq & \sup_{\st \in \compactSet} \left| \sem{\synSupBd{\lvar}{\ivalL}{\ivalR} \synExb}(\st)  \right| \\
        \eeq & \sup_{\st \in \compactSet} \left| \sup_{\xi \in \clIvalGen}\sem{\synExb}(\pStUpdate{\st}{\lvar}{\xi})  \right| \tag{definition of $\sem{\synSupBd{\lvar}{\ivalL}{\ivalR} \synExb}$} \\
        \lleq & \sup_{\st \in \compactSet}  \sup_{\xi \in \clIvalGen} \left|\sem{\synExb}(\pStUpdate{\st}{\lvar}{\xi})  \right| \tag{absolute value of $\sup$ $\leq$ $\sup$ of absolute value} \\
        \lleq & \sup_{\st \in \compactSet'}  \left|\sem{\synExb}(\st)  \right| \tag{because $\compactSet' = \{\pStUpdate{\st}{\lvar}{\xi} \mid \st \in \compactSet, \xi \in \clIvalGen \}$} \\
        \llt & \infty \tag{$\sem{\synExb}$ locally bounded by I.H.} \\
    \end{align*}
    
    \inductionCase{$\synEx = \synInfBd{\lvar}{\ivalL}{\ivalR} \synExb$}
    This case is analogous to the $\supQuantifierSymbol$-case noticing that, for every set $A \subseteq \reals$, we have $|\inf A| = | - \sup - A| = |\sup - A| \leq \sup |-A| = \sup |A|$, where $- A = \{-a \mid a \in A\}$.
    
\end{proof}

\subsection{Proof of \Cref{thm:expToFo}}
\label{proof:expToFo}
\expToFo*
\begin{proof}
    By induction on the structure of $\synEx$.
    The crucial observation is that infima and suprema can be expressed in FO.
    In all of the following we assume that $\lvarb \notin \free{\synEx} \subseteq \{\lvar_1,\ldots,\lvar_n\}$. 

    \inductionCase{$\synEx(\lvar_1,\ldots,\lvar_n) = \synTerm(\lvar_1,\ldots,\lvar_n)$}
    We first define $\foForm_{\synTerm}(\lvar_1,\ldots,\lvar_n,\lvarb)$ by induction on the structure of the syntactic term $\synTerm$.
    
    \begin{itemize}
        \item $\synTerm = \ratConst$.
        We define
        \[
            \foForm_{\synTerm}(\lvarb)
            \quad\text{as}\quad
            \ratConst = \lvarb
        ~.
        \]
        \item $\synTerm = \lvar$.
        We define
        \[
            \foForm_{\synTerm}(\lvar,\lvarb)
            \quad\text{as}\quad
            \lvar = \lvarb
            ~.
        \]
        \item $\synTerm(\lvar_1,\ldots,\lvar_n) = \synTermb(\lvar_1,\ldots,\lvar_n) + \synTermc(\lvar_1,\ldots,\lvar_n)$.
        We define
        \[
            \foForm_\synTerm(\lvar_1,\ldots,\lvar_n,\lvarb)
            \quad\text{as}\quad
            \foExists \lvarb_1 \colon \foExists \lvarb_2 \colon \lvarb = \lvarb_1 + \lvarb_2 \land \foForm_\synTermb(\lvar_1,\ldots,\lvar_n,\lvarb_1) \land \foForm_\synTermc(\lvar_1,\ldots,\lvar_n,\lvarb_2)
            ~.
        \]
        \item $\synTerm(\lvar_1,\ldots,\lvar_n) = \synTermb(\lvar_1,\ldots,\lvar_n) \cdot \synTermc(\lvar_1,\ldots,\lvar_n)$. Similar to the previous case.
        \item $\synTerm(\lvar_1,\ldots,\lvar_n) = \synTermb(\lvar_1,\ldots,\lvar_n) \synMonus \synTermc(\lvar_1,\ldots,\lvar_n)$.
        This case is the most interesting one due to monus.
        We define $\foForm_\synTerm(\lvar_1,\ldots,\lvar_n,\lvarb)$ as 
        \begin{align*}
            \foExists \lvarb_1 \colon \foExists \lvarb_2 \colon &\foForm_\synTermb(\lvar_1,\ldots,\lvar_n,\lvarb_1) \land \foForm_\synTermc(\lvar_1,\ldots,\lvar_n,\lvarb_2) \\
            &\land
            \left( \lvarb_1 \geq \lvarb_2 \limplies \lvarb = \lvarb_1 - \lvarb_2 \right )
            \land
            \left( \lvarb_1 < \lvarb_2 \limplies \lvarb = 0 \right )
            ~.
        \end{align*}
    \end{itemize}
    Now that we have defined $\foForm_{\synTerm}$, we define
    \[
        \foForm_{\synEx}(\lvar_1,\ldots,\lvar_n,\lvarb)
        \quad\text{as}\quad
        \synTerm(\lvar_1,\ldots,\lvar_n) = \lvarb
        ~.
    \]
    
    \inductionCase{$\synEx(\lvar_1,\ldots,\lvar_n) = \iv{\synGuard(\lvar_1,\ldots,\lvar_n)}\cdot \synExb(\lvar_1,\ldots,\lvar_n)$}
    \newcommand{\lvarbool}{b}
    
    We first define an FO formula $\foForm_{\synGuard}(\lvar_1,\ldots,\lvar_n,\lvarbool)$ by induction on the structure of the syntactic guard $\synGuard$.
    The idea is to encode the (Boolean) result that $\synGuard$ evaluates to as a $\{0,1\}$-valued FO variable $\lvarbool$.
    
    \begin{itemize}
        \item $\synGuard(\lvar_1,\ldots,\lvar_n) = \synTerm(\lvar_1,\ldots,\lvar_n) < \synTermb(\lvar_1,\ldots,\lvar_n)$.
        We consider $\foForm_\synTerm(\lvar_1,\ldots,\lvar_n,\lvarb)$ and $\foForm_\synTermb(\lvar_1,\ldots,\lvar_n,\lvarb)$ as defined above and define $\foForm_{\synGuard}(\lvar_1,\ldots,\lvar_n,\lvarbool)$ as
        \begin{align*}
            \foExists \lvarb_1 \colon \foExists \lvarb_2 \colon &\foForm_\synTermb(\lvar_1,\ldots,\lvar_n,\lvarb_1) \land \foForm_\synTermc(\lvar_1,\ldots,\lvar_n,\lvarb_2) \\
            &\land
            \left( \lvarb_1 < \lvarb_2 \limplies \lvarbool = 1 \right )
            \land
            \left( \lvarb_1 \geq \lvarb_2 \limplies \lvarbool = 0 \right )
            ~.
        \end{align*}
        \item $\synGuard(\lvar_1,\ldots,\lvar_n) = \synGuard_1(\lvar_1,\ldots,\lvar_n) \land \synGuard_2(\lvar_1,\ldots,\lvar_n)$.
        We consider $\foForm_{\synGuard_1}(\lvar_1,\ldots,\lvar_n,\lvarbool)$ and $\foForm_{\synGuard_2}(\lvar_1,\ldots,\lvar_n,\lvarbool)$ and define $\foForm_{\synGuard}(\lvar_1,\ldots,\lvar_n,\lvarbool)$ as
        \begin{align*}
            \foExists \lvarbool_1 \colon \foExists \lvarbool_2 \colon &\foForm_{\synGuard_1}(\lvar_1,\ldots,\lvar_n,\lvarbool_1) \land \foForm_{\synGuard_2}(\lvar_1,\ldots,\lvar_n,\lvarbool_2) \\
            &\land
            \left( \lvarbool_1 + \lvarbool_1 = 2 \limplies \lvarbool = 1 \right )
            \land
            \left( \lvarbool_1 + \lvarbool_1 \neq 2 \limplies \lvarbool = 0 \right )
            ~.
        \end{align*}
        \item $\synGuard(\lvar_1,\ldots,\lvar_n) = \neg \synGuard'(\lvar_1,\ldots,\lvar_n)$.
        We consider $\foForm_{\synGuard'}(\lvar_1,\ldots,\lvar_n,\lvarbool)$ and define $\foForm_{\synGuard}(\lvar_1,\ldots,\lvar_n,\lvarbool)$ as
        \begin{align*}
            \foExists \lvarbool'\colon \foForm_{\synGuard'}(\lvar_1,\ldots,\lvar_n,\lvarbool')  
            \land
            \lvarbool = 1 - \lvarbool'
            ~.
        \end{align*}
    \end{itemize}
    
    Now that we have defined $\foForm_{\synGuard}$, we define
    \[
        \foForm_\synEx(\lvar_1,\ldots,\lvar_n,\lvarb)
        \quad\text{as}\quad
        \foExists \lvarbool \colon \foForm_\synGuard(\lvar_1,\ldots,\lvar_n,\lvarbool) \land (\lvarbool = 1 \limplies \foForm_\synEx(\lvar_1,\ldots,\lvar_n,\lvarb) )
        \land (\lvarbool = 0 \limplies \lvarb = 0 )
    \]
    
    \inductionCase{$\synEx(\lvar_1,\ldots,\lvar_n) = \ratConst \cdot \synExb(\lvar_1,\ldots,\lvar_n)$}
    We define
    \[
    \foForm_\synEx(\lvar_1,\ldots,\lvar_n,\lvarb)
    \quad\text{as}\quad
    \foExists \lvarb' \colon \lvarb = \ratConst \cdot \lvarb' \land \foForm_\synExb(\lvar_1,\ldots,\lvar_n,\lvarb')
    ~.
    \]
    
    \inductionCase{$\synEx(\lvar_1,\ldots,\lvar_n) = \synExb(\lvar_1,\ldots,\lvar_n) + \synExc(\lvar_1,\ldots,\lvar_n)$}
    We define
    \[
    \foForm_\synEx(\lvar_1,\ldots,\lvar_n,\lvarb)
    \quad\text{as}\quad
    \foExists \lvarb_1 \colon \foExists \lvarb_2 \colon \lvarb = \lvarb_1 + \lvarb_2 \land \foForm_\synExb(\lvar_1,\ldots,\lvar_n,\lvarb_1) \land \foForm_\synExc(\lvar_1,\ldots,\lvar_n,\lvarb_2)
    ~.
    \]
    
    \inductionCase{$\synEx = \synSupBd{\lvar}{\ivalL}{\ivalR} \synExb$}
    To give an intuition about how we handle this case, we recall the standard definition of the supremum:
    Let $A \subseteq \reals$ and $y \in \reals$.
    Then
    \begin{align*}
        & y = \sup A \\
        ~{}\iff{}~ & \text{for all } a\in A \colon a \leq y \tag{$y$ is an upper bound}\\
        & \text{and for all } y' \in \reals \colon( \text{for all } a \in A \colon a \leq y') \implies y \leq y' \tag{$y$ is not greater than any upper bound}
        ~.
    \end{align*}
    We now construct FO formulae corresponding to the above two conjuncts.
    Consider the formula $\foForm_\synExb(\lvar,\lvar_1,\ldots,\lvar_n,\lvarb)$ obtained by applying the I.H.\ to $\synExb(\lvar,\lvar_1,\ldots,\lvar_n)$.
    We define
    \[
    \foForm_{\synExb}^{ub}(\lvar_1,\ldots,\lvar_n,u)
    \quad\text{as}\quad
    \foForall \ivalL \leq \xi \leq \ivalR \colon \foForall \lvarb \colon \foForm_\synExb(\xi,\lvar_1,\ldots,\lvar_n,\lvarb) 	\limplies \lvarb \leq u
    ~.
    \]
    We have for all $\st \in \reals^{\free{\synExb}}$ and $\resVal \in \reals$ that $\st ,\resVal \models \foForm_{\synExb}^{ub}$ iff for all $\xi \in \clIvalGen \colon \sem{\synExb}(\pStUpdate{\st}{\lvar}{\xi}) \leq \resVal$.
    Now we define
    \[
    \foForm_\synEx(\lvar_1,\ldots,\lvar_n,\lvarb)
    \quad\text{as}\quad
    \foForm_{\synExb}^{ub}(\lvar_1,\ldots,\lvar_n,\lvarb) \land \foForall \lvarb ' \colon \foForm_{\synExb}^{ub}(\lvar_1,\ldots,\lvar_n,\lvarb') \limplies \lvarb \leq \lvarb'
    \]
    Then $\foForm_\synEx $ has the desired property.
    
    \inductionCase{$\synEx = \synInfBd{\lvar}{\ivalL}{\ivalR} \synExb$}
    Analogous to the previous case.
    \medskip
    
    Finally, observe that for \emph{quantifier-free} $\synEx$ (i.e., without $\supQuantifierSymbol$ and $\infQuantifierSymbol$), the above inductive construction of $\foForm_\synEx$ can be done in linear time in the size of $\synEx$.
    The so-obtained $\foForm_\synEx$ is not yet necessarily in existential prenex form, but it can be translated to the latter (again in linear time) by succesively applying the following standard prenexing rules which are true for all FO formulae $\foForm$ and $\foFormb$:
    \begin{align*}
        \foForm \limplies (\foExists \lvar \colon \foFormb)
        &\text{ is equivalent to }
        \foExists \lvar \colon (\foForm \limplies \foFormb)
        \\
        \foForm \land (\foExists \lvar \colon \foFormb)
        &\text{ is equivalent to }
        \foExists \lvar \colon (\foForm \land \foFormb)
    \end{align*}
    This completes the proof.
\end{proof}

\subsection{Proof of \Cref{thm:expBoundedAndRiemannIntegrable}}
\label{proof:expBoundedAndRiemannIntegrable}
\expBoundedAndRiemannIntegrable*
\begin{proof}
    Let $\pSt \in \nnReals^\lvarsSubset$ be arbitrary and
    consider the function $\ex_\pSt = \lam{\xi}{\iv{\xi \geq 0} \cdot \sem{\synEx}(\pStUpdate{\st}{\lvar}{\xi})}$ of type $\reals \to \reals$.
    By \Cref{thm:expWellLocallyBounded}, $\ex_\pSt$ is locally bounded.
    By \Cref{thm:expToFo}, $\ex_\pSt$ is semi-algebraic.
    It thus follows with \Cref{thm:semiAlgebraicFunctionAlmostContinuous} that $\ex_\pSt$ is Riemann-integrable on $\clIvalGen$.
    Hence, since $\st$ was arbitrary, we may conclude that $\sem{\synEx}$ is Riemann-integrable on $\clIvalGen$ w.r.t.\ $\lvar$.
\end{proof}

\subsection{Proof of \Cref{thm:checkEntailment}}
\label{proof:checkEntailment}
\checkEntailment*
\begin{proof}
    The fact that $\sem{\synEx} \eleq \sem{\synExb}$ iff the FO formula in \eqref{eq:smt-query} is valid follows directly from \Cref{thm:expToFo}.
    However, note that the formula in \eqref{eq:smt-query} is not necessarily quantifier-free as $\foForm_{\synEx'}$ and $\foForm_{\synExb'}$ may contain an $\foExists$-prefix (see \Cref{thm:expToFo}).
    By applying the standard prenexing rules
    \begin{align*}
        (\foExists \lvar \colon \foForm) \land \foFormb
        &\quad\text{is equivalent to}\quad
        (\foExists \lvar \colon \foForm \land \foFormb)
        \\
        \text{and}\qquad (\foExists \lvar \colon \foForm) \limplies \foFormb
        &\quad\text{is equivalent to}\quad
        \foForall \lvar \colon (\foForm \limplies \foFormb)
    \end{align*}
    which hold for all FO formulae $\foForm$ and $\foFormb$ with $\lvar \notin \free{\foFormb}$, we can transform the formula in \eqref{eq:smt-query} to universal prenex form.
    Checking validity of the latter is equivalent to checking validity of the \emph{quantifier-free} formula obtained by dropping the $\foForall$-prefix.
\end{proof}

\subsection{Substitution Lemma}

\begin{lemma}[{Substitution Lemma~\cite{DBLP:journals/pacmpl/BatzKKM21,DBLP:journals/pacmpl/SchroerBKKM23}}]
    \label{thm:syntacticSemanticSubs}
    Let $\synEx \in \synExps$ such that $\lvar \in \free{\synEx}$ and let $\synTerm \in \synTerms$ be a syntactic term.
    Consider the expression $\synSubs{\synEx}{\lvar}{\synTerm} \in \synExps$ obtained by (syntactically) substituting every free occurrence of $\lvar$ in $\synEx$ by $\synTerm$ in a capture-avoiding manner (i.e., by possibly renaming variables appropriately).
    Then
    \[
    \sem{\synSubs{\synEx}{\lvar}{\synTerm}}
    \eeq
    \exSubs{\sem{\synEx}}{\lvar}{\sem{\synTerm}}
    \gray{
        \eeq
        \lam{\pSt}{\sem{\synEx}(\pStUpdate{\pSt}{\lvar}{\sem{\synTerm}(\pSt)})}
    }
    ~.
    \]
\end{lemma}

\subsection{Proof of \Cref{thm:exprRiemannSuitable}}
\label{proof:exprRiemannSuitable}
\exprRiemannSuitable*
\begin{proof}
    Recall \Cref{def:riemannSuitable}.
    By \Cref{thm:expBoundedAndRiemannIntegrable}, all $\sem{\synEx}$ are bounded and Riemann integrable on $\uIval$ w.r.t.\ every $\pVar \in \pVars$.
    The arithmetic expressions representable by $\synTerms$, i.e., the functions $\sem{\synTerm}$ for $\synTerm \in \synTerms$, are locally bounded since they are piecewise defined polynomials.
    The closure properties from \Cref{def:riemannSuitable} all follow immediately from the definition of $\synExps$, except the substitution $\exSubs{\sem{\synEx}}{\pVar}{\aExp}$.
    Here, the crucial observation is that \emph{semantic substitution corresponds to syntactic substitution}, see \Cref{thm:syntacticSemanticSubs}.
\end{proof}

\subsection{Computing Riemann $\wpwlpSymb$ Syntactically}
\label{proof:computeSyntacticExpression}
\begin{restatable}[Computing Riemann $\wpwlpSymb$ Syntactically]{lemma}{computeSyntacticExpression}
	\label{thm:computeSyntacticExpression}
	For all $\prog \in \pWhileWith{\synTerms}{\synGuards}$, $\prog$ loop-free, and integers $\N \geq 1$ we can effectively do the following:
	For every given ...
	\begin{enumerate}
		\setlength\itemsep{0.2em}
		\item ... $\synEx \in \syntacticExpectations$, compute $\synExb \in \syntacticExpectations$ such that \qquad \,\,\,\, $\lwp{\N}{\prog}{\sem{\synEx}} = \sem{\synExb}$ ~.
		\item ... $\synEx \in \syntacticExpectations$, compute $\synExb \in \syntacticExpectations$ such that \qquad \,\,\,\, $\uwp{\N}{\prog}{\sem{\synEx}} = \sem{\synExb}$ ~.
		\item ... $\synEx \in \bsyntacticExpectations$, compute $\synExb \in \bsyntacticExpectations$ such that \quad $\lwlp{\N}{\prog}{\sem{\synEx}} = \sem{\synExb}$ ~.
		\item ...  $\synEx \in \bsyntacticExpectations$, compute $\synExb \in \bsyntacticExpectations$ such that \quad $\uwlp{\N}{\prog}{\sem{\synEx}} = \sem{\synExb}$ ~.
	\end{enumerate}
	In (1) and (3), if $\synEx$ is $\supQuantifierSymbol$-free, then so is $\synExb$.
	In (2) and (4), if $\synEx$ is $\infQuantifierSymbol$-free, then so is $\synExb$.
\end{restatable}
\begin{proof}
    By induction on the structure of $\prog$.
    Most cases are immediate as we have essentially designed $\synExps$ such that $\repExps$ ($\boundedRepresentableExpectations$) is closed under $\lwpTrans{\N}{\prog}$ and $\uwpTrans{\N}{\prog}$ ($\lwlpTrans{\N}{\prog}$ and $\uwlpTrans{\N}{\prog}$, resp.) for loop-free $\prog$.
    For the assignment $\ASSIGN{\pVar}{\aExp}$ we rely on syntactic substitution (\Cref{thm:syntacticSemanticSubs}).
\end{proof}

\subsection{Proof of \Cref{thm:boundsOnLoopFreeWpWlp}}
\label{proof:boundsOnLoopFreeWpWlp}
\boundsOnLoopFreeWpWlp*
\begin{proof}
    By computing a syntactic representation of the corresponding Riemann $\wpwlpSymb$ as in \Cref{thm:computeSyntacticExpression}, translating to PNF via \Cref{thm:pnf}, and applying the quantitative entailment check from \Cref{thm:checkEntailment}.
    The \gray{gray} implications follow from soundness of the Riemann $\wpwlpSymb$ (\Cref{thm:soundness}).
\end{proof}

\subsection{Proof of \Cref{thm:verificationInvariants}}
\label{proof:verificationInvariants}
\verificationInvariants*
\begin{proof}
    Since $\progBody$ is loop-free we can apply \Cref{thm:computeSyntacticExpression} to compute $i \in \synExps$ (resp. $j \in \bsyntacticExpectations$) such that $\sem{i} = \uwp{\N}{\progBody}{\sem{I}}$ (resp. $\sem{j} = \lwlp{\N}{\progBody}{\sem{J}}$) with $i$ is $\infQuantifierSymbol$-free (resp $\supQuantifierSymbol$-free). Then, applying \Cref{thm:invariant_approx}, we need to check that 
    $ [\guard] \cdot \sem{\synEx} + [\neg \guard] \cdot \sem{i} \leq \sem{I}$ (resp. $\sem{J} \leq [\guard] \cdot \sem{\synEx} + [\neg \guard] \cdot \sem{j} $). By \Cref{thm:boundsOnLoopFreeWpWlp} this inequality can be reduced to $\QFNRA$.
\end{proof}

\subsection{Proof of \Cref{thm:complexity}}
\label{proof:complexity}
\complexity*
\begin{proof}
    We show membership in $\coRE$ first (which is the more interesting direction).
    Consider the complement of the first problem:
    \begin{align}
        \text{``Does there exist } \pSt \in \pStates \text{ s.t.\ } \sem{\synExb}(\pSt) < \wp{\prog}{\sem{\synEx}}(\pSt) \text{ ?''}
        \tag{complement of problem \eqref{it:problemUpperBoundOnWp}}
    \end{align}
    If the answer to the above question is \emph{yes}, then, since 
    $\sup_{n \geq 1}\lwp{n}{\unfold{\prog}{n}}{\sem{\synEx}} = \wp{\prog}{\sem{\synEx}}$
    by \Cref{thm:pointwiseConv}, there must exist $\pSt \in \pStates$ and $n \geq 1$ such that $\sem{\synExb}(\pSt) < \lwp{n}{\unfold{\prog}{n}}{\sem{\synEx}}(\pSt)$.
    Hence our semi-decision procedure will check the latter inequality for increasing values of $n$.
    To do so for a fixed $n$, the procedure constructs $\synExc$ such that $\sem{\synExc} = \lwp{n}{\unfold{\prog}{n}}{\sem{\synEx}}$ as in \Cref{thm:computeSyntacticExpression} and then invokes \Cref{thm:expToFo} to obtain FO formulae $\foForm_\synExb(\free{\synExb}, \lvarb_\synExb)$ and $\foForm_\synExc(\free{\synExc}, \lvarb_\synExc)$ encoding $\synExb$ and $\synExc$.
    It then only remains to check if the FO formula $(\foForm_\synExb \land \foForm_\synExc) \limplies \lvarb_\synExb < \lvarb_\synExc$ is satisfiable, which is decidable.
    
    To see that the complement of problem \eqref{it:problemUpperBoundOnWp} is $\RE$-hard we reduce the halting problem to it.
    We proceed as in \cite{DBLP:conf/mfcs/KaminskiK15}.
    The standard halting problem is equivalent to checking if $\wp{\prog_{std}}{1}(\pSt) = 1$ where $\prog_{std} \in \pWhileWith{\synTerms}{\synGuards}$ is a given standard program, i.e., it does not contain any probabilistic constructs, and $\pSt \in \rats^\pVars$ is a given initial state.
    Now consider \[
    \prog
    \eeq
    \PCHOICE{\SEQ{\ASSIGN{\pVar_1 }{\pSt(\pVar_1)} \,\SEMICOLON\, \ldots \,\SEMICOLON\,\ASSIGN{\pVar_n }{\pSt(\pVar_n)} }{\prog_{std}}}{\tfrac 1 2}{\SKIP} 
    \]
    where $\{\pVar_1,\ldots,\pVar_n\} = \pVars$.
    Then $\tfrac 1 2 < \wp{\prog}{1}$ iff $\wp{\prog_{std}}{1}(\pSt) > 0$ iff $\prog_{std}$ terminates on input $\pSt$.
    
    The argument for $\coRE$ membership of problem \eqref{it:problemLowerBoundOnWlp} is completely analogous.
    For $\coRE$-hardness we argue dually (we define $\prog$ as above):
    $\wlp{\prog}{0} < \tfrac 1 2$ iff $\wlp{\prog_{std}}{0}(\pSt) < \tfrac 1 2$ iff $\prog_{std}$ terminates on input $\pSt$.
    The latter argument works because for standard programs we have $\wlp{\prog_{std}}{0}(\pSt) = 0$ iff $\prog_{std}$ terminates on $\pSt$, and $\wlp{\prog_{std}}{0}(\pSt) = 1$ iff $\prog_{std}$ does not terminate on $\pSt$.
\end{proof}
\subsection{Inputs to \toolcaesar}
\label{app:caesar_inputs}

\paragraph{The Monte Carlo $\pi$-approximator}
\phantom{a}

\begin{lstlisting}[language=C++, basicstyle=\ttfamily\small]
	coproc monte_carlo_pi(M : UReal) 
			-> (x : UReal, y : UReal, count : UReal)
	pre 0.85*(M)     
	post count
	{
		var i : UReal = 1
		count = 0
		@invariant(count + ite(0 <= i && i<=M, 0.85*((M-i) + 1), 0))
		while i <= M {
			
			var N : UInt = 16; 
			var j : UInt = unif(0, 15); //discrete_uniform(16)
			
			// --- Nondeterministic assignment x := [...]
			cohavoc x; 
			coassume ?!(j / N <= x && x <= (j+1) / N)
			// ---
			
			
			j  = unif(0, 15); //discrete_uniform(16)
			
			// --- Nondeterministic assignment y := [...]
			cohavoc y; 
			coassume ?!(j / N <= y && y <= (j+1) / N)
			// ---
			
			if x*x + y*y <= 1 {
				count = count +1
			}else{}
			
			i = i+1
			
		}
		
	}
\end{lstlisting}

\paragraph{Irwin-Hall without conditioning}
\phantom{a}
\begin{lstlisting}[language=C++, basicstyle=\ttfamily\small]
	coproc irwin_hall(M : UReal) -> (x : UReal)
	pre 1.1*((M)/2)
	post x
	{
		x = 0
		var i : UReal = 1
		@invariant(ite(i <= M, (x + 1.1*((M-i) + 1)/2), x))
		while i <= M {
			
			var inc : UReal; var N : UInt = 10; 
			var j : UInt = unif(0, 9); //discrete_uniform(10)
			
			// --- Nondeterministic assignment inc := [...]
			cohavoc inc; 
			coassume ?!(j / N <= inc && inc <= (j+1) / N)
			// ---
			
			x = x + inc
			i = i + 1
		}
	}
\end{lstlisting}

\paragraph{Irwin-Hall with conditioning}
\phantom{a}

\begin{lstlisting}[language=C++, basicstyle=\ttfamily\small]
	domain Exponentials {
		func exp(exponent: UReal): EUReal
		
		axiom exp_base exp(0) == 1
		axiom exp_step forall exponent: UReal. 
			exp(exponent + 1) == 0.5 * exp(exponent)
		axiom exp_antitone forall exp1: UReal. forall exp2: UReal. 
			(exp1 <= exp2) ==> (exp(exp2) <= exp(exp1))
	}
	
	
	proc irwin_hall_conditioning_wlp(M : UReal) -> (x : UReal)
	pre exp(M)
	post 1
	{
		x = 0
		var i : UReal = 1
		@invariant(ite(i <= M, exp((M-i) + 1), 1))
		while i <= M {
			var inc : UReal; var N : UInt = 2; 
			var j : UInt = unif(0, 1); //discrete_uniform(2)
			
			// --- Nondeterministic assignment inc := [...] 
			havoc inc; 
			assume ?(j / N <= inc && inc <= (j+1) / N) 
			// ---
			
			assert ?(inc <= 1/2) // observe
			x = x + inc
			i = i + 1
		}
	}
	

	coproc irwin_hall_conditioning_wp(M : UReal) -> (x : UReal)
	pre 1.5*((M)/8)
	post x
	{
		x = 0
		var i : UReal = 1
		@invariant(ite(i <= M, (x + 1.5*((M-i) + 1)/8), x))
		while i <= M {
			var inc : UReal; var N : UInt = 20; 
			var j : UInt = unif(0, 19); //discrete_uniform(20)
			
			// --- Nondeterministic assignment inc := [...] 
			cohavoc inc; 
			coassume ?!(j / N <= inc && inc <= (j+1) / N)
			// ---
			
			assert ?(inc <= 1/2) // observe
			x = x + inc
			i = i + 1
		}
	}
\end{lstlisting}

\paragraph{Probably Diverging Loop}
\phantom{a}

\begin{lstlisting}[language=C++, basicstyle=\ttfamily\small]
	domain Exponentials {
		func exp(exponent: UReal): EUReal
		
		axiom exp_base exp(0) == 1
		axiom exp_step forall exponent: UReal. 
			exp(exponent + 1) == 0.5 * exp(exponent)
		axiom exp_antitone forall exp1: UReal. forall exp2: UReal. 
			(exp1 <= exp2) ==> (exp(exp2) <= exp(exp1))
	}
	
	
	proc diverging(x_init : UReal, a : UReal, b : UReal) -> ()
	pre [a<=b]*(1-exp(x_init))
	post 0
	{
		var x : UReal = x_init 
		
		@invariant([a<=b]*(1-exp(x)))
		while(x >= 0){
			var y : UReal
			var N : UInt = 2; 
			var i : UInt = unif(0, 1); //discrete_uniform(2)
			
			// --- Nondeterministic assignment y := [...] 
			havoc y; 
			assume ?(i / N <= y && y <= (i+1) / N)
			// ---
			
			y = (b-a)*y + a        
			
			if (y <= (a+b)/2) {
				// --- diverge
				assert 1
				assume ?(false)
				// ---
			}else {}
			
			x = x - 1
			
		}
		
	}
\end{lstlisting}

\paragraph{Race between Tortoise and Hare}
\phantom{a}

\begin{lstlisting}[language=C++, basicstyle=\ttfamily\small]
	coproc tortoise_hare(h_init : UReal, t_init : UReal) 
			-> (count : UReal)
	pre 1.5*((t_init - h_init) + 2)*2 
	post count
	{
		var h : UReal = h_init
		var t : UReal = t_init
		
		count = 0
		
		@invariant(ite(h<=t, count + 1.5*((t-h) + 2)*2, count))
		while h <= t {
			var choice : Bool = flip(0.5)
			if choice {
				
				var inc : UReal; 
				var N : UInt = 25; 
				var j : UInt = unif(0, 24); //discrete_uniform(25)
				
				// --- Nondeterministic assignment inv := [...] 
				cohavoc inc; 
				coassume ?!(j / N <= inc && inc <= (j+1) / N)
				// ---
				
				inc = (10-0)*inc + 0 
				h = h + inc
			} else {}
			
			t = t+1
			count = count + 1
		}
	}
\end{lstlisting}

}{}

\end{document}